\documentclass{article}

\usepackage[english]{babel}

\usepackage{dsfont}

\usepackage{amsmath,amsfonts,amssymb, amsthm} 

\theoremstyle{plain}
\newtheorem{definition}{Definition}
\newtheorem{assumption}{Assumption}
\newtheorem{theorem}{Theorem}
\newtheorem{proposition}{Proposition}
\newtheorem{corollary}{Corollary}
\newtheorem{lemma}{Lemma}
\newtheorem{conjecture}{Conjecture}

\usepackage[%
    font={small,sf},
    labelfont=bf,
    format=hang,    
    format=plain,
    margin=0pt,
]{caption}

\usepackage{wrapfig}
\usepackage{subcaption}
\usepackage{bbm} 
\usepackage[table,xcdraw, dvipsnames]{xcolor}

\usepackage{geometry}
 \usepackage[framemethod=TikZ]{mdframed}
\newcounter{takeaway}[section]

\ifstrempty{#1}%
{\mdfsetup{%
    frametitle={%
        \tikz[baseline=(current bounding box.east),outer sep=0pt]
        \node[anchor=east,rectangle,fill=blue!20]
        {\strut Takeaways};}
    }%
}{\mdfsetup{%
    frametitle={%
        \tikz[baseline=(current bounding box.east),outer sep=0pt]
        \node[anchor=east,rectangle,fill=blue!20]
        {\strut Takeaways};}%
    }%
}%
\mdfsetup{%
    innertopmargin=10pt,linecolor=blue!20,%
    linewidth=2pt,topline=true,%
    frametitleaboveskip=\dimexpr-\ht\strutbox\relax%
}

\usepackage[list=true]{subcaption}

\usepackage[colorinlistoftodos,textwidth=2.1cm]{todonotes}

\newcommand{\bc}[1]{}

\newcommand{\indep}{\perp \!\!\! \perp}

\usepackage{comment}

\usepackage{natbib}
\usepackage{here}

\usepackage{graphicx}
\usepackage[colorlinks=true, allcolors=blue]{hyperref}

\usepackage{cleveref}

\usepackage{tikz}
\usetikzlibrary{backgrounds}
\usetikzlibrary{arrows,shapes}
\usetikzlibrary{tikzmark}
\usetikzlibrary{calc}

\usepackage{amssymb}
\usepackage{mathtools, nccmath}

\usepackage{xspace}

\usepackage{array}
\usepackage{ragged2e}
\newcolumntype{P}[1]{>{\RaggedRight\hspace{0pt}}p{#1}}
\newcolumntype{X}[1]{>{\RaggedRight\hspace*{0pt}}p{#1}}

\usepackage{tcolorbox}

\usetikzlibrary{arrows,shapes,positioning,shadows,trees,mindmap}
\usepackage[edges]{forest}
\usetikzlibrary{arrows.meta}
\colorlet{linecol}{black!75}
\usepackage{xkcdcolors} 

\usepackage{tikz}
\usetikzlibrary{backgrounds}
\usetikzlibrary{arrows,shapes}
\usetikzlibrary{tikzmark}
\usetikzlibrary{calc}
\newcommand{\highlight}[2]{\colorbox{#1!17}{$\displaystyle #2$}}

\colorlet{mhpurple}{Plum!80}

\renewcommand{\highlight}[2]{\colorbox{#1!17}{#2}}

\title{Reweighting the RCT for generalization: finite sample error and variable selection}
\author{B\'{e}n\'{e}dicte Colnet \thanks{Soda project-team,  Premedical project-team, INRIA (email: benedicte.colnet@inria.fr).}
  \and
Julie Josse\thanks{Premedical project team, INRIA Sophia-Antipolis, Montpellier, France.}
  \and
Ga\"{e}l Varoquaux\thanks{Soda project-team, INRIA Saclay, France.}
     \and 
Erwan Scornet \thanks{Centre de Math\'{e}mathiques Appliqu\'{e}es, UMR 7641, \'{E}cole polytechnique, CNRS, Institut Polytechnique de Paris, Palaiseau, France.}
}

\date{\today}

\begin{document}
\maketitle

\begin{abstract}
Randomized Controlled Trials (RCTs) may suffer from limited scope. In particular, samples may be unrepresentative:
some RCTs over- or under- sample individuals with certain characteristics compared to the target population, for which one wants conclusions on  treatment effectiveness.
Re-weighting trial individuals to match the target population can improve the treatment effect estimation. 
In this work, we establish the expressions of the bias and variance of such reweighting procedures - also called Inverse Propensity of Sampling Weighting (IPSW) - in presence of categorical covariates for any sample size. 
Such results allow us to compare the theoretical performance of different versions of IPSW estimates. Besides, our results show how the performance (bias, variance, and quadratic risk) of IPSW estimates depends on the two sample sizes (RCT and target population). A by-product of our work is the proof of consistency of IPSW estimates. Results also reveal that IPSW performances are improved when the trial probability to be treated is estimated (rather than using its oracle counterpart).
In addition, we analyze how including covariates that are not necessary for
identifiability of the causal effect may impact the asymptotic variance.
Including covariates that are shifted between the two samples but not
treatment effect modifiers increases the variance while non-shifted but
treatment effect modifiers do not.
We illustrate all the takeaways in a didactic example, and on a semi-synthetic simulation inspired from critical care medicine.
    \vspace{12pt}\\
\textit{Keywords:} Average treatment effect (ATE);  
Sampling bias; 
External validity; 
Transportability;
Distributional shift;
IPSW.
\end{abstract}


\section{Introduction}

\paragraph{Motivation}
Modern \textit{evidence-based} medicine puts Randomized Controlled Trial
(RCT) at the core of clinical evidence. Indeed, randomization enables to
estimate the average treatment effect  (called ATE) by avoiding
confounding effects of spurious or undesirable associated factors. 
But more recently, concerns have been raised on the limited scope of RCTs: stringent eligibility criteria, unrealistic real-world compliance, short timeframe, limited sample size, etc. All these possible limitations threaten the external validity of RCT studies to other situations or populations \citep{Rothwell2007ToWhom, Stuart2017ChapterBook, Deaton2018Misunderstanding}.
The usage of complementary non-randomized data, referred to as \textit{observational} or from the \textit{real world}, brings promises as additional sources of evidence, in particular combined to trials \citep{kallus2018removing, athey2020combining, Liu2021TrialPathFinder}. 
For example, assume policy makers are studying an RCT which comes with great promises about a new treatment. But when reading the report, they may discover that the RCT is composed of substancially younger people than the target population of interest. Such a situation can be uncovered from the so-called \textit{Table 1} of this newly published trial, which summarizes the demographics of the study population.
In case of treatment effect heterogeneities, e.g. if the younger individuals respond better to the treatment, the ATE estimated from the trial is over-estimated and then biased. 
Now, assume these policy makers have also at disposal a sample of the actual patients in the district, being a representative sample of the true distribution of age in this population (typically without information on the outcome or the treatment). 
Can they use this representative sample of the target population of interest to \textit{re-weight} or to \textit{generalize} the trial's findings?
The answer is \textit{yes}: the strategy has been formalized and popularized lately \citep{stuart2011use, pearl2011transportability, Bareinboim2012Controlling, bareinboim2012completeness, tipton2013improving, Muircheartaigh2014GeneralizingApproach, Hartman2015FromSATE, kern2016assessing, dahabreh2020extending} (reviewed in \cite{Colnet2020review, Degtiar2021Generalizability}) and can come under many variants named \textit{generalization}, \textit{transportability}, \textit{recoverability}, and \textit{data-fusion}.
In fact, the idea of re-weighting a trial can be traced back before the 2010's. Several epidemiology books had already presented the core idea under the name \textit{standardization} \citep{Rothman2000ModernEpidemiology, Rothman2011bookEpidemiologyIntrod}.

In this work, we focus on one estimator used to generalize RCTs: the \textit{Inverse Propensity of Sampling Weighting} (IPSW) \citep{cole2010generalizing, stuart2011use}, also named \textit{Inverse Odds of Sampling Weights} (IOSW) \citep{Westreich2017IOSW, josey2021transporting} or \textit{Inverse probability of participation weighting} (IPPW) \citep{Degtiar2021Generalizability}.
Despite an increasing literature on \textit{generalization}, important practical questions remain open \citep{kern2016assessing, tipton2016smallsample, stuart2017CaseStudyDifficulties, Ling2022CriticalReview}. 
For instance, which covariates -- for e.g. age, and others --  should be used to build the weights? Are some covariates increasing or lowering the overall precision? What is the impact of the size of the two samples (trial and representative sample) on the IPSW's properties? 

\paragraph{Outline}
We start by illustrating the principles of trial re-weighting and some key results of this article on a toy example (Section~\ref{sec:toy-example-generalization}). Section~\ref{sec:toy-example-generalization} ends with related works. Then Section~\ref{section:notations-and-assumptions} introduces the mathematical notations, assumptions, and the precise definition of the IPSW estimator. In particular, we present several versions of the IPSW estimator: whether the covariates probability of the trial or the target population are estimated from the data or assumed as an oracle.
This links our results to classic work in causal inference and epidemiology. Section~\ref{section:theoretical-results} contains all the theoretical results, such as finite sample bias, variance, bounds on the risk, consistency, and large sample variance.
We also detail why another version of the IPSW, where the probability of treatment assignment in the trial is also estimated, has a lower variance. 
Finally, we discuss in Section~\ref{section:theoretical-results} how additional and non-necessary covariates can either improve or damage variance, depending on their status: whether they are only shifted between the two populations or only treatment-effect modifiers.
Section~\ref{section:semi-synthetic-simulations} completes the toy example and illustrates all theoretical results on an extensive semi-synthetic example inspired from the medical domain.
Finally, Section~\ref{sec:conclusion-discussion} summarizes all practical takeaways for this research and discusses it.

\section{Problem setting}\label{sec:toy-example-generalization}

\subsection{Toy example}

\subsubsection{Context and intuitive estimation strategy}
\label{subsection_toy_example}
\begin{wrapfigure}{r}{0.35\textwidth}
\centerline{\includegraphics[width=0.35\textwidth]{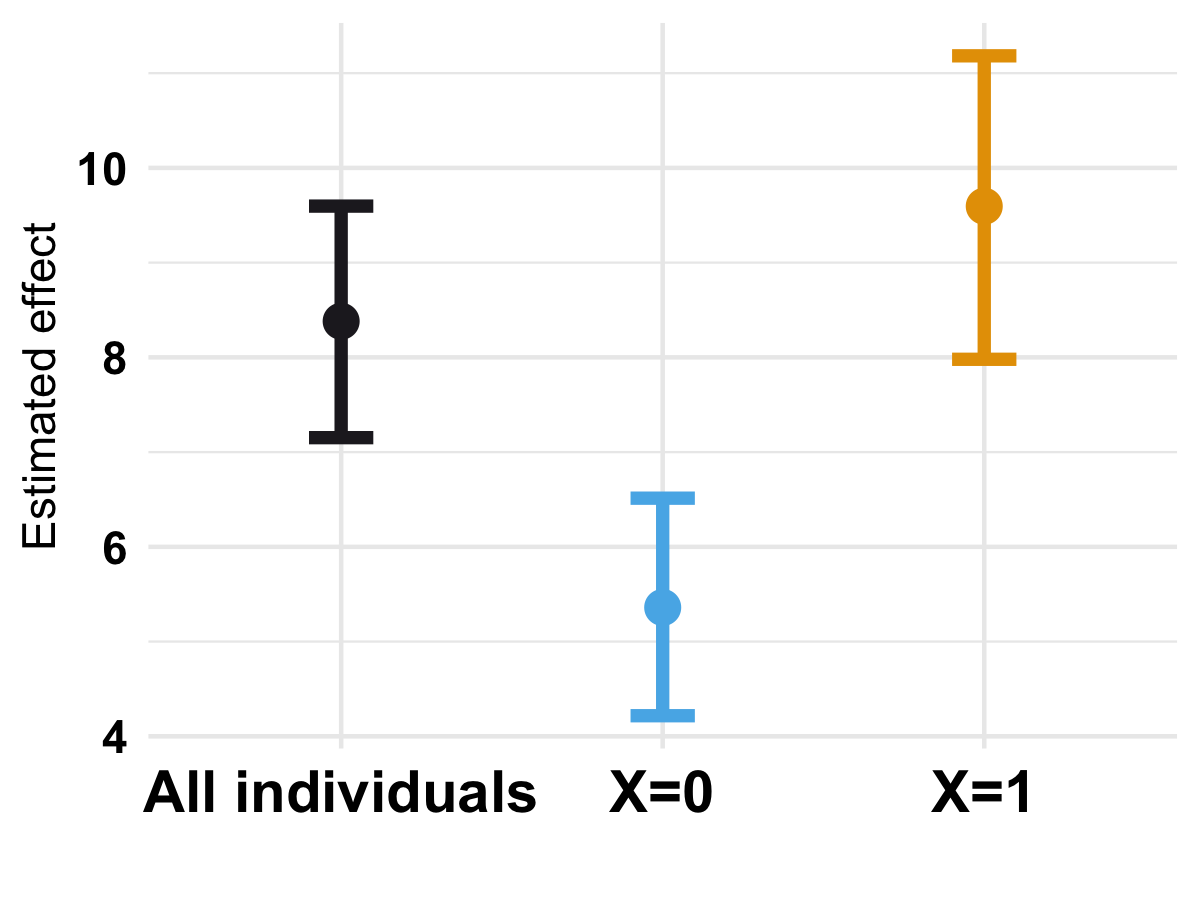}}
\caption{\textbf{Treatment effect estimates} (absolute difference) measured on a simulated trial of size $n=150$ sampled according to the trial population $\mathcal{P}_\text{\tiny R}$. On the left the estimate on \textbf{all individuals}, and on the right the two estimate stratified (\textcolor{cyan}{\textbf{$X=0$}} and \textcolor{orange}{\textbf{$X=1$}}) showing treatment effect heterogeneities along the genetic mutation $X$.}
\label{fig:toy_example_stratification}
\end{wrapfigure}
Assume that we would like to measure the average effect of a treatment (ATE) $A$ on a outcome $Y$ in a target population of interest $\mathcal{P}_\text{\tiny T}$ (for \textbf{t}arget),
and that an existing Randomized Controlled Trial (RCT) had already been conducted on $n=150$ individuals, sampled from a population $\mathcal{P}_\text{\tiny R}$ (for \textbf{r}andomized), to assess the average effect of $A$ on $Y$. A popular estimator, with interesting theoretical properties, to estimate the average treatment effect in a trial is the Horvitz-Thomson estimator \citep{HorvitzThompson1952seminal},
\begin{align}\label{eq:toy-estimator-HT}
    \hat \tau_{\text{\tiny HT},n} &=   \frac{1}{n} \sum_{i \in \text{Trial}}\left(\frac{ Y_i A_i}{\pi} - \frac{ Y_i (1-A_i)}{1-\pi} \right),
\end{align}
where $\pi$ is the probability of treatment allocation in the trial (in most applications, $\pi = 0.5$). Figure~\ref{fig:toy_example_stratification} presents results of a simulated trial with an average treatment effect around $8.2$. In addition, assume that the trial provides evidence that the treatment effect is heterogeneous with respect to a certain genetic mutation denoted $X$ (with $X=1$ for the mutation, and $X=0$ if no mutation). More specifically, the average treatment effect conditional to $X$ is larger for individuals with $X = 1$ than for those with $X = 0$. This situation is illustrated on Figure~\ref{fig:toy_example_stratification} where the average effect per strata $X$ is also represented. We have at hand a representative sample of $m=1000$ individuals from the target population we are interest in (for example from an existing observational database). We observe that individuals with the genetic mutation ($X=1$) are over-represented in the trial compared to the target population of interest (see Figure~\ref{fig:toy_example_covariate_shift}). 
As a consequence, the trial overestimates the target population's ATE we are interested in.

\begin{figure}[!h]
	\begin{minipage}{.55\linewidth}
      \caption{\textbf{Covariate shift along the genetic mutation $X$} between the trial population $\mathcal{P}_\text{\tiny R}$ and target population $\mathcal{P}_\text{\tiny T}$, highlighting the distributional shift between the two data sources. Such population's difference questions what is named the \underline{external validity of a trial}.}
      \label{fig:toy_example_covariate_shift}
    \end{minipage}
    \begin{minipage}{.45\linewidth}
    \begin{center}
    {\small
	\begin{tabular}{|l|l|l|}
\hline
\cline{2-3}
                                                             & \cellcolor[HTML]{CBCEFB}\textbf{Target ($\mathcal{P}_\text{\tiny T}$)} & \cellcolor[HTML]{CBCEFB}\textbf{Trial ($\mathcal{P}_\text{\tiny R}$)} \\ \hline
\multicolumn{1}{|l|}{\cellcolor[HTML]{ECF4FF}\textbf{$X = 1$}} & 30\%                                    & 75\%                                   \\ \hline
\multicolumn{1}{|l|}{\cellcolor[HTML]{ECF4FF}\textbf{$X = 0$}} & 70\%                                    & 25\%                                   \\ \hline
\end{tabular}}
    \end{center}
	\end{minipage}
	
\end{figure}

Fortunately, the representative sample of the target population can be used to learn weights, and re-weight the trial data in the following way,
\begin{equation}\label{eq:toy-estimator-oracle}
    \hat \tau_{n,m} =   \frac{1}{n} \sum_{i \in \text{Trial}} \underbrace{\hat w_{n,m}(X_i)}_\textrm{Weights} \underbrace{ \left(\frac{ Y_i A_i}{\pi} - \frac{ Y_i (1-A_i)}{1-\pi} \right)}_\textrm{Horvitz-Thomson}.
\end{equation}

As detailed later on, the weights $\hat w_{n,m}$ aims at estimating the probability ratio $\frac{p_\text{\tiny T}\left(x\right)}{p_\text{\tiny R}\left(x\right)}$, where $p_\text{\tiny T}\left(x\right)$ (resp. $p_\text{\tiny R}\left(x\right)$) is the probability of observing an individual with characteristics $X=x$ in the target (resp. randomized) population. The weights $\hat w_{n,m}$ depend on the sizes of the randomized and observational data sets, namely $n$ and $m$. 
Consequently, the ATE estimator $\hat  \tau_{n,m}$ depends on the size of two data sets, raising questions on how this estimator behaves (bias and variance) as function of $n$ and $m$.

\subsubsection{Simulations and first observations}

\begin{figure}[b!]
        \centering
        \begin{subfigure}[b]{0.28\textwidth}
            \centering
            \includegraphics[width=0.95\textwidth]{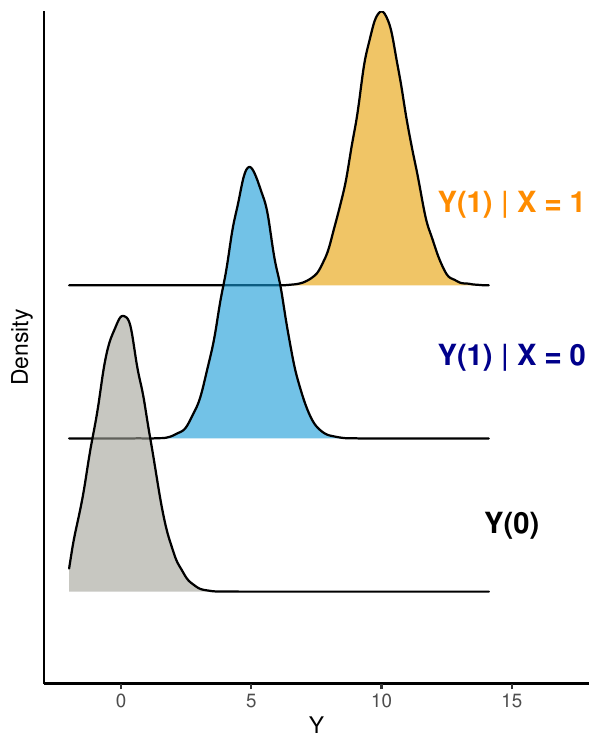}
            \caption[Network2]%
            {{\small \textbf{Toy example's data generative model}: where individuals with $X=1$ have a higher average treatment effect compared to individuals with $X=0$. The baseline, centered on $0$, is the same for both stratum.}} 
            \label{fig:toy_example_response_level}
        \end{subfigure}
        \hspace{0.2cm}
        \begin{subfigure}[b]{0.68\textwidth}
            \centering
            \includegraphics[width=0.95\textwidth]{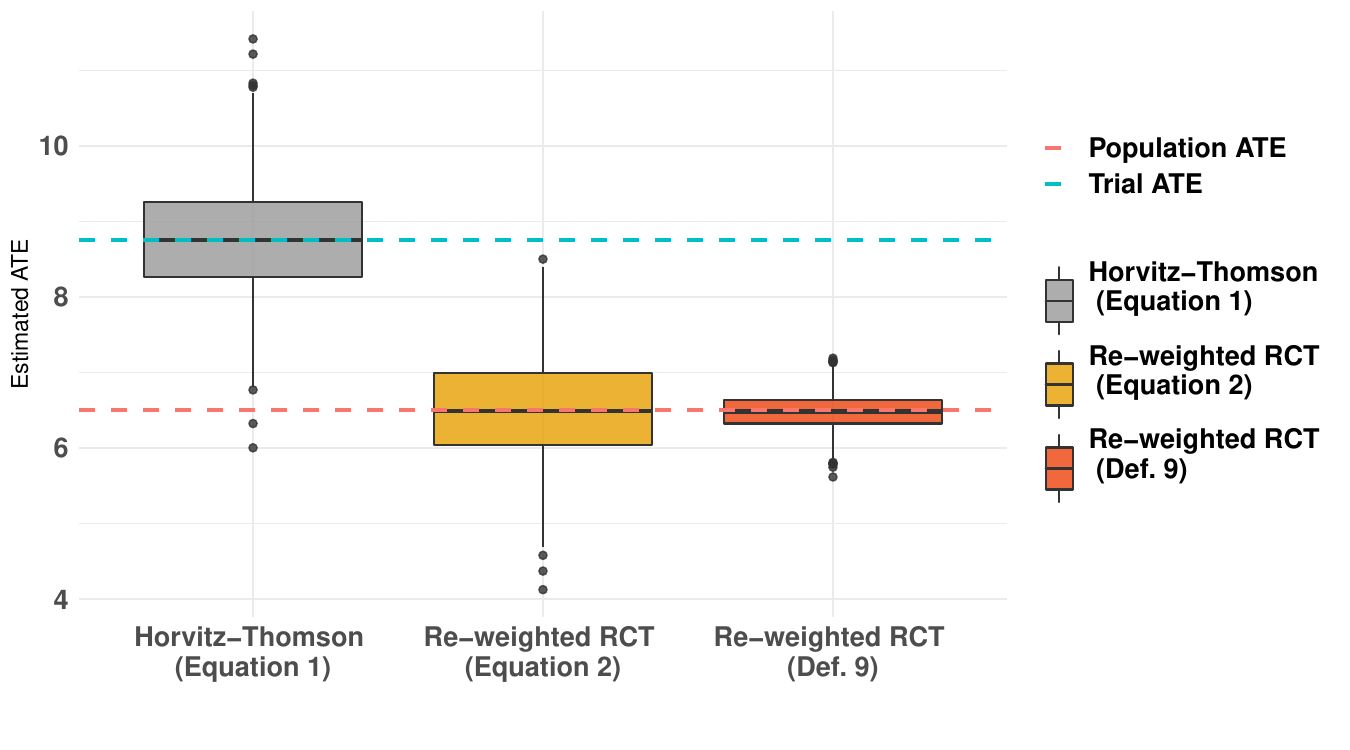}
            \caption[Network2]%
            {{\small \textbf{Re-weighting in action}: Simulations's results with a trial of size $n=150$, a target sample of size $m=1,000$ with $1,000$ repetitions, where the \textcolor{gray}{naive trial estimate} corresponds Equation~\ref{eq:toy-estimator-HT}, and \textcolor{YellowOrange}{re-weighted trial} to Equation~\ref{eq:toy-estimator-oracle}. As expected re-weighting allows to recover the ATE of the target population (red dashed line). It is also possible to \textcolor{RedOrange}{estimate $\pi$} from the data, giving another re-weighting estimator with lower variance (later introduced in Definition~\ref{def:ipsw-with-pi}).}}    
            \label{fig:toy_example_simplest_expe}
        \end{subfigure}
        \vskip\baselineskip
         \begin{subfigure}[b]{0.90\textwidth}
            \centering
            \includegraphics[width=0.85\textwidth]{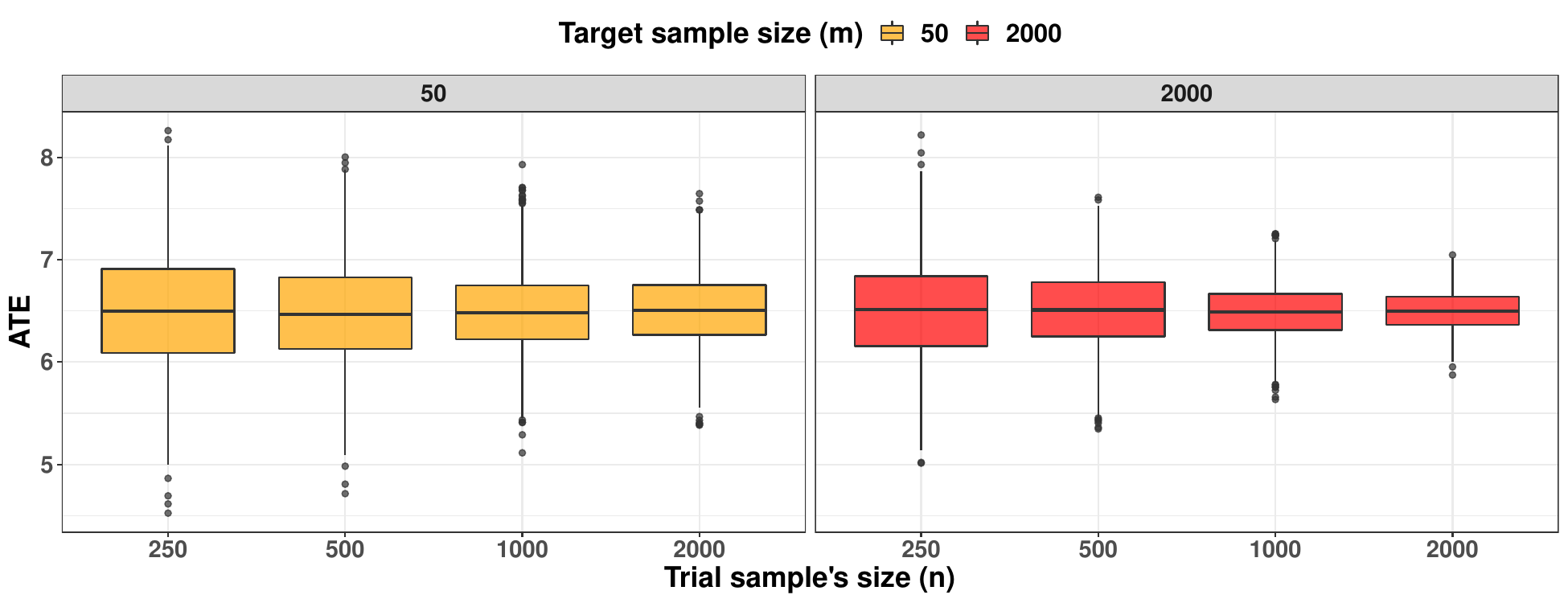}
            \caption[Network2]%
            {{\small \textbf{Two data sets leading to two asymptotic regimes}: where two situations are considered, one with a large target sample ($m=2000$) or a small target sample ($n=50$). Then, increasing $n$ leads to a variance stagnation if $m$ is small, while increasing $n$ allows to further gain in precision if $n \le m$.}}
            \label{fig:toy_example_2_asympt}
        \end{subfigure}
        \caption[ Toy example's simulations - General results under the \textbf{Minimal} adjustment set, that is covariate $X$. ]
        {\small Toy example's simulations - \textbf{Minimal} adjustment set.} 
        \label{fig:simulation-toy-example-overall-results}
\end{figure}

To investigate empirically how $\hat \tau_{n,m}$ behaves, we run simulations following the Data Generative Process (DGP) described in Section~\ref{subsection_toy_example} and represented in Figure~\ref{fig:toy_example_response_level}.  Figure~\ref{fig:toy_example_simplest_expe} shows the different estimators in action, showing that the re-weighted trial compensates for the distribution shift as expected. 

Figure~\ref{fig:toy_example_simplest_expe} also shows that estimating $\pi$ from the data and plugging it in Equation~\ref{eq:toy-estimator-oracle} leads to a clear gain in variance. 
     This phenomenon is linked to seminal works in causal inference, 
     and is further demonstrated in Section~\ref{subsec:also-estimating-pi}. 
Finally, Figure~\ref{fig:toy_example_2_asympt} shows that 
if $m$ remains small compared to $n$ or if $n$ remains small compared to $m$, then the asymptotic variance regime differs (see Corollary~\ref{cor_asympt_completely_estimated} for a formal statement, and Figure~\ref{fig:corollary3} for an illustration of the theoretical results).\\

For correct trial generalization, all shifted treatment effect modifier baseline covariates (see Definition~\ref{def:V-is-not-treat-effect-modifier} and \ref{def:V-is-not-shifted}, Section~\ref{subsec:extended_adjustement_set}), such as the genetic mutation $X$, are necessary \citep{stuart2011use}. 
But, in practice \emph{one may be tempted to add as many covariates $V$ as available to account for all possible sources of external validity bias}. 
Doing so, we may add covariates $V$ that are not needed to properly estimates the weights. This is the case if \textit{(i)} $V$ is shifted between the two data sets, but in reality is not a treatment effect modifier or if \textit{(ii)} $V$ is a treatment effect modifier, but not shifted between the two data sets.
Figure~\ref{fig:toy_example_shifted_covariates} shows that
in $(i)$, 
the covariate $V$ should not be added,
as it can considerably inflate the variance and therefore damage the precision (see Corollary~\ref{proposition:adding-shifted-covariates} for a formal statement); 
while in $(ii)$, Figure~\ref{fig:toy_precision_covariates} highlights that 
the covariate $V$ should be added as 
 the precision can be augmented by adding such covariates (see Corollary~\ref{proposition:adding-treat-effect-modifier-covariates} for a formal statement).\\ 
\begin{figure}[!h]
        \centering
        \begin{subfigure}[b]{0.53\textwidth}
            \centering
            \includegraphics[width=0.85\textwidth]{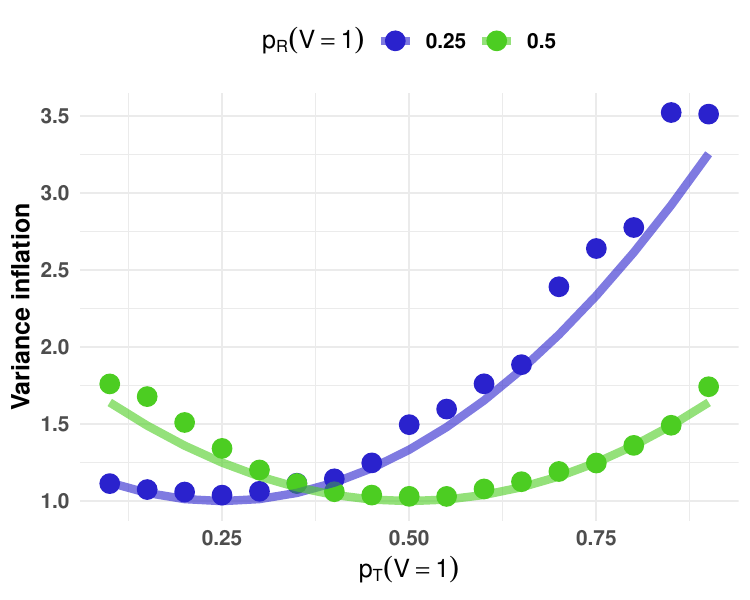}
            \caption[Network2]%
            {{\small \textbf{Adding shifted covariate that is not a treatment effect modifier} leads to a variance inflation. Simulation represents the situation of a binary shifted covariate $V$ added or not in the adjustment set. The $y$-axis represents how much the variance with the minimal set is multiplied compared to a situation with this additional shifted covariate. The plain lines comes from the Theory (see Corollary~\ref{proposition:adding-shifted-covariates}) while dots are empirical variance (obtained from $1,000$ repetitions with $n=150$ and $m=1,000$). The more shifted the covariate, the higher the inflation. The phenomenon is amplified if the covariate is \textcolor{BlueViolet}{\textbf{imbalanced}} in the trial (in opposition with a \textcolor{LimeGreen}{\textbf{balanced}}).}}    
            \label{fig:toy_example_shifted_covariates}
        \end{subfigure}
        \hspace{0.5cm}
        \begin{subfigure}[b]{0.3\textwidth}  
            \centering 
            \includegraphics[width=0.92\textwidth]{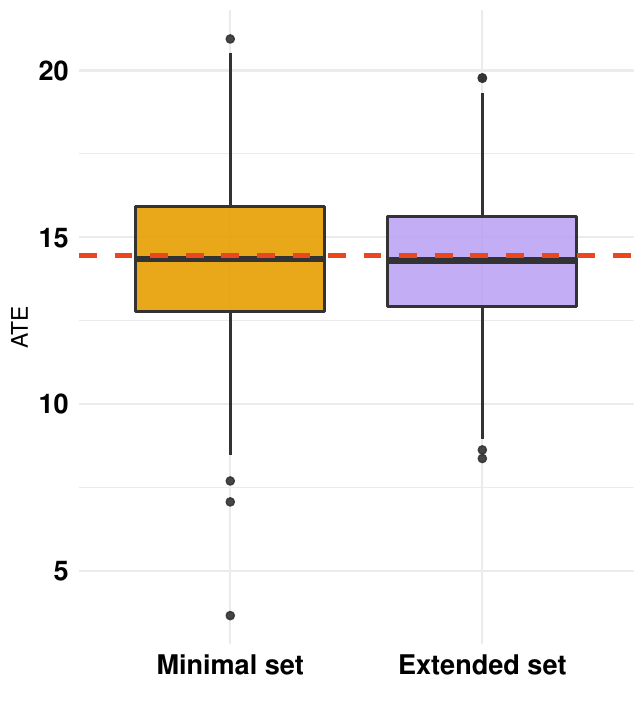}
            \caption[]%
            {{\small \textbf{Adding non-shifted treatment effect modifier} leads to a gain in precision compared to a situation with only the necessary covariate. In this plot DGP from Figure~\ref{fig:toy_example_response_level} is adapted to add one non-shifted treatment effect modifier. Adding such covariate (\textcolor{Orchid}{\textbf{extended set}}) compared to an adjustment set with only $X$ (\textcolor{YellowOrange}{\textbf{minimal set}}) lowers the variance.}}
            \label{fig:toy_precision_covariates}
        \end{subfigure}
        \caption{Toy example's simulations - \textbf{Extended} adjustment set.}
        \label{fig:simulation-toy-example-additional-covariates}
\end{figure}

In Section~\ref{section:theoretical-results}, we prove these phenomenons, deriving explicit finite sample and asymptotic results to characterize the re-weighting process.


\subsection{Related work}\label{subsec:related-work}
\bc{ajouter le papier de Kara Rudolph}
The estimator $\hat{\tau}_{n,m}$ introduced in the toy example (Equation \ref{eq:toy-estimator-oracle})  is an exact implementation of the so-called \textit{Inverse Propensity of Sampling Weighting} (IPSW) where the word \textit{sampling} comes from the popular habit of modeling the problem as the one of a randomized trial suffering from selection bias \citep{cole2010generalizing, Bareinboim2012Controlling, tipton2013improving, dahabreh2019generalizing}. 
Note that the estimator introduced in Equation \ref{eq:toy-estimator-oracle} can also be linked to post-stratification \citep{imbens2011experimental, miratrix2013adjusting}, where post-stratification belongs to the family of adjustment methods on a single RCT.
Note that beyond trial re-weighting, other estimation strategies can be chosen when it comes to generalization, for example stratification \citep{tipton2013improving, Muircheartaigh2014GeneralizingApproach}, modeling the response (G-formula or Outcome Modeling) \citep{kern2016assessing, dahabreh2019generalizing}, using both strategies in a so-called doubly-robust approach (AIPSW) \citep{dahabreh2019generalizing, dahabreh2020extending}, or entropy balancing \citep{josey2021transporting, dong2020integrative}. 

\paragraph{Link with IPW} The IPSW can be related -- to a certain extent -- to the well-known \textit{Inverse Propensity Weighting} (IPW) estimator in the context of a single observational data set \citep{Hirano2003Efficient}. Indeed, this corresponds to a mirroring situation, where the weights are no longer the probability ratio, but the probability to be treated  \citep[propensity score,][]{rosenbaum1983centralrolepropensity}. \cite{Robins1992EstimatingEE, Hahn1998efficiencybound, Hirano2003Efficient} showed that IPW is more efficient when weights are estimated, rather than relying on oracle weights. 
This curious phenomenon can even be found in other areas of statistics \citep{Efron1978ObservedVsExpected}.
Beyond efficient estimation with a minimal adjustment set, it is known that additional and non-necessary baseline covariates in the adjustment set of the IPW can either increase the variance (the so-called instruments) \citep{Velentgas2013DevelopingAP, Schnitzer2015VariableSelection, Wooldridge2016Instrument}, while another class of covariates (the ones linked only to the outcome -- and also called \textit{outcome-related covariates} or \textit{risk factors} or \textit{precision covariates}) improves precision \citep{Hahn2004FunctionalRestriction, Lunceford04stratificationand, brookhart2006variable, Lefebvre2008Mispecification, witte2018selection}. A recent crash-course about good and bad controls recalls this phenomenon \citep{Cinelli2020CrashCourse}.
Finally, another very recent line of research consists in determining -- given a Directed Acyclic Graph (DAG) --  the asymptotically-efficient adjustment set for ATE estimation. This is also named ‘optimal’ valid adjustment set (O-set), corresponding to the adjustment set ensuring the smallest limiting variance compared to other adjustment sets. \cite{Henckel2019GraphicalAdjs} propose a result for linear model, and \cite{Rotnitzky2020Efficient} extend this work for any non-parametrically adjusted estimator. Such methods are meant for complex DAGs where several possible adjustment sets can be used.\\

\paragraph{Theoretical results on IPSW} Expression of the variance has been proposed for an estimator related to the IPSW: the stratification estimator \citep{Muircheartaigh2014GeneralizingApproach, tipton2013improving}. These results only consider the situation of an infinite target sample. 
Similar 
expressions can also be found in \cite{Rothman2000ModernEpidemiology}, also assuming an infinite target sample compared to the trial sample size.
\cite{buchanan2018generalizing} propose theoretical properties such as limiting variance of a variant of IPSW under a parametric model, using M-estimation methods for the proof \citep{Stefanski2002Mestimation}. \textit{Why a variant?} Because their proof is under the situation of a so-called nested design, that is a trial embedded in a larger observational population, so that there is only one single data set to consider and not two. 
In addition, we have found no discussion - neither empirical nor theoretical - about the impact of adding non-necessary covariates 
on the IPSW (or any other generalization's estimator) properties (e.g., bias, variance). 
\cite{Egami2021CovariateSelection} propose a method to estimate a separating set -- \emph{i.e.} a set of variables affecting both the sampling mechanism and treatment effect heterogeneity – and in particular when the trial contains many more covariates than the target population sample. However, their work focuses on identification. 
\cite{Huitfeldt2019EffectHeterogeneity} also consider covariate selection for generalization, but analyze whether the necessary number of covariates can be reduced by considering a specific effect measure (ratio, difference, other) instead of the entire counterfactual distribution. 
\cite{Yang2020DoublyRI} addresses a similar problem (for non-probability sample and mean estimation), where they advocate selecting all variables, even instrumental variables, for robustness, although it may come at the cost of drop in efficiency.
Note that some existing practical recommendations advocate to add as many covariates as possible \citep{stuart2017CaseStudyDifficulties}.\\

\paragraph{Contributions}
This work considers several variants of the IPSW estimator, whether or not the weights are oracle, semi-oracle, or estimated.
In this context, we derive the limiting variance of all the variants of IPSW and we show that several asymptotic regimes exist, depending on the relative size of the RCT compared to the target sample. 
We also provide finite sample expression of the bias and variance for all the IPSW variants introduced, allowing to bound the risk on this estimator for any samples sizes (trial and target population).
From these theoretical results, we explain why the addition of some additional but non-necessary covariates in the adjustment set has a large impact on precision, for the best or the worst.
Indeed, while non-shifted treatment effect modifiers improve precision by lowering the variance, adding shifted covariates that are not predictive of the outcome considerably reduces the statistical power of the analysis by inflating the variance.
For this latter situation, we provide an explicit formula of the variance inflation when the additional covariate set is independent of the necessary one.
These results have important consequences for practitioners because they allow to give precise recommendations about how to select covariates.
Note that we link our work to seminal works in causal inference, showing that semi-oracle estimation outperforms a completely oracle estimation, while the exact result on IPW on efficient estimation can not be completely extended to the case of generalization. \\

All our results assume neither a parametric form of the outcome nor the sampling process, but are established at the cost of restricting the scope to categorical covariates for adjustment. Within the medical domain, scores or categories are often used to characterize individuals, which justifies this approach.

\newpage
\section{Notations and assumptions for causal identifiability}\label{section:notations-and-assumptions}

\subsection{Notations}

\subsubsection{Problem setting}\label{subsec:model}

The notations and assumptions used in this work are grounded in the potential outcome framework \citep{imbens2015causal}.
We assume to have at hand two data sets:
\begin{description}
    \item[A randomized controlled trial] denoted $\mathcal{R}$ (for randomized), assessing the efficacy of a binary treatment $A$ on an outcome $Y$ (ordinal, binary, or continuous) conducted on $n$ \textit{iid} observations. Each observation $i$ is labelled from $1$ to $n$ and can be modelled as sampled from a distribution $P_\text{\tiny R}(X,Y^{(1)}, Y^{(0)}, A) \in \mathds{X} \times \mathbb{R}^2 \times\{0,1\}$, where $\mathds{X}$ is a categorical support. For any observation $i$, $A_i$ denotes the binary treatment assignment (with $A_i=0$ if no treatment and $A_i=1$ if treated), and $Y_i^{(a)}$ is the outcome had the subject been given treatment $a$ (for $a\in\{0,1\}$), which is assumed to be squared integrable.
    $Y_i$ denotes the observed outcome, defined as $Y_i = A_i \, Y_i^{(1)} +  (1-A_i) \, Y_i^{(0)}$. In addition, this trial is assumed to be a Bernoulli trial with a constant probability of treatment assignment for all units and independence of treatment allocation between units (see in appendix Definition~\ref{def:bernoulli-trial})\footnote{For a review of trial designs, in particular explaining the difference between a Bernoulli and a completely randomized design, we refer the reader to Chapter 2 of \cite{imbens2015causal}.}. We denote $\mathbb{P}_\text{\tiny R}\left[A_i = 1 \right] = \pi$. $X_i$ is a $p$-dimensional vector of categorical covariates accounting for individual characteristics on the observation $i$;
    \item[A sample of the target population of interest] denoted $\mathcal{T}$ (for target), containing $m$ \textit{iid} individuals samples drawn from a distribution $P_\text{\tiny T}(X,Y^{(1)}, Y^{(0)}, A) \in \mathds{X} \times \mathbb{R}^2 \times\{0,1\}$, labelled from $n+1$ to $n+m$. In this data set, we only observe individual categorical characteristics $X_i$. For simplicity, we further use the notation $P_\text{\tiny T}(X)$ for the marginal of $X$ on distribution $P_\text{\tiny T}$.
\end{description}

Finally, the probability of $X$ in the target population (resp. trial population) is denoted $p_{\text{\tiny T}}(x)$ (resp. $p_{\text{\tiny R}}(x)$).
Mathematically, a covariate shift between the two populations occurs when there exists $x \in \mathds{X}$ such that $p_{\text{\tiny R}}(x) \neq p_{\text{\tiny T}}(x)$. The setting and notations are summarized on Figure~\ref{fig:notations-helper}.
 
 \begin{figure}[!h]
    \begin{minipage}{.38\textwidth}
	\caption{\textbf{Summary of the data at hand}: on the left, a randomized controlled trial $\mathcal{R}$ of size $n$ sampled according to $P_\text{\tiny R}$ and informing about the effect of a treatment $A$ on the outcome $Y$. On the right, a sample $\mathcal{T}$ of size $m$ sampled from the target population of interest $P_\text{\tiny T}$, containing only information on covariates $X$. As suggested on the drawing, $n$ is often smaller than $m$, as trials are usually of limited size compared to large national data base or cohort.}
     \label{fig:notations-helper}
    \end{minipage}%
    \hfill%
    \begin{minipage}{.6\textwidth}
	\includegraphics[width=\linewidth]{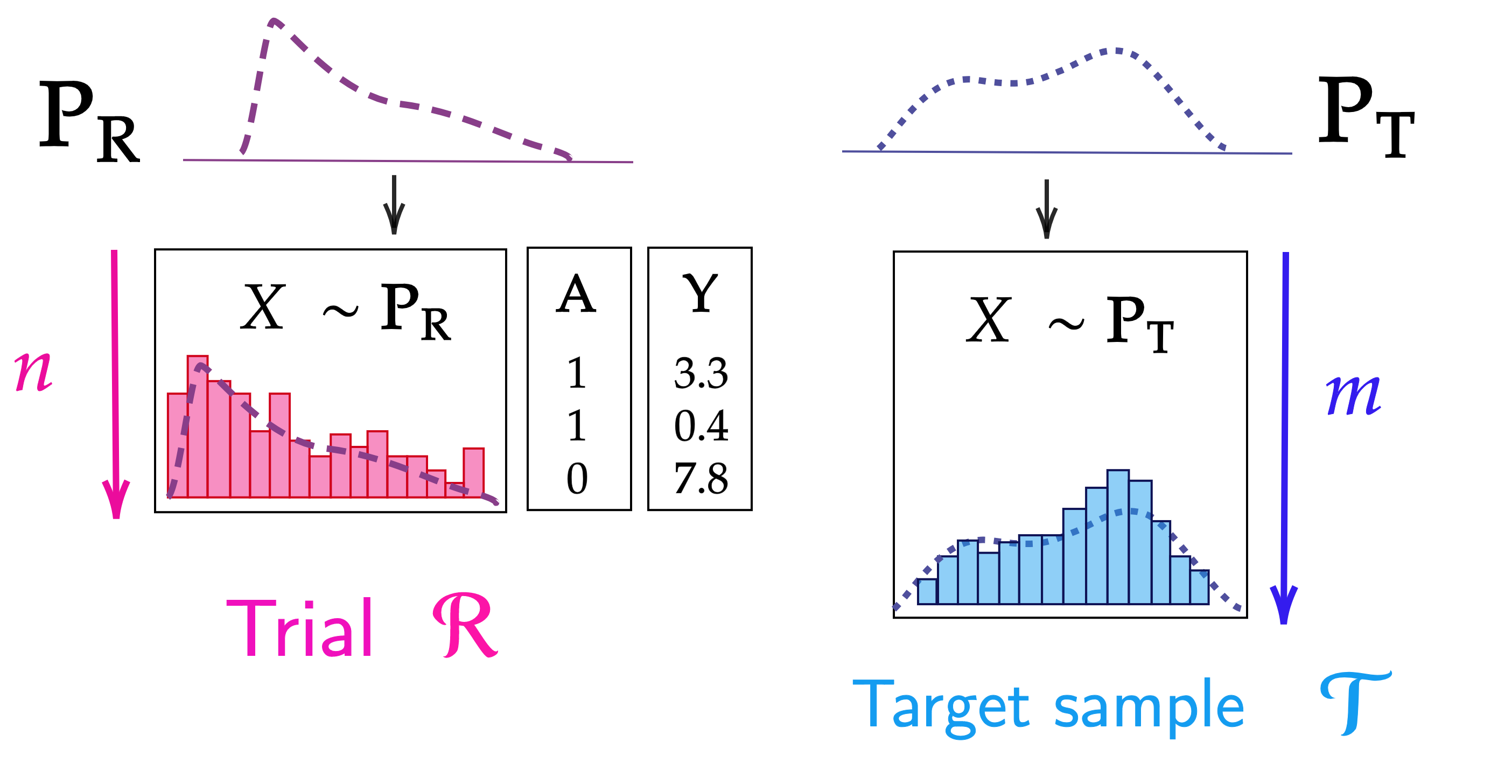}%
    \end{minipage}
 \end{figure}

 \paragraph{Comments on the notations}
 Note that a large part of the literature models the problem with a sampling mechanism from a super population. Doing so, the target and the trial samples are assumed sampled from this super population, with different mechanisms leading to a distributional shift of the trial \cite[e.g. the framing in ][]{stuart2011use, hartman2021inbook}. Still, as soon as we are not working with a nested trial (that is a trial embedded in the target sample) and if only baseline covariates are considered for adjustment, the framing with a sampling model is equivalent to the problem setting introduced above \citep{Colnet2020review, Westreich2017IOSW}.
 Note that the literature is increasing adopting the framing that we use here
\citep{kern2016assessing, nie2021covariate, Chattopadhyay2022OneStep}.
 
 \subsubsection{Target quantity of interest}

Recall that two distributions, indexed by $\text{R}$ and $\text{T}$ are involved in our problem setting (Section~\ref{subsec:model}). Therefore, we will use these indices to denote quantities (expectations, probabilities) taken with respect to these distributions, 
for example $\mathbb{E}_{\text{\tiny R}}\left[ .\right]$ (resp. $\mathbb{E}_{\text{\tiny T}}\left[ .\right]$) for an expectation over $P_\text{\tiny R}$ (resp. $P_\text{\tiny T}$).

We define the target population average treatment effect ATE (sometimes called \textbf{T}ATE for \textbf{T}arget): 
\begin{align}\label{eq:target-quantity}
 \tau := \mathbb{E}_{\text{\tiny T}}\left[Y^{(1)} - Y^{(0)}\right].  
\end{align}
Because the randomized controlled data $\mathcal{R}$ are not sampled from the target population of interest, the sample average treatment effect $\tau_{\text{\tiny R}}$ (sometimes called \textbf{S}ATE for \textbf{S}ample) estimated from this population,
\begin{equation*}\label{eq:key-issue-equation}
   \tau_{\text{\tiny R}}:=  \mathbb{E}_{\text{\tiny R}}\left[Y^{(1)} - Y^{(0)}\right],
\end{equation*}
may be biased, that is $\tau_{\text{\tiny R}} \neq \tau$. While not being the target quantity of interest, we also introduce the so-called Conditional Average Treatment Effect (CATE), as

\begin{equation*}
    \forall x \in \mathds{X},\, \tau(x):= \tau_{\text{\tiny T}}(x):= \mathbb{E}_{\text{\tiny T}}\left[ Y^{(1)} - Y^{(0)} \mid X =x \right].
\end{equation*}


\subsection{Identification assumptions}\label{subsec:assumptions}

Assumptions are needed
to be able to generalize the findings from the population data $P_\text{\tiny R}$ toward the population $P_\text{\tiny T}$. 

\paragraph{Assumptions on the trial}
We first need validity of the trial, also called \textit{internal validity}. These assumptions are the usual ones formulated in causal inference, and in particular for randomized controlled trials within the potential outcomes framework \citep{imbens2015causal, hernan2020book}. 

\begin{assumption}[Representativity of the randomized data]\label{a:repres-rct} For all $i \in \mathcal{R}, X_i \sim P_\text{\tiny R}(X)$ where $P_\text{\tiny R}$ is the population distribution from which the RCT was sampled. 
\end{assumption}

\begin{assumption}[Trial's internal validity]\label{a:trial-internal-validity} The RCT at hand $\mathcal{R}$ is assumed to be internaly valid, such that 
\begin{enumerate}
    \item[(i)] Consistency and no interference hold, that is:  $\forall i\in \mathcal{R},\, Y_i = A_i \, Y_i^{(1)}  + (1-A_i) \, Y_i^{(0)}$
    --an assumption often termed SUTVA (stable unit treatment value);
    \item[(ii)] Treatment randomization holds, that is: $\forall i\in \mathcal{R},\, \left\{Y^{(1)}_{i}, Y^{(0)}_{i}\right\} \perp A_{i}$;  
    \item[(iii)] Positivity of trial treatment assignment holds, that is: $0 < \pi < 1$ (usually $\pi = 0.5$). 
\end{enumerate}
\end{assumption}

\paragraph{Assumptions for generalization} 
The two following assumptions are specific to generalization or transportability. Let us define the CATE $ \tau_{\text{\tiny R}}$ on the RCT as $ \forall x \in \mathds{X}$,
\begin{equation*}
   \tau_{\text{\tiny R}}(x):= \mathbb{E}_{\text{\tiny R}}\left[ Y^{(1)} - Y^{(0)} \mid X =x \right].
\end{equation*}

\begin{assumption}[Transportability]
\label{a:cate-indep-s-knowing-X}
We have, $\forall x \in \mathds{X}, \tau_{\text{\tiny R}}(x) = \tau_{\text{\tiny T}}(x)$.
\end{assumption}

The transportability assumption \citep{stuart2011use, pearl2011transportability}, also called \textit{sample
ignorability for treatment effects} \citep{kern2016assessing} or \textit{Conditional Ignorability} \citep{hartman2021inbook}, is probably the most important assumption to generalize or transport the trial findings to the target population, as this requires to have access to all shifted covariates being treatment modifiers. 
In other words, it assumes that all the systematic variations in the treatment effect are captured by the covariates $X$ \citep{Muircheartaigh2014GeneralizingApproach}. 
The covariates $X$ are usually named the adjustment or separating set.
Note that the concept of treatment effect modifiers depends on the causal measure chosen; in this paper, we only consider the absolute difference most common for a continuous outcome as detailed in Equation~\ref{eq:target-quantity}. 
Would we have chosen the risk-ratio, for instance, then the covariates being treatment effect modifiers could be different.
Finally, note that \cite{pearl2011transportability} introduces \textit{selection} diagram to formalize this assumption relying on causal diagrams. \cite{pearl2015findings} details why diagrams can contain more identification scenarii. But in this work, we only consider baseline covariates for the transportability assumption (i.e no \textit{front-door} adjustment). 
\begin{assumption}[Support inclusion]\label{a:pos} $\forall x \in \mathds{X}, \; p_{\text{\tiny R}}(x) > 0$, and \hspace{0.05cm}
$\operatorname{supp}(P_T(X)) \subset \operatorname{supp}(P_R(X))$.

\end{assumption} 
Note that this last assumption is sometimes referred as the positivity of trial participation 
and can also be viewed as a sampling process with non-zero probability for all individuals. 

\subsection{Estimators}

In this work, we denote any estimator targeting a quantity $\tau$ as $\hat \tau_{n,m}$ where the the index $n$ or $m$ is employed to characterise which data were used in the estimation strategy. For example, an estimator $\hat \tau_{n}$ (resp. $\hat \tau_{m}$) only uses the trial data (resp. observational data) whereas $\hat \tau_{n,m}$ uses both data sets.

\subsubsection{Within-trial estimators of ATE}

Two classical estimators targeting $\tau_{\text{\tiny R}}$ from trial data are the Horvitz-Thomson and Difference-in-means estimators.

\begin{definition}[Horvitz-Thomson - \cite{HorvitzThompson1952seminal}]\label{def:HT} The Horvitz-Thomson estimator is denoted $\hat \tau_{\text{\tiny HT},n}$ and defined as,
\begin{equation*}
 \hat{\tau}_{\text {\tiny HT,}n} = \frac{1}{n} \sum_{i=1}^{n}\left(\frac{A_{i} Y_{i}}{\pi}-\frac{\left(1-A_{i}\right) Y_{i}}{1-\pi}\right)   .
\end{equation*}

\end{definition}
Under a Bernoulli design (constant and independent probability to be treated $\pi$) the Horvitz-Thomson estimator $\hat \tau_{\text{\tiny HT},n}$ is an unbiased and consistent estimator of $\tau_{\text{\tiny R}}$, and its variance satisfies, for all $n$,  
\begin{equation}\label{eq:variance-for-HT}
     n \operatorname{Var}\left[ \hat \tau_{\text{\tiny HT},n} \right] =   \mathbb{E}_{\text{\tiny R}}\left[ \frac{\left( Y^{(1)} \right)^2}{\pi} \right]  + \mathbb{E}_{\text{\tiny R}}\left[ \frac{\left( Y^{(0)} \right)^2}{1-\pi} \right]  - \tau_{ \text{\tiny R}}^2:= V_{ \text{\tiny HT}}.
\end{equation}

\begin{definition}[Difference-in-means - Neyman (1923) and its English translation \cite{SplawaNeyman1990Translation}]\label{def:difference-in-means} The Difference-in-means estimator is denoted $  \hat{\tau}_{\text{\tiny DM,}n}$ and defined as the difference between two terms 
\begin{align}
\hat{\tau}_{\text{\tiny DM,}n} & = \hat{\mu}_{1,n} - \hat{\mu}_{0,n},    
\end{align}
where, for all $a \in \{0,1\}$, 
\begin{align}
    \hat{\mu}_{a,n} = \left\lbrace
    \begin{array}{ll}
           \frac{1}{n_a} \sum_{i=1}^n Y_i \mathds{1}_{A_i = a} & \textrm{if } n_a > 0\\
          0 & \textrm{otherwise}
    \end{array}
    \right.
\end{align}
with $n_{a}= \sum_{i=1}^n \mathbbm{1}_{A_i = a}$.
\end{definition}

The Difference-in-means is also referred to as the \textit{simple difference estimator}, e.g. in \cite{miratrix2013adjusting}, or \textit{difference in the sample means of the observed outcome variable between the treated and control groups}, e.g. in \cite{Imai2008Misunderstanding}.
Under a Bernoulli design, the difference-in-means estimator is a consistent estimator of $\tau_{\text{\tiny R}}$, and its finite sample variance is bounded by
\begin{equation}\label{eq:finite-sample-variance-for-DM}
      n \operatorname{Var}\left[ \hat{\tau}_{\text{\tiny DM},n} \right] \le  \frac{ \operatorname{Var}\left[Y^{(1)}\right] }{\pi}  + \frac{ \operatorname{Var}\left[Y^{(0)}\right] }{1-\pi}+ \mathcal{O}\left( n^{-1/2}\right),
\end{equation}
and its large sample variance satisfies,
\begin{equation}\label{eq:variance-for-DM}
   \lim_{n\to \infty}   n \operatorname{Var}\left[ \hat \tau_{\text{\tiny DM},n} \right] = \frac{\operatorname{Var}\left[Y_i^{(1)}\right]}{\pi} + \frac{\operatorname{Var}\left[Y_i^{(0)}\right]}{1-\pi}:=  V_{\text{\tiny DM},\infty}.
\end{equation}

An explicit expression of the finite sample bias and variance of $ \hat{\tau}_{\text{\tiny DM},n} $ are given in appendix (see Lemma~\ref{lemma:DM-bias-and-variance}).
Note that the Difference-in-Means estimator is known to be unbiased. However, this is due do the fact that  a slightly different version of the estimator is often used, which is undefined if all units are either treated or control. While the unconditional estimator we consider (\Cref{def:difference-in-means}) is biased, the estimator, conditionally on the fact that there exist at least one control and one treated unit, is unbiased  (see also \cite{miratrix2013adjusting} for a discussion on this topic).

What will be used later on, is the fact that the Difference-in-Means estimator can be viewed as a variant of the Horvitz-Thomson estimator, where the probability to be treated $\pi$ (or propensity score) is estimated, that is,
\begin{equation*}
     \hat{\tau}_{\text{\tiny DM,}n}= \frac{1}{n} \sum_{i=1}^n \left( \frac{A_i\,Y_i}{\hat \pi} - \frac{(1-A_i)\,Y_i}{1-\hat \pi} \right), \quad \text{where } \hat \pi = \frac{\sum_{i=1}^n A_i}{n}.
\end{equation*}

Counter-intuitively, the benefit of estimating $\pi$ is to lower the variance. 
Even if the true probability is $\pi = 0.5$, the actual treatment allocation in the sample can be different (e.g., $\hat \pi = 0.48$), and using $\hat \pi$ rather than $\pi$ leads to a smaller large sample variance by adjusting to the exact observed probability to be treated in the trial. In particular, it is possible to be convinced of this phenomenon when comparing the two variances, 
\begin{equation}\label{eq:ineq-dm-ht}
      V_{ \text{\tiny DM},\infty}= V_{ \text{\tiny HT}} - \left( \sqrt{\frac{1-\pi}{\pi}} \mathbb{E}_{ \text{\tiny R}}[Y^{(1)}] + \sqrt{\frac{\pi}{1-\pi}} \mathbb{E}_{\text{\tiny  R}}[Y^{(0)}]\right)^2 \le V_{ \text{\tiny HT}}. 
\end{equation}
Appendix~\ref{appendix:useful-results-rct} recalls
derivations to obtain \eqref{eq:variance-for-HT} to \eqref{eq:ineq-dm-ht}.
Other estimators of $\tau_{\text{\tiny R}}$ exist, and rely on prognostic covariates (also called adjustement) such as outcome-modeling or post-stratification. Below (Section~\ref{subsec:also-estimating-pi}), we introduce the post-stratification estimator, corresponding to the Horvitz-Thomson estimator where $\pi$ is estimated according to different stratum.
 
\subsubsection{Re-weighting estimator for generalizing the trial findings}

As mentioned in Subsection~\ref{subsec:related-work}, in this work we  focus on the reweighting strategy, that is the \textit{Inverse Propensity of Sampling Weighting} (IPSW) 
estimator \citep{cole2010generalizing, stuart2011use}.

\begin{definition}[Completely oracle IPSW]\label{def:ipsw-oracle} The completely oracle IPSW estimator is denoted $\hat \tau_{\pi, \text{\tiny T, R}, n}^*$, and defined as
\begin{equation}
\hat \tau_{\pi, \text{\tiny T,R}, n}^*= \frac{1}{n} \sum_{i=1}^{n}  \frac{p_{\text{\tiny T}}(X_i)}{p_{\text{\tiny R}}(X_i)}Y_i \left( \frac{A_i}{\pi} - \frac{1-A_i}{1- \pi} \right)\,,
\end{equation}
where $\frac{p_{\text{\tiny T}}(X_i)}{p_{\text{\tiny R}}(X_i)}$ are called the weights or the nuisance components.
\end{definition}
Definition~\ref{def:ipsw-oracle} corresponds to a completely oracle IPSW, where $p_{\text{\tiny T}}$, $p_{\text{\tiny R}}$, and the trial allocation probability $\pi$ are known.

\subsubsection{Probability ratio estimation}

In practice neither $p_\text{\tiny R}$ nor $p_\text{\tiny T}$ are known, and therefore one needs to estimate these probabilities. As explained in Subsection~\ref{subsec:model}, we consider the case where $X$ is composed of categorial covariates only. In such a situation, a practical IPSW estimator can be built from Definition~\ref{def:ipsw-oracle} by estimating each probability $p_{\text{\tiny T}} $ and  $p_{\text{\tiny R}}$ by their empirical counterpart (that is counting how many observations fall in each categories in the trial and target samples). 

\begin{definition}[Probability estimation]\label{def:procedure-for-densities} 
Under the setting defined in Subsection~\ref{subsec:model},
\begin{equation*}
   \forall x \in \mathcal{X},\;\; \hat p_{\text{\tiny T},m} (x):=  \frac{1}{m}\sum_{i \in \mathcal{T}} \mathbbm{1}_{X_i = x}\;\; \text{ and,    }\,\hat p_{\text{\tiny R},n}(x) := \frac{1}{n}\sum_{i \in \mathcal{R}} \mathbbm{1}_{X_i = x}.
\end{equation*}
\end{definition}

Having defined a method for probability estimation, one can build practical IPSW variants.

\begin{definition}[Semi-oracle IPSW]\label{def:ipsw-semi-oracle} The semi-oracle IPSW estimator $\hat \tau_{\pi, \text{\tiny T}, n}^*$ is defined as
\begin{equation}
\hat \tau_{\pi, \text{\tiny T}, n}^*= \frac{1}{n} \sum_{i=1}^{n}  \frac{p_{\text{\tiny T}}(X_i)}{\hat p_{\text{\tiny R},n}(X_i)}Y_i \left( \frac{A_i}{\pi} - \frac{1-A_i}{1- \pi} \right)\,,
\end{equation}
where $\hat p_{\text{\tiny R},n}$ is estimated according to Definition~\ref{def:procedure-for-densities}.
\end{definition}
Note that this semi-oracle estimator corresponds to the so-called standardization procedure described in \cite{Rothman2000ModernEpidemiology}.

\begin{definition}[IPSW]\label{def:ipsw} The (estimated) IPSW estimator $\hat \tau_{\pi,n,m}$ is defined as
\begin{equation}
\hat \tau_{\pi, n, m}= \frac{1}{n} \sum_{i=1}^{n}  \frac{\hat p_{\text{\tiny T},m}(X_i)}{\hat p_{\text{\tiny R},n}(X_i)}Y_i \left( \frac{A_i}{\pi} - \frac{1-A_i}{1- \pi} \right)\,,
\end{equation}
where $\hat p_{\text{\tiny R},n}$ and $\hat p_{\text{\tiny T},m}$ are estimated according to Definition~\ref{def:procedure-for-densities}.
\end{definition}

Definition~\ref{def:ipsw} corresponds to the classical implementation of the IPSW since, practically, the probabilities $\hat p_{\text{\tiny R},n}$ and $\hat p_{\text{\tiny T},m}$ are not known and must be estimated.

\paragraph{Another interpretation of IPSW}
Note that the IPSW can be understood differently, thanks to the fact that covariates used to adjust are categorical. Indeed, it is possible to re-write the IPSW estimator from Definition~\ref{def:ipsw} as,
\begin{align*}
    \hat \tau_{\pi, n, m} = \sum_{x \in \mathds{X}} \frac{m_x}{m} \sum_{i=1}^n \mathds{1}_{X_i = x}\frac{1}{n_x}\left(\frac{ A_i Y_i^{(1)}}{\pi} - \frac{ (1-A_i) Y_i^{(1)}}{1-\pi}\right) = \sum_{x \in \mathds{X}} \frac{m_x}{m} \hat \tau_{\text{\tiny HT},n_x},
\end{align*}

where $m_x = \sum_{i=n+1}^m \mathds{1}_{X_i=x}$ and $n_x = \sum_{i=1}^n \mathds{1}_{X_i=x}$. This corresponds to a procedure where stratum average treatment effects are estimated with an Horvitz-Thomson procedure, and then aggregated with weights corresponding to the target sample proportions. \cite{miratrix2013adjusting} also discusses a similar approach in their section 5, but where the sample proportions corresponds to the true target population of interest. In a way, our work extends this situation to a more general case, considering the noise due to the sampling process from two populations.

\paragraph{Comment about oracle and semi-oracle interest}
The completely-oracle and the semi-oracle estimators are not used in practice, as usually none of the true probabilities are known. Still, they both correspond to some asymptotic situations that are of interest to understand the IPSW. For instance:

\begin{itemize}
    \item Studying $\hat \tau_{\pi, \text{\tiny T}, \text{\tiny R}, n}^*$ allows us to observe the effect of averaging over the trial sample $\mathcal{R}$, without the variability due to covariates probabilities estimation ($\hat p_{\text{\tiny R},n}$ and $\hat p_{\text{\tiny T},m}$); 
    \item Studying $\hat \tau_{\pi, \text{\tiny T}, n}^*$ allows to understand the situation where the target sample $\mathcal{T}$ is infinite ($m \rightarrow \infty$). 
\end{itemize} 
In addition, studying these estimators allows us to link our results with seminal works in causal inference showing that the estimated propensity score can lead to better properties than an oracle one \citep{Robins1992EstimatingEE, Hahn1998efficiencybound, Hirano2003Efficient}.
Note that we could introduce another semi-oracle estimator, where $p_\text{\tiny R}$ is known but not $p_\text{\tiny T}$. This specific estimator does not correspond to a limit situation helping to figuring out the results, as it is as if the covariates probabilities in the trial are learned on a infinite data sample, but where the treatment effect estimate is still averaged on a finite sample. Finally, since all covariates are assumed to be categorical in our framework, trial and observational densities (continuous covariates) turn into trial and observational probabilities (categorical covariates). Oracles and semi-oracles will be different when considering continuous covariates as the weights will be replaced by density estimation or estimation of the probability of being in the target population (instead of the experimental sample) \citep[e.g. see][]{kern2016assessing, nie2021covariate}), sometimes directly estimating the ratio by binding the two data sources and therefore making the notion of semi-oracle outdated.  

\section{Theoretical results}\label{section:theoretical-results}

In this section, we establish upper bounds on the bias, variance and quadratic risk of several variants of IPSW estimates. The exact derivation of these quantities, together with the proofs of the results, are deferred to Appendix~\ref{proof:completely-oracle}. As a by-product of our analysis, all variants of IPSW considered here are $L_2$ consistent, as their risk tends to zero as the sample size increases. 

\subsection{Bias and variance of IPSW variants in finite-sample regime}

In this section, we expose our main theoretical results on the three variants of the IPSW estimator (Definition~\ref{def:ipsw-oracle}, \ref{def:ipsw-semi-oracle} and \ref{def:ipsw}).
The following results rely on the variance of the Horvitz-Thomson estimator on a given strata $x$ (see Definition~\ref{def:HT}), denoted $V_{ \text{\tiny HT}}(x)$, and defined as ,
\begin{equation}\label{eq:conditional-VHT}
    V_{ \text{\tiny HT}}(x) := \mathbb{E}_{\text{\tiny R}}\left[ \frac{\left( Y^{(1)} \right)^2}{\pi} \mid X = x\right]  + \mathbb{E}_{\text{\tiny R}}\left[ \frac{\left( Y^{(0)} \right)^2}{1-\pi} \mid X = x\right]  - \tau(x)^2.
\end{equation}

In this equation, we removed the index $R$ of $\tau(x)$ as $\tau_{\text{\tiny R}}(x) = \tau_{\text{\tiny T}}(x) = \tau(x)$, thanks to Assumption~\ref{a:cate-indep-s-knowing-X}. 
Removing the index on the two conditional expectations would require to go beyond the classical transportability assumption, by assuming that \\ 
$$\forall a \in \{0,1\}, P_{\text{\tiny R}}(Y^{(a)} \mid X =x) = P_{\text{\tiny T}}(Y^{(a)} \mid X =x),$$
\textit{i.e.} $X$ contains all the covariates being shifted and predictive of the outcome, which is stronger than Assumption~\ref{a:cate-indep-s-knowing-X}. Besides, according to the law of total variance, the following expression holds
\begin{align}
V_{ \text{\tiny HT}} =  \operatorname{Var}_\text{\tiny R}\left[ \tau(X) \right] + \mathbb{E}_\text{\tiny R} \left[  V_{ \text{\tiny HT}}(X) \right]. \label{eq_variance_HT_law_total_variance}   
\end{align}

\subsubsection{Properties of the completely oracle IPSW}

The following result establishes consistency and finite sample bias and variance for the oracle IPSW, which extends the preceding results from \cite{Egami2021CovariateSelection} (see their appendix, Section SM-2).
\begin{theorem}[Properties of the completely oracle IPSW]\label{thm:completely-oracle}

Under the general setting defined in Subsection~\ref{subsec:model}, granting Assumptions~\ref{a:repres-rct}-\ref{a:pos}, the completely oracle IPSW is unbiased and has an explicit variance expression, that is, for all $n$, 
\begin{align*}
\mathbb{E}\left[ \hat \tau_{\pi, \text{\tiny T}, \text{\tiny R}, n}^*\right] &=  \tau,\quad \text{and} \quad \operatorname{Var}\left[\hat \tau_{\pi, \text{\tiny T,R}, n}^* \right] =\frac{V_o}{n}, \quad  \text{where}\quad
V_o:=  \operatorname{Var}_\text{\tiny R}\left[ \frac{p_\text{\tiny T}(X)}{p_\text{\tiny R}(X)} \tau(X) \right] +  \mathbb{E}_\text{\tiny R}\left[ \left(\frac{p_\text{\tiny T}(X)}{p_\text{\tiny R}(X)} \right)^2V_{ \text{\tiny HT}}(X)\right].
\end{align*}
\end{theorem}
The proof, to be found in Appendix~\ref{proof:completely-oracle}, sheds light on the technical tools used for more complex IPSW variants, and also establishes the quadratic risk and the consistency of the completely oracle IPSW.
The finite-sample  variance $V_o$ depends on the probability ratio, the amplitude of the heterogeneity of treatment effect (through $\tau(x)$), and variances of the potential outcomes. In particular, if for some category $x$, the values $p_{\text{\tiny T}}(x)$ and $p_{\text{\tiny R}}(x)$ are very different, a large variance will be obtained when generalizing the trial's findings. Note that the variance converges at the classical rate $1/n$. 
Although it is not our main contribution, Theorem~\ref{thm:completely-oracle} is of primary importance for comparing the impact of sample sizes on the performances of the different IPSW variants. Note that, if there is no distribution shift between the two population (that is $p_\text{\tiny T}=p_\text{\tiny R}$), then the variance $V_o$ is simply the Horvitz-Thomson variance, as established in equation~\eqref{eq_variance_HT_law_total_variance}. 

\subsubsection{Properties of the semi-oracle IPSW}

In this section, we study the behaviour of the semi-oracle IPSW (Definition~\ref{def:ipsw-semi-oracle}), for which the probability $p_{\text{\tiny T}}$ is known but the probability $p_{\text{\tiny R}}$ is estimated. One can obtain for a certain $x$, $\hat p_{\text{\tiny R},n}(x) = 0$ for some $x \in \mathbb{X}$, even if the true probability is non-negative $p_{\text{\tiny R}}(x)>0$. This phenomenon, occurring when no observations in the trial correspond to the covariate vector $x$, induces a finite sample bias of the IPSW estimate. 


\begin{proposition}
\label{prop_semi_oracle_bias_variance}
Under the general setting defined in Subsection~\ref{subsec:model}, granting Assumptions~\ref{a:repres-rct}-\ref{a:pos}, the bias of the semi-oracle IPSW satisfies, for all $n$,  
\begin{align*}
\left| \mathbb{E}\left[\hat \tau_{\pi,\text{\tiny T}, n}^*  \right] - \tau \right| \leq \left(1 - \min_x p_{\text{\tiny R}}(x)\right)^n \mathbb{E}_{\text{\tiny T}} \left[ \left| \tau(X) \right| \right].
\end{align*}
Moreover, the variance of the semi-oracle IPSW satisfies, for all $n$, 
\begin{align*}
\operatorname{Var}\left[   \hat \tau_{\pi, \text{\tiny T}, n}^* \right] \le &\,\frac{2 V_{so}}{n+1}  +  \left( 1 - \min_{x \in \mathbb{X}} p_\text{\tiny R}(x)\right)^n \left( \mathbb{E}_{\text{\tiny T}} \left[ |\tau(X)| \right]\right)^2, 
\end{align*}
with
\begin{align*}
    V_{\text{so}}:= &\,\mathbb{E}_\text{\tiny R}\left[ \left(\frac{p_\text{\tiny T}(X)}{p_\text{\tiny R}(X)} \right)^2V_{ \text{\tiny HT}}(X)\right].
\end{align*}
\end{proposition}

Closed-form expressions of bias and variance are derived in \Cref{prop_app_semioracleipsw} in Appendix~\ref{proof:explicit-bias-variance-and-bounds-semi-oracle}, with the proof of Proposition~\ref{prop_semi_oracle_bias_variance}.
Unlike the completely oracle IPSW, the semi-oracle IPSW is biased for small trials (\textit{i.e.} small $n$), which can be understood by undercoverage of some categories in the trial. Indeed, for small trials, the probability that a category is not represented at all in the RCT may not be negligible. Fortunately, as shown in Proposition~\ref{prop_semi_oracle_bias_variance}, this bias converges to zero exponentially with the trial size $n$. 
Corollary~\ref{cor_asympt_semi_oracle} provides asymptotic results for the bias and variance of the semi-oracle IPSW. 

\begin{corollary}[Asymptotics]\label{cor_asympt_semi_oracle}
Under the same assumptions as in Proposition~\ref{prop_semi_oracle_bias_variance}, the semi-oracle IPSW is asymptotically unbiased, and its limiting variance satisfies, 
\begin{equation*}
    \lim_{n\to\infty} \mathbb{E}\left[   \hat \tau_{\pi,\text{\tiny T}, n}^* \right] = \tau, 
\,\quad \text{and}\quad   \lim_{n\to\infty} n\operatorname{Var}\left[   \hat \tau_{\pi, \text{\tiny T}, n}^* \right] = V_{\text{so}}.
\end{equation*}
\end{corollary}
The proof is detailed in Subsection~\ref{proof:asympt-bias-variance-semi-oracle}.
The quantity $V_{so}$ already exists in the literature, for example in \cite{Rothman2000ModernEpidemiology}, where a form of semi-oracle IPSW was introduced under the name \textit{standardization}. Here, we clarify the fact that this formula is valid only for large sample and we provide detailed derivations (see \Cref{prop_app_semioracleipsw} in Appendix~\ref{proof:explicit-bias-variance-and-bounds-semi-oracle}). Therefore, Corollary~\ref{cor_asympt_semi_oracle} is the first theoretical result establishing the limiting variance of the semi-oracle IPSW. 
One can observe from the explicit derivations that the semi-oracle estimator $\hat \tau_{\pi,\text{\tiny T}, n}^*$  has a lower asymptotic variance than the oracle IPSW  $\hat \tau_{\pi, \text{\tiny T,R}, n}^*$ recalled in Theorem~\ref{prop_semi_oracle_bias_variance}. In particular,
    \begin{equation*}
 V_{so} = V_o - \tikzmarknode{amp}{\highlight{Bittersweet}{\color{black} $\operatorname{Var}_\text{\tiny R}\left[ \frac{p_\text{\tiny T}(X)}{p_\text{\tiny R}(X)} \tau(X) \right]$  }}. 
\end{equation*}
\vspace*{0.5\baselineskip}
\begin{tikzpicture}[overlay,remember picture,>=stealth,nodes={align=left,inner ysep=1pt},<-]
\path (amp.north) ++ (-1.1,-2.8em) node[anchor=north west,color=Bittersweet] (sotext){\textsf{\footnotesize Always positive}};
\end{tikzpicture}

This phenomenon has similar explanations\footnote{In fact, similar considerations appear outside causal inference, for example \cite{Efron1978ObservedVsExpected} argued that the observed information rather than the expected Fisher information should be used to characterize the distribution of maximum-likelihood estimates.} with the common (and often surprising) result stating that an estimated propensity score lowers the variance when re-weighting observational data compared to an estimator relying on oracle propensity score \citep[see][regarding IPW estimator]{Robins1992EstimatingEE, Hahn1998efficiencybound, Hirano2003Efficient, Lunceford04stratificationand}. 
Intuitively, 
we only need to generalize from the actual sample to the target population, and not from a source trial population to a target population.\\ 

The semi-oracle estimate has a lower limiting variance compared to the completely oracle IPSW  but is also biased. One can thus wonder how the risk of the two estimates compare. Theorem~\ref{thm:semi-oracle} upper bounds the risk of the semi-oracle estimate.

\begin{theorem}[Properties of the semi-oracle IPSW]\label{thm:semi-oracle}
Under the general setting defined in Subsection~\ref{subsec:model}, granting Assumptions~\ref{a:repres-rct}-\ref{a:pos}, the quadratic risk of the semi-oracle IPSW satisfies,
\begin{equation*}
\mathbb{E}\left[ \left( \hat \tau_{\pi,\text{\tiny T}, n}^* - \tau \right)^2\right] \,\le\, \frac{2 V_{so}}{n+1} \,+\, 2 \left(1 - \min_x  p_\text{\tiny R} \right)^{n} \left( \mathbb{E}_{\text{\tiny T}} \left[ |\tau(X)| \right]\right)^2.
\end{equation*}
\end{theorem}

Appendix~\ref{proof:risk-and-bound-semi-oracle} details the proof.
The second term in the upper bound of Theorem~\ref{thm:semi-oracle} decreases exponentially with $n$, whereas the first term decreases at rate $1/n$. At first, it is not easy to compare this upper bound to the risk of the completely oracle IPSW, due to the factor two before $V_{so}$. Close inspection of the proof of Theorem~\ref{thm:semi-oracle} reveals that the factor $2$ can be replaced by $(1+\varepsilon)$, for all $\varepsilon$, assuming that $n$ is large enough (see Lemma~\ref{lemma:ineq-binomial-pi-hat}). The bound presented here is valid for all $n$ and can be improved if $n$ is taken large enough. Therefore, for all $n$ large enough, the first term in the upper bound is close to $V_{so}/(n+1)$ which is smaller than $V_o/(n+1)$ (see above), which makes the risk of the semi-oracle smaller than that of the completely oracle, for $n$ large enough. This bound 
opens the doors to guarantees even on small sample size. 
Also note that unlike $V_o$, $V_{so}$ can be estimated with the data.

\subsubsection{Properties of the (estimated) IPSW}

Previous results on IPSW are valid when the size of the target population goes to infinity. 
In this subsection, we establish theoretical guarantees for the estimated IPSW in a more complex setting: we consider finite trial and target population datasets and establish bounds depending on both sample sizes ($n$ and $m$). 

\begin{proposition}
\label{prop_completely_estimated}
Under the general setting defined in Subsection~\ref{subsec:model}, granting Assumptions~\ref{a:repres-rct}-\ref{a:pos}, the bias of the estimated IPSW satisfies, for all $n, m$,  
\begin{align*}
  \left| \mathbb{E}\left[\hat \tau_{\pi, n,m}  \right] - \tau \right| & \leq \left(1 - \min_x p_{\text{\tiny R}}(x)\right)^n \mathbb{E}_{\text{\tiny T}} \left[ \left| \tau(X) \right| \right].
\end{align*}
Moreover, the variance of the estimated IPSW satisfies, for all $n,m$, 

\begin{align}
     \operatorname{Var}\left[   \hat \tau_{\pi, n,m} \right]  &\le  \frac{2V_{so}}{n+1}  +  \frac{ \operatorname{Var}_{\text{\tiny T}} \left[ \tau(X)\right]}{m}  +  \frac{2}{m\left(n+1\right)}\mathbb{E}_{\text{\tiny R}}\left[ \frac{p_\text{\tiny T}\left( X \right)(1-p_\text{\tiny T}\left( X \right))}{p_\text{\tiny R}\left( X \right)^2} V_{\text{\tiny HT}}(X)\right]   \nonumber \\
     & \qquad +  \left(1 - \min_x p_{\text{\tiny R}}(x) \right)^{n/2}\mathbb{E}_{\text{\tiny T}} \left[   \tau(X)^2 \right]\left( 1 +  \frac{4}{m} \right).
     \label{eq:thm_completely_estimated_upper_bound_variance}
\end{align}
\end{proposition}

Closed-form expressions of bias and variance as well as proofs are derived in \Cref{prop:appendix_soIPSW} in Appendix~\ref{proof_prop_completely_estimated}.
\Cref{prop_completely_estimated} is the first result to upper bound the bias and variance of the estimated IPSW in a finite-sample setting. 
A first observation is that the bias of the (estimated) IPSW is the same as that of the semi-oracle, showing that only a limited trial sample size can explain a finite sample bias (see \Cref{prop:appendix_soIPSW} for details).
On the other side, the variance terms differ, due to the additional estimation of the target probability $p_{\text{\tiny T}}$ in the estimated IPSW. 
Therefore, all additional terms compared to the variance of the semi-oracle $\hat \tau_{\text{\tiny T},\pi, n}$ depend on $m$.
The variance upper bound tends to zero as $n$ and $m$ go to infinity. In this setting, the variance is dominated by the first two terms in inequality~\eqref{eq:thm_completely_estimated_upper_bound_variance}. If $m \gg n$, the variance is dominated by the first term, which is the dominant term of the semi-oracle variance. Following this idea,  Corollary~\ref{cor_asympt_completely_estimated} establishes the asymptotic bias and variance of the estimated IPSW in different sample size regimes.

\begin{corollary}\label{cor_asympt_completely_estimated}
Under the same assumptions as in Proposition~\ref{prop_completely_estimated}, the estimated IPSW is asymptotically unbiased when $n$ tends to infinity, that is 
\begin{equation*}
    \lim_{n\to\infty} \mathbb{E}\left[   \hat \tau_{\pi, n, m} \right] = \tau. 
\end{equation*}
Besides, letting $ \lim\limits_{n,m\to\infty} m/n = \lambda \in [0,\infty]$, the limiting variance of the estimated IPSW satisfies
\begin{equation*}
 \lim\limits_{n,m\to\infty} \min(n,m) \operatorname{Var}\left[   \hat \tau_{\pi, n,m} \right] = \min(1, \lambda) \left( \frac{\operatorname{Var}\left[ \tau(X) \right]}{\lambda} + V_{so} \right).
\end{equation*}
\end{corollary}

\begin{wrapfigure}[18]{r}{0.3\textwidth} 
    \centering
    \includegraphics[width=0.3\textwidth]{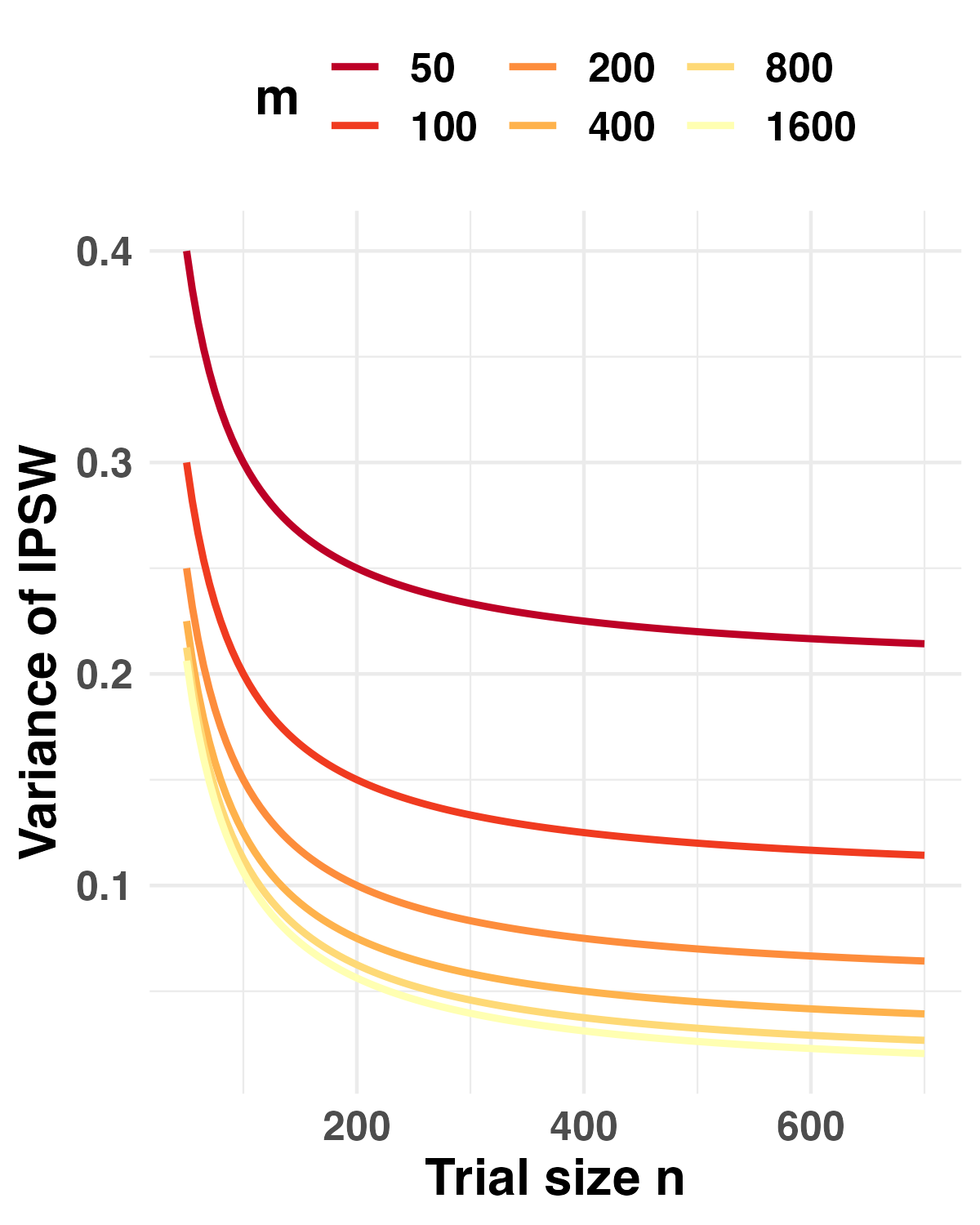}
    \caption{\textbf{Illustration of Corollary~\ref{cor_asympt_completely_estimated}}}
    \label{fig:corollary3}
\end{wrapfigure} 
A proof is detailed in Subsection~\ref{proof_asympt_completely_estimated}.\\

As highlighted in Corollary~\ref{cor_asympt_completely_estimated}, there is not a unique limiting variance for the estimated IPSW. Its limiting variance depends on how the sample sizes $n$ and $m$ compare to each other asymptotically. 
For example, 
\begin{itemize}
    \item If $m/n \to \infty$, (i.e., $\lambda=\infty$) then the limiting variance of the estimated IPSW corresponds to the semi-oracle's one;
    \item If we consider an asymptotic regime where the observational sample is about ten times bigger than the trial ($\lambda = 10$), then the asymptotic variance is equal to $\lim\limits_{n,m\to\infty} n \operatorname{Var}\left[   \hat \tau_{\pi, n,m} \right] = \operatorname{Var}\left[ \tau(X) \right]/10 + V_{so} >  V_{so} $;
    \item Finally, if $m/n \to 0$, (i.e., $\lambda=0$) then the limiting variance of the estimated IPSW has no more link to that of the semi-oracle IPSW, and $  \lim\limits_{n,m\to\infty} m \operatorname{Var}\left[   \hat \tau_{\pi, n,m} \right] = \operatorname{Var}\left[ \tau(X) \right]$.
\end{itemize}

This formula can be used to guide data collection. For example, and using the formula, one could say that at some point gathering $N$ additional individuals information in the target population (which has a cost) could lead to less gain in precision than gathering a bit more data on the trial (if possible). This phenomenon is illustrated on Figure~\ref{fig:corollary3}. \\

Upper bound on the risk of the estimated IPSW can be established, based on Proposition~\ref{prop_completely_estimated}.

\begin{theorem}[Properties of the IPSW]\label{thm:ipsw}
Under the general setting defined in Subsection~\ref{subsec:model}, granting Assumptions~\ref{a:repres-rct}-\ref{a:pos}, the quadratic risk of the estimated IPSW satisfies,
 \begin{align}
\mathbb{E}\left[ \left( \hat \tau_{\pi, n, m} - \tau \right)^2\right] & \le \frac{2V_{so}}{n+1} + \frac{\operatorname{Var}\left[\tau(X) \right]}{m} + \frac{2}{m(n+1)}  \mathbb{E}_{\text{\tiny R}}\left[ \frac{p_\text{\tiny T}\left( X \right)(1-p_\text{\tiny T}\left( X \right))}{p_\text{\tiny R}\left( X \right)^2} V_{\text{\tiny HT}}(X)\right] \nonumber \\
& \quad  + 2 \left(1 - \min_x p_{\text{\tiny R}}\left(x\right)\right)^{n} \mathbb{E}_{\text{\tiny T}}[\tau(X)^2] \left( 1 + \frac{2}{m} \right). \label{eq_theorem_completely_est}
\end{align}
\end{theorem}

Proof is detailed in Subsection~\ref{proof_thm_ipsw}.
The first and fourth terms in inequality~\eqref{eq_theorem_completely_est} correspond to the bound of the semi-oracle estimator (see Theorem~\ref{thm:semi-oracle}). Following the intuition, the bound on the risk of the estimated IPSW is larger than the one of the semi-oracle. This is due to the cost of estimating $p_{\text{\tiny T}}$ from a finite sample of size $m$. However, when $m \gg n$, the dominant terms in the risk of the estimated and semi-oracle IPSW are the same. 
Consistency of the (estimated) IPSW for continuous covariates has been proven in the literature, e.g.  \cite{buchanan2018generalizing} demonstrate consistency and asymptotic normality under a nested-design and assuming a parametric selection process. \cite{Colnet2021Sensitivity} demonstrate consistency assuming uniform convergence of the probability ratio under a cross-fitting procedure and no parametric assumption. Our results are the first to establish the bias and the variance of the estimated IPSW in finite and asymptotic regimes, with an explicit dependence on both sample sizes. 

\paragraph{What if the probability to be treated depends on $x$?} In some trials, the probability to receive treatment depends on the strata (e.g. for ethical reason). If so, all the previous results are kept unchanged, replacing $\pi$ by $\pi(x)$, and the proofs are written with $\pi(x)$, even if the main results are reported with a constant $\pi$ for briefness. In particular, all the covariates used to stratify the propensity to receive treatment in the trial should be used in the IPSW.
\subsection{Estimating the probability to be treated in the trial?}\label{subsec:also-estimating-pi}
So far, we have considered an estimation procedure where $\pi$, the probability to be treated in the trial, is plugged in the formula.
Still, one may want to estimate it for the purpose of precision. 
This idea follows the same spirit of what can be done with the Horvitz-Thomson (Definition~\ref{def:HT}) and the Difference-in-means (Definition~\ref{def:difference-in-means}), where the large-sample gain in variance is recalled in Equation~\eqref{eq:ineq-dm-ht}. 
To our knowledge, different version of IPSW are currently present in the literature, with or without an estimated $\pi$ (see Table~\ref{tab:non-exhaustive-review-IPSW} in appendix for a non-exhaustive review). 
In our work, we propose to estimate $\pi$ \underline{per strata}, and then adapt the semi-oracle IPSW (Definition~\ref{def:ipsw-semi-oracle}) and the estimated IPSW (Definition~\ref{def:ipsw}).


\begin{definition}[Estimation of $\hat \pi$ for each strata]\label{def:procedure-for-densities-and-pi}

Under the setting defined in Subsection~\ref{subsec:model}, 
\begin{equation*}
   \forall x\in \mathds{X},\, \hat \pi_n(x) =  \frac{\sum_{i \in \mathcal{R}} \mathbbm{1}_{X_i = x} \mathbbm{1}_{A_i = 1}}{\sum_{i \in \mathcal{R}} \mathbbm{1}_{X_i = x}}.
\end{equation*}

\end{definition}

Strange as it may seem, estimating $\pi$ per strata and not on the whole sample  can also be beneficial in RCTs to improve precision. \cite{imbens2011experimental,miratrix2013adjusting} introduce the post-stratification procedure, a technique aiming to use covariate information for precision when estimating the ATE from a single trial. These two research works detail why a so-called post-stratification estimator yields a lower variance compared to the Difference-in-Means -- and therefore a Horvitz-Thomson -- as soon as the covariates used for stratification are predictive of the outcome. More particularly, the post-stratification estimator on a single trial is defined as follows.

\begin{definition}[Post-stratification - \cite{imbens2011experimental, miratrix2013adjusting}]\label{def:post-stratification-estimator}

The post-stratification estimator is denoted $\hat \tau_{\text{\tiny PS},n}$ and defined as,
\begin{align*}
      \hat \tau_{\text{\tiny PS},n} &= \frac{1}{n} \sum_{i=1}^n \frac{A_iY_i}{\hat \pi_n(x)} -  \frac{(1-A_i)Y_i}{1-\hat \pi_n(x)},
\end{align*}
where $\pi$ is estimated according to Definition~\ref{def:procedure-for-densities-and-pi}.
\end{definition}
The different displays of the post-stratification estimator $\hat \tau_{\text{\tiny PS},n}$ in literature are recalled in Section~\ref{appendix:useful-results-rct}. The gain in efficiency of an IPSW version with estimated $\pi$ follows this intuition.

\begin{definition}[Semi-oracle IPSW with $\hat \pi$]\label{def:ipsw-semi-oraclewith-pi} The semi-oracle IPSW estimator $\hat \tau_{\text{\tiny T}, n}^*$ with estimated propensity scores $\hat{\pi}_n$ is defined as 
\begin{equation}
\hat \tau_{\text{\tiny T}, n}^*= \frac{1}{n} \sum_{i=1}^{n} \frac{p_{\text{\tiny T}}(X_i)}{\hat p_{\text{\tiny R}, n}(X_i)}  Y_i \left( \frac{A_i}{\hat \pi_n(X_i)} - \frac{1-A_i}{1-\hat \pi_n(X_i)} \right)\,,
\end{equation}
with $\hat p_{\text{\tiny R}, n}(x)$ and $\hat \pi_n(x)$ defined in Definitions~\ref{def:procedure-for-densities} and  \ref{def:procedure-for-densities-and-pi}.
\end{definition}

\begin{definition}[IPSW with $\hat \pi$]\label{def:ipsw-with-pi} The completely-estimated IPSW estimator $\hat \tau_{n,m}$ with estimated propensity scores $\hat{\pi}_n$ is defined as 
\begin{equation}
\hat \tau_{n, m}= \frac{1}{n} \sum_{i=1}^{n} \frac{\hat p_{\text{\tiny T}, m}(X_i)}{\hat p_{\text{\tiny R}, n}(X_i)} Y_i \left( \frac{A_i}{\hat \pi_n(X_i)} - \frac{1-A_i}{1-\hat \pi_n(X_i)} \right)\,,
\end{equation}
where $\hat p_{\text{\tiny T}, m}(x)$, $\hat p_{\text{\tiny R}, n}(x)$, and $\hat \pi_n(x)$  defined in Definitions~\ref{def:procedure-for-densities} and  \ref{def:procedure-for-densities-and-pi}.
\end{definition}

Before stating the formal results, and following the spirit of what was done with the variance of the Horvitz-Thomson per strata \eqref{eq:conditional-VHT}, we introduce $V_{ \text{\tiny DM},n}(x)$:
\begin{align}
\label{eq:conditional-VDM}
V_{\text{\tiny DM},n}(x)
 =  \frac{\mathds{1}_{Z_n(x) >0}}{Z_n(x)} \operatorname{Var} \left[    \sum_{i=1}^{n} \mathds{1}_{X_i=x}   \left( \frac{A_i Y_i^{(1)}}{\hat \pi_n(x)} - \frac{(1-A_i)Y_i^{(0)}}{1-\hat \pi_n(x)} \right) \mid \mathbf{X}_{n} \right].
\end{align}
The explicit variance of the Difference-in-Means under a Bernoulli design is provided in Appendix (see Lemma~\ref{lemma:DM-bias-and-variance}), and not displayed here for conciseness.


\begin{proposition}[IPSW's properties when also estimating $\pi$]\label{prop:when-estimating-pi}
Under the general setting defined in Subsection~\ref{subsec:model}, granting Assumptions~\ref{a:repres-rct}-\ref{a:pos}, the bias of the estimated IPSW with estimated $\hat \pi_n$ (see Definition~\ref{def:procedure-for-densities-and-pi}) satisfies, for all $n$, 
\begin{align*}
    \left| \mathbb{E} \left[ \hat \tau_{n, m} \right]- \tau \right| 
    &  \leq \left(  1- \min_x \left( (1 - \tilde \pi(x)) p_{\text{\tiny R}}(x) \right) \right)^n \left(   \mathbb{E}_\text{\tiny T} \left[ |  Y^{(1)}   | \right]+  \mathbb{E}_\text{\tiny T} \left[| Y^{(0)}  | \right]  \right),
\end{align*}
where $\tilde \pi(x) = \max ( \pi(x), 1 - \pi(x))$.
Besides, the variance of the estimated IPSW with estimated $\hat \pi_n$ satisfies, for all $n$
\begin{align*}
\operatorname{Var}\left[    \hat \tau_{n, m} \right]  \leq &\frac{2\,  \tilde V_{so}}{n+1}   + \frac{\operatorname{Var}\left[   \tau(X) \right] }{m}   +  \frac{2}{(n+1)m}  \mathbb{E}_{\text{\tiny R}}\left[ \frac{p_\text{\tiny T}\left( X \right)(1-p_\text{\tiny T}\left( X \right))}{p_\text{\tiny R}\left( X \right)^2} V_{\text{\tiny DM}, n}(X)\right] \\
&\quad + 2 \left( 1 + \frac{3}{m} \right) \left( 1 - \min_x \left( (1 - \tilde \pi(x)^2) p_{\text{\tiny R}}(x) \right) \right)^{n/2} \mathbb{E} \left[ (Y^{(1)})^2 + (Y^{(0)})^2 \right],
\end{align*}
where
\begin{align*}
\tilde V_{so}:=\mathbb{E}_{\text{\tiny R}}\left[ \left( \frac{p_\text{\tiny T}\left( X \right)}{p_\text{\tiny R}\left( X \right)}\right)^2 V_{\text{\tiny DM},n}(X)\right].
\end{align*}
\end{proposition}

Proof is detailed in Subsection~\ref{proof:cor-when-estimating-pi}. 
The bound on the variance of $\hat \tau_{n,m}$ is very close to the one of $\hat \tau_{\pi,n,m}$, and in particular for any fixed $m$,
\begin{align*}
    \operatorname{Var}\left[    \hat \tau_{n, m} \right]  & \leq \frac{2\,  \tilde V_{so}}{n+1}   + \frac{\operatorname{Var}\left[   \tau(X) \right] }{m}   +  \frac{2}{(n+1)m}  \mathbb{E}_{\text{\tiny R}}\left[ \frac{p_\text{\tiny T}\left( X \right)(1-p_\text{\tiny T}\left( X \right))}{p_\text{\tiny R}\left( X \right)^2} V_{\text{\tiny DM}, n}(X)\right] + o\left(\frac{1}{n}\right),
\end{align*}
where the main difference comes from $\tilde V_{so}$ that contains $V_{\text{\tiny DM},n}(X)$ rather than $V_{\text{\tiny HT}}(X)$. 
According to Lemma~\ref{lem_inequality_vdm_vht} in Appendix~\ref{proof:variance-inequality}, we have, for all $x$, 
\begin{align*}
& ~~ \mathbb{E} \left[ V_{\text{\tiny DM},n}(x) \mathds{1}_{Z_n(x) >0}\right]  \leq V_{ \text{\tiny HT}}(x) - \alpha(x)^2 + O\left(n^{-1/4}\right),
\end{align*}
which allows to conclude that for all $n$ large enough, the bound on the variance of $\hat \tau_{n,m}$ is tighter than the bound on the variance of $\hat \tau_{\pi,n,m}$. This can also be observed on the large sample variance.


\begin{corollary}
\label{cor_asympt_completely_estimated_pi_estimated}
Under the same assumptions as in Proposition~\ref{prop:when-estimating-pi}, the completely estimated IPSW is asymptotically unbiased when $n$ tends to infinity, that is 
\begin{equation*}
    \lim_{n\to\infty} \mathbb{E}\left[   \hat \tau_{n, m} \right] = \tau. 
\end{equation*}
Besides, letting $ \lim\limits_{n,m\to\infty} m/n = \lambda \in [0,\infty]$, the limiting variance of completely estimated IPSW satisfies
\begin{align*}
 \lim\limits_{n,m\to\infty} \min(n,m) \operatorname{Var}\left[   \hat \tau_{n,m} \right] = \min(1, \lambda) \left( \frac{\operatorname{Var}\left[ \tau(X) \right]}{\lambda} +  \tilde V_{so, \infty} \right),
\end{align*}
\begin{flalign*}
\text{where} &&
    \tilde V_{so, \infty}:= &\,\mathbb{E}_{\text{\tiny R}}\left[ \left( \frac{p_\text{\tiny T}\left( X \right)}{p_\text{\tiny R}\left( X \right)}\right)^2 V_{\text{\tiny DM},\infty}(X)\right],
&&\\
\text{and} &&
    V_{\text{\tiny DM}, \infty} (x) := & \,\frac{\operatorname{Var}_{\text{\tiny R}}\left[ Y^{(1)} \mid X = x \right]}{\pi} + \frac{\operatorname{Var}_{\text{\tiny R}}\left[ Y^{(0)} \mid X = x \right]}{1-\pi}.
&&
\end{flalign*}

\end{corollary}

Proof is detailed in Subsection~\ref{proof:cor_asympt_completely_estimated_pi_estimated}. 
Because $\forall x \in \mathds{X}$, $V_{ \text{\tiny DM},\infty}(x) \le V_{ \text{\tiny HT}}(x)$, then $  \tilde V_{so, \infty} \le V_{so} $, so that the large sample variance of the semi-oracle and completely estimated IPSW are smaller than with an oracle $\pi$, regardless of the regime at which $n$ and $m$ tend to infinity. 
Similarly to the result on $\hat \tau_{\pi,n,m}$, upper bound on the risk of the completely estimated IPSW can be established, based on Proposition~\ref{prop:when-estimating-pi}.

\begin{theorem}[Properties of the IPSW]\label{thm:ipsw_pi_est}
Under the general setting defined in Subsection~\ref{subsec:model}, granting Assumptions~\ref{a:repres-rct}-\ref{a:pos}, the quadratic risk of the completely estimated IPSW with estimated $\hat \pi$ satisfies,
 \begin{align}
\mathbb{E}\left[ \left( \hat \tau_{n, m} - \tau \right)^2\right] & \leq  \frac{2 \tilde V_{so}}{n+1}   + \frac{\operatorname{Var}\left[   \tau(X) \right] }{m} + \frac{2}{m(n+1)}  \mathbb{E}_{\text{\tiny R}}\left[ \frac{p_\text{\tiny T}\left( X \right)(1-p_\text{\tiny T}\left( X \right))}{p_\text{\tiny R}\left( X \right)^2} V_{\text{\tiny DM}, n}(X)\right]  \nonumber \\
& \quad   +  2 \left( 2 + \frac{3}{m} \right) \left( 1 - \min_x \left( (1 - \tilde \pi(x)) p_\text{\tiny R}(x) \right) \right)^{n/2} \mathbb{E} \left[ (Y^{(1)})^2 + (Y^{(0)})^2 \right]. \label{eq_theorem_completely_est_pi_est}
\end{align}
\end{theorem}
Proof is detailed in Subsection~\ref{proof:thm:ipsw_pi_est}.
For the risk, and for the same arguments than for the bound on the variance, it can be shown that for a reasonable $n$, the bound on the risk of $ \hat \tau_{n, m} $ is tighter than for $ \hat \tau_{\pi, n, m}$.
All the previous results establish theoretical guidance explaining why an estimator also estimating $\pi$ per strata should be preferred in practice, at least for some trial sample size $n$. To our knowledge, we have not found work explicitly stating that estimating $\pi$ in the IPSW should be preferred, even if \cite{dahabreh2020extending} uses a logistic regression to estimate the propensity to receive treatment in the trial. 

\subsection{Extended adjustment set: when using extra covariates}
\label{subsec:extended_adjustement_set}

In this section, we detail the impact of adding covariates that are not necessary for adjustment -- for example being only shifted or only treatment effect modifiers -- on the IPSW performances. 
Indeed, in the literature, one of the natural approach is to adjust on all shifted covariates, also named the \textit{sampling set} \citep{cole2010generalizing, tipton2013improving}. 
Another adjustment set is also possible, being the \textit{heterogeneity set} comprising all the treatment effect modifiers \citep{hartman2021inbook}, even if, knowing which covariate is treatment effect modifier is harder. 
As mentioned in the related work (Subsection~\ref{subsec:related-work}), there is an important literature about optimal adjustment set for precision in the causal inference literature, but to our knowledge the topic has not been tackled yet when it comes to efficiency in generalization. \cite{Egami2021CovariateSelection} discuss extensively the usage of these two sets for identification but do not study their impact on the asymptotic variance.\\

In this section the theoretical results hold for a specific regime, where the target sample is bigger than the trial sample, that is $m \gg n$. In other word, this situation is equivalent as considering the semi-oracle IPSW with estimated $\pi$ (Definition~\ref{def:ipsw-semi-oraclewith-pi}).

\paragraph{Formalization}

Consider that the user has at disposal an external set of baseline categorical covariates denoted $V$. We assume that Assumptions~\ref{a:cate-indep-s-knowing-X} and \ref{a:pos} are preserved when adding $V$ to the adjustment set $X$ previously considered\footnote{Note that if preserving transportability is pretty straitghforward as $V$ is a baseline covariate too (for e.g. no collider bias), the support inclusion's assumption can be more challenging when adding too many covariates (see \cite{Damour2017JournalOfEconometrics} for a discussion).}. As mentioned above, this external covariates set can be of two different natures.


\begin{definition}[$V$ is not a treatment effect modifier]\label{def:V-is-not-treat-effect-modifier} $V$ does not modulate treatment effect modifier, that is 
\begin{equation*} 
\forall v \in \mathds{V},\; \forall s \in \{ T, R\},\quad \mathbb{P}_{\text{s}}(Y^{(1)} - Y^{(0)} \mid X=x, V=v) = \mathbb{P}_{\text{s}}(Y^{(1)} - Y^{(0)} \mid X=x).
\end{equation*}
\end{definition}

\begin{definition}[$V$ is not shifted]\label{def:V-is-not-shifted} $V$ is not shifted, that is
\begin{equation*}
     \forall v \in \mathds{V}, \quad p_{\text{\tiny T}}(v) = p_{\text{\tiny R}}(v).
\end{equation*}
\end{definition}

To distinguish estimator using the set $X$ or the extended set $X,V$, we denote $\hat \tau (X)$ and $\hat \tau (X,V)$ the two estimations strategies. 
One can show that adding only shifted covariates $V$ leads to a loss of precision, when the set $V$ is independent of the set $X$.

\begin{corollary}[Adding shifted and independent covariates]\label{proposition:adding-shifted-covariates}
Consider the semi-oracle IPSW estimator $\hat \tau_{\text{\tiny T},n}^*$ (Definition~\ref{def:ipsw-semi-oraclewith-pi}), and a set of additional shifted covariates $V$ (Definition~\ref{def:V-is-not-treat-effect-modifier}) independent of $X$, which are not treatment effect modifiers. Then,
\begin{equation*}
 \lim_{n\to\infty} n \operatorname{Var}_\text{\tiny R}\left[\hat \tau_{\text{\tiny T},n}^*(X,V) \right] =  \left( \sum_{v \in \mathcal{V}} \frac{p_{\text{\tiny T}}(v)^2}{p_{\text{\tiny R}}(v)} \right)  \lim_{n\to\infty} n \operatorname{Var}_\text{\tiny R}\left[\hat \tau_{\text{\tiny T},n}^*(X)\right].
\end{equation*}
\end{corollary}
Proof is detailed in Subsection~\ref{proof_adding-shifted-covariates}.
This results states that the asymptotic variance of the semi-oracle estimator is always bigger if an additional independent shifted covariate set $V$ is added in the adjustment. Moreover, the stronger the shift, the bigger the variance inflation. 
Note that this specific rule was retrieved in the toy example, where the plain line (corresponding to Corollary~\ref{proposition:adding-shifted-covariates}) matches the empirical dots on Figure~\ref{fig:toy_example_shifted_covariates}.\\

On the contrary, adding an additional treatment effect modifier covariate set leads to a gain in precision. 

\begin{corollary}[Adding non-shifted treatment effect modifiers]\label{proposition:adding-treat-effect-modifier-covariates}

Consider the semi-oracle IPSW estimator $\hat \tau_{\text{\tiny T},n}^*$ (Definition~\ref{def:ipsw-semi-oraclewith-pi}). Consider an additional non-shifted treatment effect modifier set (Definition~\ref{def:V-is-not-shifted})  independent of $X$. Then,
\begin{align*}
\lim_{n\to\infty}  n \operatorname{Var}_\text{\tiny R}\left[\hat \tau_{\text{\tiny T},n}^*(X,V)\right]  \,=\,  \lim_{n\to\infty} n \operatorname{Var}_\text{\tiny R}\left[\hat \tau_{\text{\tiny T},n}^*(X) \right] - 
    \mathbb{E}_\text{\tiny R} \left[ \frac{p_\text{\tiny T}(X)}{p_\text{\tiny R}(X)} \operatorname{Var}\left[ \tau(X,V) \mid X \right] \right].
\end{align*}
\begin{flalign*}
\text{In particular,}&&
\lim_{n\to\infty}  n \operatorname{Var}_\text{\tiny R}\left[\hat \tau_{\text{\tiny T},n}^*(X,V)\right]  \le   \lim_{n\to\infty} n  \operatorname{Var}_\text{\tiny R}\left[ \hat \tau_{\text{\tiny T},n}^*(X) \right].
&&
\end{flalign*}
\end{corollary}

Proof is detailed in Subsection~\ref{proof_adding-treat-effect-modifier-covariates}.
This result follows a similar spirit as \cite{Rotnitzky2020Efficient} due to the comparison of two limiting variances, even though the context and the theoretical tools are different.

\section{Synthetic and semi-synthetic simulations}\label{section:semi-synthetic-simulations}

In this section, one additional analysis based on the toy example is provided to illustrate the different asymptotic regimes from Section~\ref{section:theoretical-results}. In addition, results are also illustrated on a semi-synthetic simulation aiming to mimic a medical scenario. The code to reproduce the simulations and the different figures is available on Github\footnote{\texttt{BenedicteColnet/IPSW-categorical}.}.

\subsection{Synthetic: additional experiment from the toy example}

While most of the results are illustrated at the beginning of the article through the toy example, here we more thoroughly investigate empirically the different asymptotic regimes of the IPSW and its variants. 
In particular we complete Figure~\ref{fig:toy_example_2_asympt} that highlights the phenomenon of different asymptotic regimes, with a complete visualization of risks and variances allowing to more precisely illustrate the theoretical results, and in particular Corollary~\ref{cor_asympt_completely_estimated}. More precisely, the quadratic risk is depicted in Figure~\ref{fig:toy.example.risks}, while the variance via $min(n,m) \operatorname{Var}\left[\hat \tau_{n,m} \right]$ is displayed in Figure~\ref{fig:toy.example.regimes}. In both figures, different estimators (oracle or not) are considered with different regimes for $m$, as $n$ grows to infinity.  
In particular, this simulation confirms that 
\begin{itemize}
    \item[\textit{(i)}] all IPSW variants are consistent, even though their convergence speeds depend on the regime (Figure~\ref{fig:toy.example.risks}),
    
    \item[\textit{(ii)}]  the completely oracle IPSW has a bigger variance than the semi-oracle IPSW (Figure~\ref{fig:toy.example.regimes}),
    
    \item[\textit{(iii)}] the limiting variance depends on the asymptotic regime (Figure~\ref{fig:toy.example.regimes}),
    
    \item[\textit{(iv)}] the completely estimated IPSW reaches the variance of the semi-oracle one if the target population sample is bigger than the trial (Figure~\ref{fig:toy.example.regimes}). 
\end{itemize}   
\begin{figure}[!h]
        \centering
        \begin{subfigure}[b]{0.48\textwidth}
            \centering

            \includegraphics[width=0.98\textwidth]{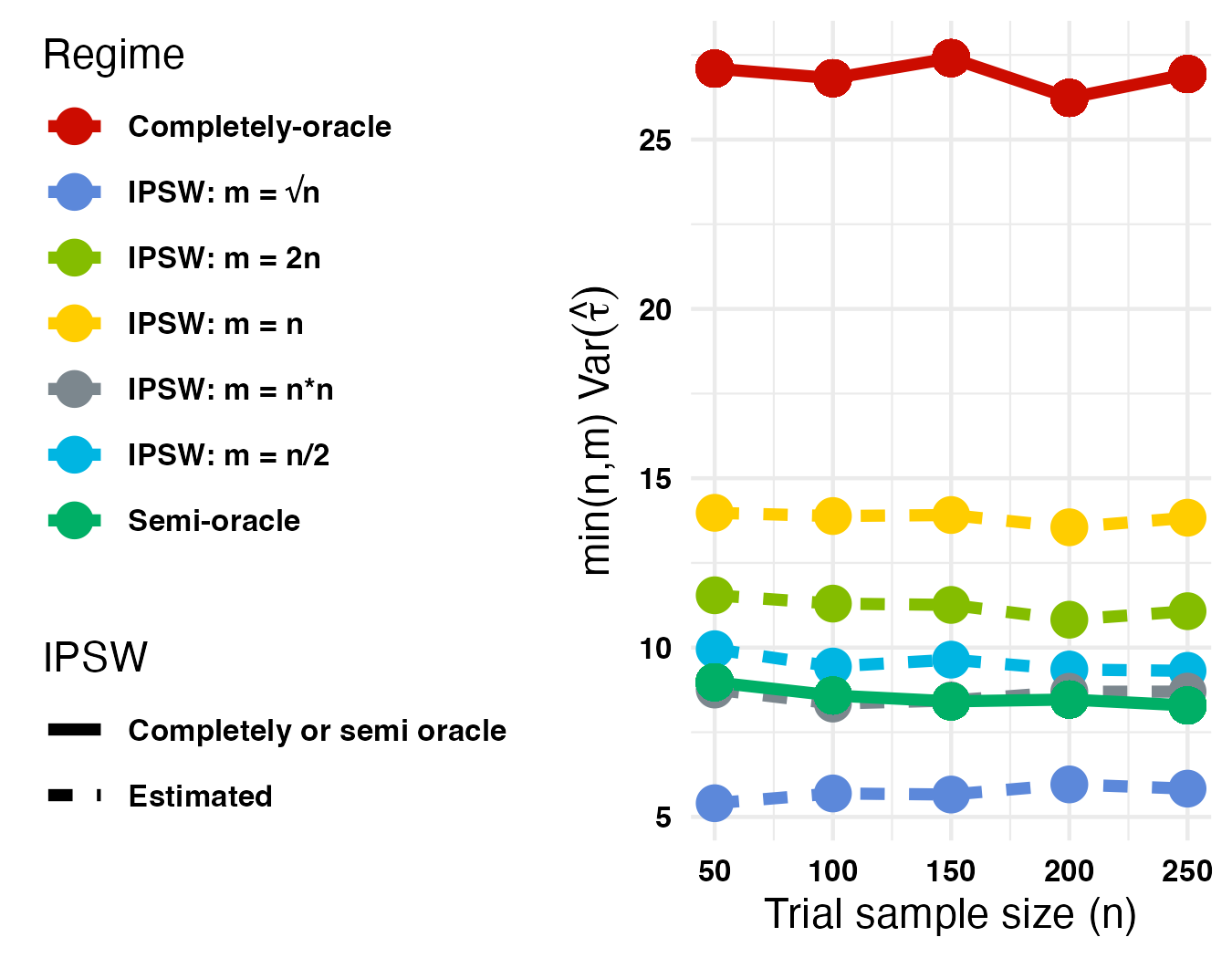}
            \caption[]%
           {{\small \textbf{Asymptotic variance}}}    
           \label{fig:toy.example.regimes}
           
        \end{subfigure}
        \hspace{0.1cm}
        \begin{subfigure}[b]{0.48\textwidth}  
            \centering 
        \includegraphics[width=0.98\textwidth]{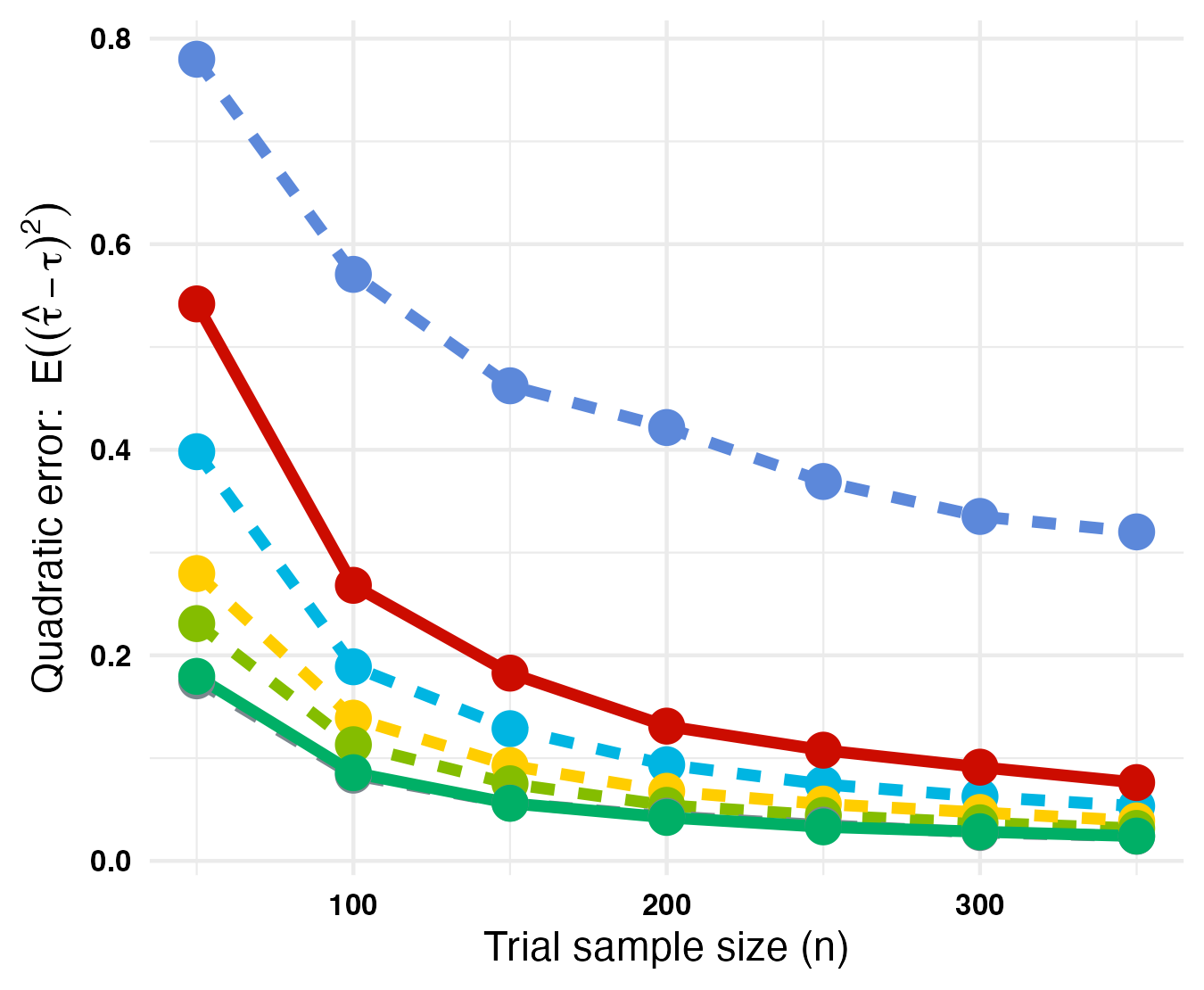}
            \caption[Network2]%
            {{\small \textbf{Quadratic-error or Risk}}}    
             \label{fig:toy.example.risks}
        \end{subfigure}
        \caption{\small \textbf{Risks and different asymptotic regimes}: Based on the toy example simulation (see Section~\ref{sec:toy-example-generalization} and data-generative process from Figure~\ref{fig:toy_example_response_level}) where empirical variance from either the completely oracle (Definition~\ref{def:ipsw-oracle}), the semi-oracle (Definition~\ref{def:ipsw-semi-oracle}) or the estimated IPSW (Definition~\ref{def:ipsw}) are estimated repeating $6,000$ times each simulation for each trial sample size ($x$-axis). Simulations cover different regimes of size $n$ and $m$. On the $y$-axis the empirical variance $min(n,m) \operatorname{Var}\left[\hat \tau_{n,m} \right]$ is plotted (with the exception of $min(n,m)=n$ for completely- and semi- oracle variants, represented in plain lines). Each color represents one specific estimator and regime.}
        \label{fig:stoy.example.risks.regimes}
\end{figure}
\subsection{Semi-synthetic}

In the semi-synthetic simulation, the data are taken from an application in critical care medicine, and only the outcome generative model is simulated, such that the covariate distribution and in particular the distribution shift between populations is inherited from a real situation.

\subsubsection{Design}
Two data-sets are used to generate two sources:
\begin{enumerate}
    \item A randomized controlled trial (RCT), called CRASH-3 \citep{crash3protocol}, aiming to measure the effect of Tranexamic Acide (TXA) to prevent death from Traumatic Brain Injury (TBI). A total of 175 hospitals in 29 different countries participated to the RCT, where adults with TBI suffering from intracranial bleeding were randomly administrated TXA \citep{crash32019}. The inclusion criteria of the trial are patients with a Glasgow Coma Scale (GCS)\footnote{The Glasgow Coma Scale (GCS) is a neurological scale which aims to assess a person's consciousness. The lower the score, the higher the gravity of the trauma.} score of 12 or lower or any intracranial bleeding on CT scan, and no major extracranial bleeding.
    \item An observational cohort, called Traumabase, comprising 23 French Trauma centers, collects detailed clinical data from the scene of the accident to the release from the hospital. The resulting database, called the Traumabase, comprises 23,000 trauma admissions to date, and is continually updated, representing a fair, almost-exhaustive data base about actual individuals taken in charge in France and suffering from trauma.
\end{enumerate}

These two data sources are turned into two source populations representing a real-world situation with six covariates so that the distribution structure and, in particular, the distributional shift mimics a real-world situation. The six covariates kept in common are:  GCS (categorical), gender (categorical), pupil reactivity (categorical), age (continuous), systolic blood pressure (continuous), and time-to-treatment (TTT) (continuous). The continuous covariates are then turned into categories. Additional details about data preparation are available in Appendix (see Section~\ref{appendix:additional-info-semi-synthetic-simulation}). In this semi-synthetic simulation, only the outcome model is completely synthetic, and follows 

\begin{equation}\label{equation:semi-synthetic-outcome-model}
  \small  Y := f(\texttt{GCS},\texttt{Gender})  + A\, \tau(\texttt{TTT}, \texttt{Blood Pressure}) + \epsilon_{\texttt{TTT}},
\end{equation}

where $f$ and $\tau$ are two functions of the covariates, and $\epsilon_{\texttt{TTT}}$ is a gaussian noise such that $\mathbb{E}[\epsilon_{\texttt{TTT}} \mid X] = 0$, but where heteroscedasticity is observed along the covariate $\texttt{TTT}$. The higher the time-to-treatment, the higher $\operatorname{Var}\left[ \epsilon_{\texttt{TTT}} \mid \text{TTT} \right]$, and so the noise on $Y$ (see Section~\ref{appendix:additional-info-semi-synthetic-simulation} for the detailed generated function).
This outcome model is such that only time-to-treatment (TTT) and blood pressure are effect modifiers, while other covariates only affects the baseline value or have no impact. 
Each time a simulation is conducted observations are sampled from the two populations with replacement, and the outcome is created following equation~\eqref{equation:semi-synthetic-outcome-model}. The trial is such that $\pi=0.5$.

\subsubsection{Results}

\paragraph{Minimal adjustment set is sufficient to generalize}
The minimal adjustment set to generalize the trial results is constituted of the time-to-treatment(\texttt{TTT}) and the systolic blood pressure (\texttt{blood}).
Using only these two covariates, the simulations illustrate how the re-weighting procedure allows to correct for the population shift between the trial and the target population as presented on Figure~\ref{fig:semi_synth_pi_hat_or_not} ($1,000$ repetitions). 
\begin{figure}[tbh!]
 \begin{minipage}{.28\linewidth}
  \caption{\textbf{IPSW estimating $\pi$ or not}: Simulations with $n=500$, $m=10\,000$ where the IPSW estimator from Definitions~\ref{def:ipsw} and \ref{def:ipsw-with-pi} are compared to the estimates of the non-reweighted trials (Definitions~\ref{def:HT} and \ref{def:difference-in-means}) showing that the IPSW allow to recover the true ATE on the target population represented by the red dashed line (illustrating consistency from Theorem~\ref{thm:ipsw}). Estimating $\pi$ leads to a lower variance as expected (Corollary~\ref{cor_asympt_completely_estimated_pi_estimated}).
     \label{fig:semi_synth_pi_hat_or_not}
    }%
    \end{minipage}%
    \hfill%
    \begin{minipage}{.7\linewidth}
    \begin{center}
           \includegraphics[width=0.9\textwidth]{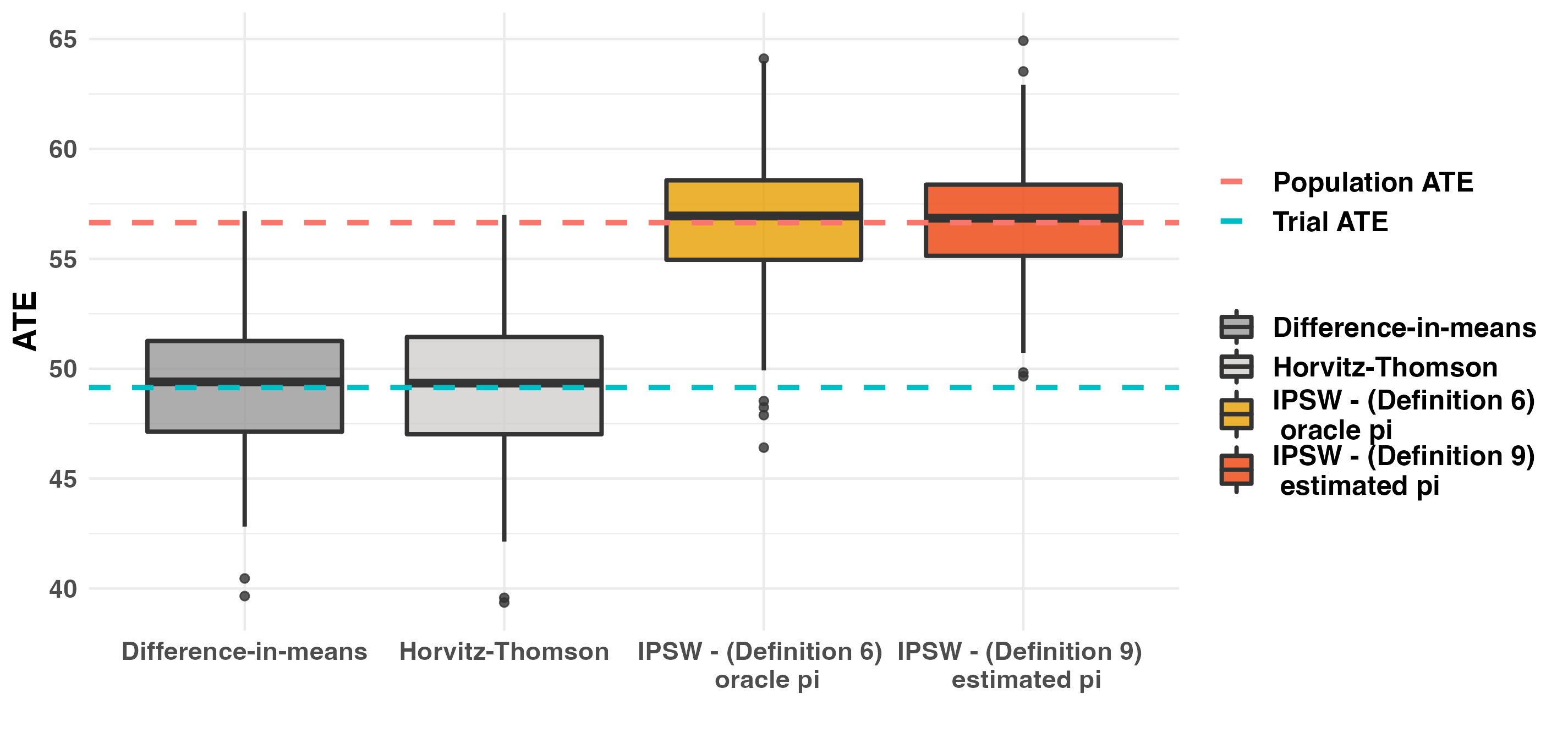}
     \end{center}
    \end{minipage}
\end{figure}
\paragraph{Estimating $\pi$ lowers the variance}
Simulations also illustrate the fact that estimating $\pi$ (Definition~\ref{def:ipsw-with-pi}) compared to not estimating it (Definition~\ref{def:ipsw}) lowers the variance, as shown on Figure~\ref{fig:semi_synth_pi_hat_or_not}. 
This is expected from Corollary~\ref{cor_asympt_completely_estimated_pi_estimated}.
\paragraph{The generalized (or re-weighted) estimate is not necessarily noisier than the trial's estimate}
Note that the variance of the IPSW - with estimation of $\pi$ or not - has a similar variance as the estimates coming from the RCT only (Horvitz-Thomson or difference-in-means). This is due to the presence of heteroscedasticity in the generative model (see equation~\eqref{equation:semi-synthetic-outcome-model}). 
Indeed, we would like to emphasize that re-weighting the trial does not necessarily lead to wider confidence intervals.
This somehow challenges a common and intuitive idea present in the literature and stating that a re-weighted trial always has a larger variance than the trial itself \citep{Stuart2017ChapterBook, Ling2022CriticalReview}.
This intuition comes from the multiplication of weights that can take large values (in particular if, for some $x$, $p_{\text{\tiny R}}(x) \ll p_{\text{\tiny T}}(x)$), making this idea valid as soon as the outcome noise is homoscedastic.
However, the asymptotic variance of the semi-oracle IPSW from Corollary~\ref{cor_asympt_semi_oracle} highlights that this intuitive and reasonable idea is not necessarily true, as soon as there is heteroscedascity, which occurs if some categories for which potential outcomes have higher uncertainty (larger noise) are more represented in the trial than in the target population: 

\begin{equation*}
V_{so} =  \sum_{x\in \mathcal{X}} \tikzmarknode{we}{\highlight{Purple}{ \color{black}  $ \frac{p_{\text{\tiny T}}^2(x) }{p_{\text{\tiny R}}(x) }$ }} \, \left( \tikzmarknode{amp}{\highlight{Bittersweet}{\color{black}  $\frac{\operatorname{Var}\left[ Y^{(1)} \mid X=x \right]}{\pi} + \frac{\operatorname{Var}\left[ Y^{(0)} \mid X=x \right]}{1-\pi}$}} \right)
\end{equation*}
\vspace*{0.5\baselineskip}
\begin{tikzpicture}[overlay,remember picture,>=stealth,nodes={align=left,inner ysep=1pt},<-]
\path (we.south) ++ (-0.65,+2.5em) node[anchor=south west,color=Purple!85] (sotext){\textsf{\footnotesize Weights}};
\path (amp.north) ++ (-2.7,-2.5em) node[anchor=north west,color=Bittersweet] (sotext){\textsf{\footnotesize Can be small for some $x$ with high weights $\frac{p_{\text{\tiny T}}^2(x) }{p_{\text{\tiny R}}(x) }$}};

\end{tikzpicture}

In particular in this simulation, having a variance of the IPSW estimate smaller than that of the treatment effect estimator on the trial is possible because individuals treated earlier have less uncertainty in the response than individuals with high $\texttt{TTT}$ (encoded in $\epsilon_{\texttt{TTT}}$), and the simulation is made such that in the target population such individuals are more present than in the trial. 

\paragraph{Shifted and not treatment effect modifier covariate increases variance: the example of Glasgow score (\texttt{GCS})}
It is possible to illustrate the results from Section~\ref{subsec:extended_adjustement_set} with the semi-synthetic simulation. 
For example, the Glasgow score (\texttt{GCS}) can be added to the minimal adjustment set previously used (see Figure~\ref{fig:semi_synth_pi_hat_or_not}), and leads to a loss of precision as this covariate is relatively strongly shifted between the two data sets and is not a treatment effect modifier (even if in the simulation this covariate has an impact on the outcome).
The increase in variance can be observed on Figure~\ref{fig:semi_synth_add_covariates}, where the green boxplot on the left represents such situation.

\begin{figure}[tbh!]
 \begin{minipage}{.33\linewidth}
  \caption{\textbf{Effect of non-necessary covariates on the variance}: IPSW (Definition~\ref{def:ipsw-with-pi}) with $n=3\,000$ and $m=10\,000$ showing that the addition of the covariate \texttt{GCS} (\textcolor{LimeGreen}{\textbf{shifted covariate not being a treatment effect modifier}}) increases the variance of the IPSW, while the addition of a \textcolor{Orchid}{\textbf{non-shifted treatment effect modifier}} (here simulated as no covariates from the actual data base where not shited) leads to an improvement in variance, compared to the \textcolor{YellowOrange}{\textbf{minimal set}}. Simulations are repeated $1,000$ times.}
     \label{fig:semi_synth_add_covariates}
    \end{minipage}%
    \hfill%
    \begin{minipage}{.6\linewidth}
    \begin{center}
           \includegraphics[width=0.9\textwidth]{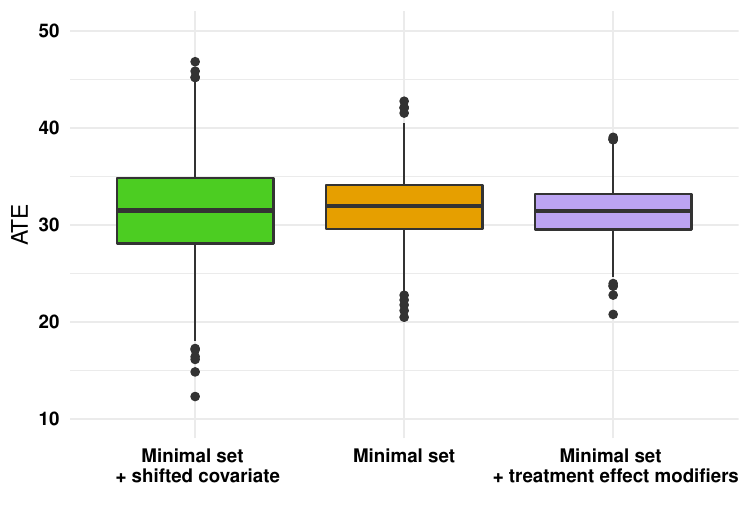}
     \end{center}
    \end{minipage}
\end{figure}

\paragraph{While a non-shifted but treatment effect modifier lowers the variance}
To illustrate a gain in precision due to the addition of a non-shifted treatment effect modifier, it was not possible to use the \textit{natural} covariates from the two original data sets as a distributional shift was always present in all covariates. To model such a situation, we added a categorical covariate \texttt{X\_sup} (5 levels), independent with all other covariates and without shift, in the data generative model to represent such a situation:
\begin{equation}\label{equation:semi-synthetic-outcome-model-transformed}
  \small  Y := f(\texttt{GCS},\texttt{Gender})  + A\, \tau(\texttt{TTT}, \texttt{Blood Pressure}, \texttt{X\_sup}) + \epsilon_{\texttt{TTT}}.
\end{equation}

Doing so, it is possible to illustrate that adding \texttt{X\_sup} in the adjustment set allows to lower the variance, and Figure~\ref{fig:semi_synth_add_covariates} presents such situation with the purple boxplot on the right.



\section{Conclusion and future work}\label{sec:conclusion-discussion}


In this work, we establish finite-sample and asymptotic results on different versions of the so-called \textit{Inverse Propensity Sampling Weights} estimator, when the adjustment set is constituted of categorical covariates. We give the explicit expressions of the biases and variances for all estimates, together with their quadratic risk. Our detailed analysis allows us to compare these different estimates in several finite-sample regimes. 
Indeed, to the best of our knowledge, our work is the first to study the impact of finite trial and observational data sets on IPSW performance in the context of generalization, by providing rate of convergence for several IPSW estimates.
By doing so, we link our work with previous results in epidemiology where one data source was considered infinite, and also explain how certain observations can be seen through the eyes of seminal work in causal inference (efficient estimation with IPW).

\paragraph{Which covariate to include?}
This work also reveals that care should be taken when selecting the covariates to generalize. 
From applied literature, we have noticed that practitioners usually select almost all available covariates to build the weights, which is encouraged by the fear of missing an important shifted treatment effect modifier.
We show that inclusion of many covariates comes with the risk of adjusting on shifted covariates that are not treatment effect modifiers, which can drastically damage the precision.
On the contrary, even though adding some non-shifted covariates may sound counterintuitive, we show that such practice improves asymptotic precision, as soon as the non-shifted additional covariate set modulates treatment effect. 
Still, adding too many covariates endangers overlap and therefore can lead to finite sample bias. 
In light of these theoretical results, we believe that physicians and epidemiologists have an important role to play in selecting a limited number of covariates when generalizing trial's findings. 

\paragraph{Future work}

Studying only categorical covariates is probably the main restriction of this work, as data can be hybrid and composed of continuous and categorical information.
However, even when facing a hybrid set of covariates - continuous and categorical - the user can still create bins for continuous covariates.
Even if such data-processing is not necessarily recommended, for a limited number of covariates this should allow to extend the analysis.
Indeed, binning covariates leads to within-stratum confounding, that is residual confounding due to rough bins, and therefore to an asymptotic bias due to factors that are poorly controlled on. 
To avoid within-stratum residual confounding, it is desirable to create more bins and split the data into more strata, but stratifying too finely with a finite sample may lead to \textit{(i)} a variance inflation and \textit{(ii)} the support inclusion assumption's invalidity.
Indeed, the performances of the IPSW in a high-dimensional setting can be limited. For example, if all input variables are binary, the finite-sample bias and variance can be rewritten as a function of $n/2^d$ (where $d$ is the number of input variables) and can thus spin out of control if the sample sizes are too small compared to the dimension of the problem.
Future work should investigate how our conclusion on the different asymptotic regimes and the covariates selection's impact on variance can be extended to settings with mixed-type covariates (for e.g. a smoother version of IPSW with density ratio estimation). \\

In practice, the limitation due to categorical covariates is balanced by the fact that within the medical field, clinical indicators and covariates are often scores and categories. 
For example, \cite{Berkowitz2018GeneralizingBlood} apply the IPSW to generalize the effect of blood pressure control relying on many categorical covariates such as health insurance status (insured, uninsured), tobacco smoking status (never, current, former), and so on.
When facing continuous covariates in practice, and having in mind the current theoretical understanding of the different generalization estimators, this IPSW version has interests.
A solution would be found at the crossroads between identification bias (due to imprecise bins) and variance inflation or finite sample bias (due to numerous bins). 
Quantifying such a tradeoff in specific settings would definitely help the practitioners by providing clear guidelines.

\section*{Acknowledgment} 
Part of this work was done while Bénédicte Colnet was visiting the Simons Institute for the Theory of Computing.
Also, part of this work was done during a visiting scholar period in the Stanford's Statistics Department. 
Therefore, we would like to thank the department for their welcoming and all the helpful and inspiring discussions, in particular with Prof. Trevor \textsc{Hastie}. Besides, we would like to thank the acting editor and the two referees for their work and their numerous comments, which helped us to improve a lot the quality of the manuscript.

\section*{Funding and conflict of interest}
Authors are all funded by their respective employer (\textsc{Inria} or École polytechnique) and have declared no conflict of interest.

\bibliographystyle{chicago}
\bibliography{biblio}

\newpage

\begin{center}
    {\Large APPENDIX}
\end{center}

\appendix

\section{Main proofs}\label{appendix:original-proofs}

\subsection{Proof of Theorem~\ref{thm:completely-oracle} - Completely oracle estimator $\hat \tau_{\pi, \text{\tiny T,R}, n}^*$}\label{proof:completely-oracle}

We prove Theorem~\ref{thm:completely-oracle} and the following corollary. 

\begin{corollary}
Under Assumptions of Theorem~\ref{thm:completely-oracle},  for all $n$, the quadratic risk of the completely oracle IPSW is given by,
\begin{equation*}
\mathbb{E}\left[ \left( \hat  \tau_{\pi, \text{\tiny T}, \text{\tiny R}, n}^* - \tau \right)^2\right] = \frac{V_o}{n},
\end{equation*}
which implies its $L^2$-consistency as $n$ tends to infinity, that is,
\begin{equation*}
\hat \tau_{\pi,\text{\tiny T,R}, n}^* \stackrel{L^2}{\longrightarrow} \tau.
\end{equation*}
\end{corollary}

We first recall the expression of the completely oracle estimator  introduced in Definition~\ref{def:ipsw-oracle},
\begin{equation*}
    \hat \tau_{\pi, \text{\tiny T,R}, n}^* =   \frac{1}{n} \sum_{i = 1}^n \frac{p_\text{\tiny T}\left( X_i \right)}{p_\text{\tiny R}\left( X_i \right)}\left(\frac{ Y_i A_i}{\pi} - \frac{ Y_i (1-A_i)}{1-\pi} \right).
\end{equation*}

This estimator can be rewritten as,
\begin{equation*}
     \hat \tau_{\pi, \text{\tiny T,R}, n}^* = \sum_{x\in\mathds{X}} \frac{p_\text{\tiny T}(x)}{p_\text{\tiny R}(x)} \left( \frac{1}{n}\sum_{i =1}^n \mathbbm{1}_{X_i = x} \left(\frac{ Y_i A_i}{\pi} - \frac{ Y_i (1-A_i)}{1-\pi} \right)  \right),
\end{equation*}
since $X_i$ take values in a categorical set $\mathbb{X}$. 
This rewriting is extensively used in the proof.\\

\textbf{Bias}\\

Recall that, for all $x\in \mathds{X}$, $p_\text{\tiny R}(x)$ and $p_\text{\tiny T}(x)$ are not random variables. We have
\begin{align*}
    \mathbb{E}\left[  \hat \tau_{\pi, \text{\tiny T,R}, n}^*\right] &=   \mathbb{E} \left[\sum_{x\in \mathds{X}}\frac{p_\text{\tiny T}(x)}{p_\text{\tiny R}(x)} \frac{1}{n}\sum_{i =1}^n \mathbbm{1}_{X_i = x} \left(\frac{ Y_i A_i}{\pi} - \frac{ Y_i (1-A_i)}{1-\pi} \right) \right] &&\text{By definition} \\
     &= \sum_{x\in \mathds{X}} \mathbb{E}\left[\frac{p_\text{\tiny T}(x)}{p_\text{\tiny R}(x)} \frac{1}{n}\sum_{i =1}^n \mathbbm{1}_{X_i = x} \left(\frac{ Y_i A_i}{\pi} - \frac{ Y_i (1-A_i)}{1-\pi} \right)\right] && \text{Linearity of $\mathbb{E}[.]$}\\
    &= \sum_{x\in \mathds{X}} \frac{p_\text{\tiny T}(x)}{p_\text{\tiny R}(x)} \mathbb{E}\left[\frac{1}{n}\sum_{i =1}^n \mathbbm{1}_{X_i = x} \left(\frac{ Y_i A_i}{\pi} - \frac{ Y_i (1-A_i)}{1-\pi} \right)\right] && \text{$p_\text{\tiny R}(x)$ and $p_\text{\tiny T}(x)$ are not random}\\
    &= \sum_{x\in \mathds{X}}\frac{p_\text{\tiny T}(x)}{p_\text{\tiny R}(x)}  \mathbb{E}_\text{\tiny R}\left[ \mathbbm{1}_{X_i = x} \left(\frac{ Y_i A_i}{\pi} - \frac{ Y_i (1-A_i)}{1-\pi} \right) \right]&& \text{Linearity \& \textit{iid} trial}\\
    &= \sum_{x\in \mathds{X}} \frac{p_\text{\tiny T}(x)}{p_\text{\tiny R}(x)}  \mathbb{E}_\text{\tiny R}\left[ \mathbbm{1}_{X_i = x} \left(\frac{ Y_i^{(1)} A_i}{\pi} - \frac{ Y_i^{(0)} (1-A_i)}{1-\pi} \right) \right] && \text{SUTVA (see Assumption~\ref{a:trial-internal-validity}).}
\end{align*}

Noting that,

\begin{equation*}
    p_\text{\tiny R}(x) = \mathbb{P}_\text{\tiny R}[X=x]=\mathbb{P}_\text{\tiny R}[X_i=x]=\mathbb{E}_\text{\tiny R}\left[ \mathbbm{1}_{X_i = x}\right],
\end{equation*}

one can condition on the random variable $X_i$, yielding
\begin{equation*}
    \mathbb{E}_\text{\tiny R}\left[ \mathbbm{1}_{X_i = x} \left(\frac{ Y_i^{(1)} A_i}{\pi} - \frac{ Y_i^{(0)} (1-A_i)}{1-\pi} \right) \right] = \mathbb{E}_\text{\tiny R}\left[  \frac{ Y_i^{(1)} A_i}{\pi} - \frac{ Y_i^{(0)} (1-A_i)}{1-\pi}  \mid X_i = x\right]\underbrace{\mathbb{E_\text{\tiny R}}\left[ \mathbbm{1}_{X_i = x}\right]}_{= p_\text{\tiny R}(x)}.
\end{equation*}
Then,
\begin{align*}
     \mathbb{E}\left[  \hat \tau_{\pi, \text{\tiny T,R}, n}^*\right]  &= \sum_{x\in \mathds{X}}  p_\text{\tiny T}(x) \mathbb{E}_\text{\tiny R}\left[  \frac{ Y_i^{(1)} A_i}{\pi} - \frac{ Y_i^{(0)} (1-A_i)}{1-\pi}  \mid X_i = x\right] && \text{From previous derivations} \\
    &= \sum_{x\in \mathds{X}}  p_\text{\tiny T}(x)   \left( \frac{\mathbb{E}_\text{\tiny R}\left[  Y_i^{(1)} A_i \mid X_i = x\right] }{\pi}- \frac{\mathbb{E}_\text{\tiny R}\left[  Y_i^{(0)} (1-A_i)  \mid X_i = x\right]}{1-\pi} \right) && \text{Linearity of $\mathbb{E}[.]$ and $\pi$ is constant} \\
    &= \sum_{x\in \mathds{X}}  p_\text{\tiny T}(x) \Bigg(  \frac{\mathbb{E}_\text{\tiny R}\left[  Y_i^{(1)} \mid X_i = x\right] \mathbb{E}_\text{\tiny R}\left[  A_i \mid X_i = x\right] }{\pi} \\
    &\qquad - \frac{\mathbb{E}_\text{\tiny R}\left[  Y_i^{(0)} \mid X_i = x\right]\mathbb{E}_\text{\tiny R}\left[  (1-A_i)  \mid X_i = x\right]}{1-\pi}\Bigg) && \text{Randomization (see Assumption~\ref{a:trial-internal-validity})} \\
    &= \sum_{x\in \mathds{X}}  p_\text{\tiny T}(x)   \left(\mathbb{E}_\text{\tiny R}\left[  Y_i^{(1)}  \mid X_i = x\right]- \mathbb{E}_\text{\tiny R}\left[  Y_i^{(0)} \mid X_i = x\right] \right)&& \text{$\mathbb{E}_\text{\tiny R}\left[  A_i  \mid X_i = x\right] = \pi$} \\
    &= \sum_{x\in \mathds{X}}  p_\text{\tiny T}(x)   \mathbb{E}_\text{\tiny R}\left[  Y_i^{(1)} -  Y_i^{(0)} \mid X_i = x\right]&& \text{Linearity of $\mathbb{E}[.]$} \\
    &= \sum_{x\in \mathds{X}}  p_\text{\tiny T}(x)   \mathbb{E}_\text{\tiny T}\left[  Y_i^{(1)} -  Y_i^{(0)} \mid X_i = x\right]&& \text{Transportability (see Assumption~\ref{a:cate-indep-s-knowing-X})} \\
     &= \tau, && \text{Law of total probability} \\
\end{align*}
which concludes the first part of the proof.\\ 

Note that the previous derivations, relying on \textit{iid}, Assumption~\ref{a:trial-internal-validity} (Trial internal validity with SUTVA, definition of $\pi$, and randomization), Assumption~\ref{a:cate-indep-s-knowing-X}, and the law of total probability, lead to the following intermediary result,

\begin{equation}\label{eq:intermediary-result-proof-cate}
    \mathbb{E}_\text{\tiny R}\left[  \frac{ Y_i^{(1)} A_i}{\pi} - \frac{ Y_i^{(0)} (1-A_i)}{1-\pi}  \mid X_i = x\right] = \mathbb{E}_\text{\tiny T}\left[  Y_i^{(1)} -  Y_i^{(0)} \mid X_i = x\right] = \tau(x).
\end{equation}

\eqref{eq:intermediary-result-proof-cate} will be used in other proofs.

\textbf{Variance}\\


To shorten notation, we denote by $\mathbf{X}_{n}\in \mathds{X}^{n}$ the vector composed of the $n$ observations in the trial. We then use the law of total variance, conditioning on $\mathbf{X}_{n}$, 
\begin{equation}
    \operatorname{Var}\left[   \hat \tau_{\pi, \text{\tiny T,R}, n}^* \right] =  \operatorname{Var}\left[ \mathbb{E}\left[ \hat \tau_{\pi, \text{\tiny T,R}, n}^*\mid \mathbf{X}_{n} \right] \right] + \mathbb{E} \left[ \operatorname{Var}\left[ \hat \tau_{\pi, \text{\tiny T,R}, n}^* \mid \mathbf{X}_{n}  \right] \right]. \label{eq_proof_compl_oracle_var}
\end{equation}
  
 
Considering the first term in the right-hand side of \eqref{eq_proof_compl_oracle_var}, 
 \begin{align*}
     \mathbb{E}\left[ \hat \tau_{\pi, \text{\tiny T,R}, n}^*  \mid \mathbf{X}_n  \right]  &=    \mathbb{E}\left[ \sum_{x\in\mathds{X}} \frac{p_\text{\tiny T}(x)}{p_\text{\tiny R}(x)} \frac{1}{n}\sum_{i =1}^n \mathbbm{1}_{X_i = x} \left(\frac{ Y_i^{(1)} A_i}{\pi} - \frac{ Y_i^{(0)} (1-A_i)}{1-\pi} \right)\mid  \mathbf{X}_n  \right] && \text{By definition (and SUTVA)} \\
     &= \sum_{x\in\mathds{X}} \frac{p_\text{\tiny T}(x)}{p_\text{\tiny R}(x)} \frac{1}{n}\mathbb{E}\left[\sum_{i =1}^n \mathbbm{1}_{X_i = x} \left(\frac{ Y_i^{(1)} A_i}{\pi} - \frac{ Y_i^{(0)} (1-A_i)}{1-\pi} \right)\mid  \mathbf{X}_n  \right]. && \text{Linearity of $\mathbb{E}[.]$} 
     \end{align*}
     
Note that this last derivation also uses the fact that neither $p_\text{\tiny T}(x)$ nor $p_\text{\tiny R}(x)$ are random variables. 

\begin{align*}
     \mathbb{E} \left[ \hat \tau_{\pi, \text{\tiny T,R}, n}^*  \mid \mathbf{X}_n  \right] 
     &=  \sum_{x\in\mathds{X}} \frac{p_\text{\tiny T}(x)}{p_\text{\tiny R}(x)} \sum_{i =1}^n\frac{ \mathbbm{1}_{X_i = x}}{n}  \mathbb{E} \left[ \frac{ Y_i^{(1)} A_i}{\pi} - \frac{ Y_i^{(0)} (1-A_i)}{1-\pi} \mid X_i \right] && \text{\textit{iid} individuals}\\
     &=   \sum_{x\in\mathds{X}}  \frac{p_\text{\tiny T}(x)}{p_\text{\tiny R}(x)}  \sum_{i =1}^n \frac{\mathbbm{1}_{X_i = x} }{n}  \tau(X_i) && \\
     &= \frac{1 }{n}   \sum_{x\in\mathds{X}}  \frac{p_\text{\tiny T}(x)}{p_\text{\tiny R}(x)} \tau(x) \sum_{i =1}^n \mathbbm{1}_{X_i = x}  && \text{Transportability (see Assumption~\ref{a:cate-indep-s-knowing-X})}
 \end{align*}

Now, this last term can be written as a unique sum on $i \in \{1, \hdots, n\}$, that is,
\begin{equation*}
    \frac{1 }{n}  \sum_{x\in\mathds{X}}  \frac{p_\text{\tiny T}(x)}{p_\text{\tiny R}(x)} \tau(x) \sum_{i =1}^n \mathbbm{1}_{X_i = x}  = \frac{1}{n} \sum_{i=1}^n \frac{p_\text{\tiny T}(X_i)}{p_\text{\tiny R}(X_i)} \tau(X_i).
\end{equation*}

Taking the variance of this term leads to,
\begin{align}
   \operatorname{Var}\left[ \mathbb{E}_\text{\tiny R}\left[ \hat \tau_{\pi, \text{\tiny T,R}, n}^* \mid \mathbf{X}_n \right] \right] &=   \operatorname{Var}\left[\frac{1}{n} \sum_{i=1}^n \frac{p_\text{\tiny T}(X_i)}{p_\text{\tiny R}(X_i)} \tau(X_i) \right] \nonumber \\
   &= \frac{1}{n}  \operatorname{Var}_\text{\tiny R}\left[ \frac{p_\text{\tiny T}(X)}{p_\text{\tiny R}(X)} \tau(X) \right].  && \text{\textit{iid} observations on trial (Assumption~\ref{a:trial-internal-validity})} \label{eq_proof_compl_oracle_var2}
\end{align}

 
Regarding the second term, 
 \begin{align}
      \operatorname{Var}\left[ \hat \tau_{\pi, \text{\tiny T,R}, n}^* \mid  \mathbf{X}_n  \right] &=       \operatorname{Var}_{\text{\tiny R}}\left[ \frac{1}{n} \sum_{i = 1}^n \frac{p_\text{\tiny T}\left( X_i \right)}{p_\text{\tiny R}\left( X_i \right)}\left(\frac{ Y_i A_i}{\pi} - \frac{ Y_i (1-A_i)}{1-\pi} \right) \mid  \mathbf{X}_n  \right] \nonumber  \\ 
      &=  \frac{1}{n^2} \sum_{i = 1}^n \left( \frac{p_\text{\tiny T}\left( X_i \right)}{p_\text{\tiny R}\left( X_i \right)}\right)^2 \operatorname{Var}_{\text{\tiny R}}\left[\left(\frac{ Y_i A_i}{\pi} - \frac{ Y_i (1-A_i)}{1-\pi} \right)  \mid  \mathbf{X}_n \right] \nonumber 
\\
&=  \frac{1}{n^2} \sum_{i = 1}^n \left( \frac{p_\text{\tiny T}\left( X_i \right)}{p_\text{\tiny R}\left( X_i \right)}\right)^2 \operatorname{Var}_{\text{\tiny R}}\left[\left(\frac{ Y_i A_i}{\pi} - \frac{ Y_i (1-A_i)}{1-\pi} \right)  \mid  X_i \right]. \label{eq_variance_oracle_intermediary_before_VHT(x)}
      \end{align}

Recall that the variance of the Horvitz-Thomson estimator (see Definition~\ref{def:HT}) conditioned on $X_i$ is given by 
\begin{equation}\label{eq_variance_random_variable_like_HT}
           \operatorname{Var}_{\text{\tiny R}}\left[ \hat \tau_{\text{\tiny HT},n} \mid X_i \right] = \frac{1}{n} \operatorname{Var}_{\text{\tiny R}}\left[\left(\frac{ Y_i A_i}{\pi} - \frac{ Y_i (1-A_i)}{1-\pi} \right)  \mid  X_i \right].
      \end{equation}


Then, one can use principles used for the proof of Lemma~\ref{lemma:HT-unbiased-and-variance} (see Section~\ref{appendix:useful-results-rct}) to have

\begin{align}\label{eq_explanation_for_g(x)_VHTx}
      n\operatorname{Var}\left[ \hat{\tau}_{\text{\tiny HT},n}  \mid X_i \right] 
      &= \mathbb{E}_{\text{\tiny R}}\left[ \frac{\left( Y^{(1)} \right)^2}{\pi}  \mid X_i \right]  + \mathbb{E}_{\text{\tiny R}}\left[ \frac{\left( Y^{(0)} \right)^2}{1-\pi}  \mid X_i  \right]  - \tau(X_i)^2 := V_{ \text{\tiny HT}}(X_i)  .
\end{align}

Then, coming back to \eqref{eq_variance_oracle_intermediary_before_VHT(x)},
      
   \begin{align}
       \mathbb{E}_\text{\tiny R}\left[\operatorname{Var}\left[ \hat \tau_{\pi, \text{\tiny T,R}, n}^* \mid  \mathbf{X}_n  \right] \right]&=\mathbb{E}_\text{\tiny R}\left[\frac{1}{n^2} \sum_{i = 1}^n \left( \frac{p_\text{\tiny T}\left( X_i \right)}{p_\text{\tiny R}\left( X_i \right)}\right)^2 V_{ \text{\tiny HT}}(X_i) \right] \nonumber \\
       &=\mathbb{E}_\text{\tiny R}\left[\frac{1}{n^2} \sum_{i = 1}^n \left( \sum_{x\in\mathds{X}}  \mathbbm{1}_{X_i = x} \right) \left(\frac{p_\text{\tiny T}\left( X_i \right)}{p_\text{\tiny R}\left( X_i \right)}\right)^2 V_{ \text{\tiny HT}}(X_i) \right]\nonumber  \\
        &=\mathbb{E}_\text{\tiny R}\left[ \sum_{x\in\mathds{X}}   \frac{1}{n^2}  \left(  \frac{p_\text{\tiny T}\left( x \right)}{p_\text{\tiny R}\left( x \right)}\right)^2 V_{ \text{\tiny HT}}(x) \sum_{i = 1}^n \mathbbm{1}_{X_i = x}  \right] \nonumber  \\
        &= \sum_{x\in\mathds{X}}   \frac{1}{n^2}  \left(  \frac{p_\text{\tiny T}\left( x \right)}{p_\text{\tiny R}\left( x \right)}\right)^2 V_{ \text{\tiny HT}}(x)  \mathbb{E}_\text{\tiny R}\left[\sum_{i = 1}^n \mathbbm{1}_{X_i = x}  \right]\nonumber  \\
           &= \sum_{x\in\mathds{X}}   \frac{1}{n}  \left(  \frac{p_\text{\tiny T}\left( x \right)}{p_\text{\tiny R}\left( x \right)}\right)^2 V_{ \text{\tiny HT}}(x)  \mathbb{E}_\text{\tiny R}\left[ \frac{\sum_{i = 1}^n \mathbbm{1}_{X_i = x} }{n} \right] \nonumber  \\
    &= \sum_{x\in\mathds{X}}   \frac{1}{n}  \left(  \frac{p_\text{\tiny T}\left( x \right)}{p_\text{\tiny R}\left( x \right)}\right)^2 V_{ \text{\tiny HT}}(x)  p_\text{\tiny R}\left( x \right)  && \text{Assumption~\ref{a:repres-rct}} \nonumber  \\
    &= \sum_{x\in\mathds{X}}   \frac{1}{n}   \frac{p^2_\text{\tiny T}\left( x \right)}{p_\text{\tiny R}\left( x \right)} V_{ \text{\tiny HT}}(x) \nonumber \\
    &= \frac{1}{n}   \sum_{x\in\mathds{X}}    \frac{p^2_\text{\tiny T}\left( x \right)}{p_\text{\tiny R}\left( x \right)} \left( \mathbb{E}_{\text{\tiny R}}\left[ \frac{\left(Y^{(1)}\right)^2}{\pi}   \mid X = x \right] + \mathbb{E}_{\text{\tiny R}}\left[ \frac{\left(Y^{(0)}\right)^2}{1-\pi}  \mid X = x\right] - \tau(x)^2 \right), \label{eq_proof_compl_oracle_var1}\\
   \end{align}

Combining \eqref{eq_proof_compl_oracle_var1} and \eqref{eq_proof_compl_oracle_var2} into \eqref{eq_proof_compl_oracle_var} leads to, for all $n$, 
\begin{equation*}
 \operatorname{Var}\left[   \hat \tau_{\pi, \text{\tiny T,R}, n}^* \right] =  \frac{V_o}{n}
 \end{equation*}
 where
 \begin{equation*}
 V_o =     
 \operatorname{Var}\left[ \frac{p_\text{\tiny T}(X_i)}{p_\text{\tiny R}(X_i)} \tau(X_i) \right] +   \sum_{x\in\mathds{X}}    \frac{p^2_\text{\tiny T}\left( x \right)}{p_\text{\tiny R}\left( x \right)} V_{\tiny \text{\tiny HT}}(x).\\
 \end{equation*}

Note that it is also possible to write the result such as,

\begin{equation*}
     V_o =     
 \operatorname{Var}\left[ \frac{p_\text{\tiny T}(X)}{p_\text{\tiny R}(X)} \tau(X) \right] + \mathbb{E}_\text{\tiny R}\left[  \frac{p^2_\text{\tiny T}\left( X \right)}{p^2_\text{\tiny R}\left( X \right)} V_{\tiny \text{\tiny HT}}(X) \right],
\end{equation*}

noting that 

\begin{equation*}
     \sum_{x\in\mathds{X}}    \frac{p^2_\text{\tiny T}\left( x \right)}{p_\text{\tiny R}\left( x \right)} V_{\tiny \text{\tiny HT}}(x)=  \mathbb{E}_\text{\tiny R}\left[  \frac{p^2_\text{\tiny T}\left( X \right)}{p^2_\text{\tiny R}\left( X \right)} V_{\tiny \text{\tiny HT}}(X) \right]
\end{equation*}

\textbf{Quadratic risk and consistency}\\

For any estimate $\hat{\tau}$, we have 
\begin{equation*}
\mathbb{E}\left[ \left( \hat  \tau - \tau \right)^2\right] = \left( \mathbb{E}\left[ \hat  \tau \right] - \tau \right)^2 + \operatorname{Var}\left[   \hat  \tau \right].
\end{equation*}
Therefore, the risk of the completely oracle IPSW estimate satisfies
\begin{equation*}
\mathbb{E}\left[ \left( \hat  \tau - \tau \right)^2\right] = \frac{V_o}{n}.
\end{equation*}
The $L^2$ consistency holds by letting $n$ tend to infinity.

\subsection{Proofs for the semi-oracle IPSW $\hat \tau_{\pi, \text{\tiny T}, n}^*$}\label{proof:semi-oracle}

\subsubsection{Proof of Proposition~\ref{prop_semi_oracle_bias_variance}}\label{proof:explicit-bias-variance-and-bounds-semi-oracle}

We prove the following more detailed Proposition  where $Z_n(x) = \sum_{i =1}^n \mathbbm{1}_{X_i = x}$.

\begin{proposition}
\label{prop_app_semioracleipsw}
Under the general setting defined in Subsection~\ref{subsec:model}, granting Assumptions~\ref{a:repres-rct}-\ref{a:pos}, the bias of the semi-oracle IPSW satisfies, for all $n$,  
\begin{flalign*}
&&  \mathbb{E}\left[\hat \tau_{\pi, \text{\tiny T}, n}^*  \right] - \tau = - \sum_{x\in\mathds{X}} p_{\text{\tiny T}}(x)\left(1 - p_{\text{\tiny R}}(x)\right)^n \tau(x),&&
\\
\text{and} &&
 \biggl|\mathbb{E}\left[\hat \tau_{\pi,\text{\tiny T}, n}^*  \right] - \tau \biggr| \leq \left(1 - \min_x p_{\text{\tiny R}}(x)\right)^n \mathbb{E}_{\text{\tiny T}} \left[ \left| \tau(X) \right| \right].&&
\end{flalign*}
Moreover, under the same set of assumptions, the variance of the semi-oracle IPSW satisfies, for all $n$, 
\begin{flalign*}
&&    n \operatorname{Var}\left[ \hat \tau_{\pi, \text{\tiny T}, n}^*\right]  =&  \sum_{x\in\mathds{X}}  p_\text{\tiny T}\left( x \right)^2 V_{\text{\tiny HT}}(x) \mathbb{E}_\text{\tiny R}\left[ \frac{\mathbbm{1}_{Z_n(x)>0}}{ \hat p_{\text{\tiny R},n}(x) } \right] + n \operatorname{Var}\left[ \mathbb{E}_{\text{\tiny T}} \left[ \tau(X) \mathds{1}_{Z_n(X)=0} | \mathbf{X}_n \right] \right],&&
\\
\text{and}
&&
\operatorname{Var}\left[   \hat \tau_{\pi, \text{\tiny T}, n}^* \right] \le &\,\frac{2 V_{so}}{n+1}  +  \left( 1 - \min_{x \in \mathbb{X}} p_\text{\tiny R}(x)\right)^n \left( \mathbb{E}_{\text{\tiny T}} \left[ |\tau(X)| \right]\right)^2,  
&&
\\
\text{with} &&
    V_{\text{so}}:= &\,\mathbb{E}_\text{\tiny R}\left[ \left(\frac{p_\text{\tiny T}(X)}{p_\text{\tiny R}(X)} \right)^2V_{ \text{\tiny HT}}(X)\right].
&&
\end{flalign*}
\end{proposition}

Note that, as soon as $\tau(x)$ is of constant sign, the sign of the bias is known and opposite to that of $\tau(x)$. In fact, because of potentially empty categories in the trial, the expectation of the semi-oracle IPSW estimate $\mathbb{E}\left[\hat \tau_{\pi, \text{\tiny T}, n}^*  \right]$ is pushed toward zero, if $\tau(x)$ is of constant sign.

\begin{proof}
We first recall the definition of the semi-oracle estimator introduced in Definition~\ref{def:ipsw-semi-oracle}:

\begin{equation*}
    \hat \tau_{\pi, \text{\tiny T}, n}^* =   \frac{1}{n} \sum_{i = 1}^n \frac{p_\text{\tiny T}\left( X_i \right)}{\hat  p_{\text{\tiny R},n} (X_i)}\left(\frac{ Y_i A_i}{\pi} - \frac{ Y_i (1-A_i)}{1-\pi} \right),
\end{equation*}

where, for all $x\in \mathds{X}$, \begin{equation}\label{eq-proof:dep-hat-pr}
    \hat p_{\text{\tiny R},n}\left( x \right)= \frac{\sum_{i=1}^n \mathbbm{1}_{X_i = x}}{n}.
\end{equation}

Similarly to the completely oracle estimator, the semi-oracle estimator can be written as,

\begin{equation*}
     \hat \tau_{\pi, \text{\tiny T}, n}^* = \sum_{x\in\mathds{X}} \frac{p_\text{\tiny T}(x) \mathds{1}_{Z_n(x) >0}}{\hat  p_{\text{\tiny R},n}(x)} \left( \frac{1}{n}\sum_{i =1}^n \mathbbm{1}_{X_i = x} \left(\frac{ Y_i A_i}{\pi} - \frac{ Y_i (1-A_i)}{1-\pi} \right)  \right) ,
\end{equation*}
where $Z_n(x) = \sum_{i =1}^n \mathbbm{1}_{X_i = x}$,
since $X_i$ take values in a categorical set $\mathbb{X}$. From now on, we use the convention that $\mathds{1}_{Z_n(x) >0} / \hat  p_{\text{\tiny R},n}(x) = 0$ if $Z_n(x) = 0$.\\



\textbf{Bias}\\

To shorten notation, we denote the full vector of covariates $\mathbf{X}_n\in \mathds{X}^n$, comprising the $n$ observations $X_1, X_2, \dots X_n \in \mathds{X}$ in the trial. We have

\begin{align*}
    \mathbb{E}\left[\hat \tau_{\pi, \text{\tiny T}, n}^*  \right] &= \mathbb{E}\left[ \sum_{x\in\mathds{X}} \frac{p_\text{\tiny T}(x)\mathds{1}_{Z_n(x) >0}}{\hat  p_{\text{\tiny R},n}(x)} \left( \frac{1}{n}\sum_{i =1}^n \mathbbm{1}_{X_i = x} \left(\frac{ Y_i A_i}{\pi} - \frac{ Y_i (1-A_i)}{1-\pi} \right)  \right)\right]   && \text{By definition} \\
    &=  \sum_{x\in\mathds{X}}\mathbb{E}\left[ \frac{p_\text{\tiny T}(x)\mathds{1}_{Z_n(x) >0}}{\hat  p_{\text{\tiny R},n}(x)} \left( \frac{1}{n}\sum_{i =1}^n \mathbbm{1}_{X_i = x} \left(\frac{ Y_i^{(1)} A_i}{\pi} - \frac{ Y_i^{(0)} (1-A_i)}{1-\pi} \right)  \right)\right]  && \text{Linearity and SUTVA} \\
    &= \sum_{x\in\mathds{X}}\mathbb{E}\left[\mathbb{E}\left[ \frac{p_\text{\tiny T}(x)\mathds{1}_{Z_n(x) >0}}{\hat  p_{\text{\tiny R},n}(x)} \left( \frac{1}{n}\sum_{i =1}^n \mathbbm{1}_{X_i = x} \left(\frac{ Y_i^{(1)} A_i}{\pi} - \frac{ Y_i^{(0)} (1-A_i)}{1-\pi} \right)  \right)\mid \mathbf{X}_n  \right] \right] && \text{Law of total expect.} \\
   &= \sum_{x\in\mathds{X}}\mathbb{E} \left[ p_\text{\tiny T}(x)  \mathbb{E}\left[ \frac{\mathds{1}_{Z_n(x) >0}}{\hat  p_{\text{\tiny R},n}(x)} \left( \frac{1}{n}\sum_{i =1}^n \mathbbm{1}_{X_i = x} \left(\frac{ Y_i^{(1)} A_i}{\pi} - \frac{ Y_i^{(0)} (1-A_i)}{1-\pi} \right)  \right)\mid \mathbf{X}_n  \right] \right] && \text{$p_\text{\tiny T}(x)$ is deterministic} \\
   &= \sum_{x\in\mathds{X}}\mathbb{E} \left[ \frac{p_\text{\tiny T}(x)\mathds{1}_{Z_n(x) >0}}{\hat  p_{\text{\tiny R},n}(x)} \mathbb{E}\left[  \left( \frac{1}{n}\sum_{i =1}^n \mathbbm{1}_{X_i = x} \left(\frac{ Y_i^{(1)} A_i}{\pi} - \frac{ Y_i^{(0)} (1-A_i)}{1-\pi} \right)  \right)\mid \mathbf{X}_n  \right] \right]  \\
    &= \sum_{x\in\mathds{X}}\mathbb{E}\left[ \frac{p_\text{\tiny T}(x)\mathds{1}_{Z_n(x) >0}}{\hat  p_{\text{\tiny R},n}(x)}  \frac{1}{n} \sum_{i =1}^n\mathbbm{1}_{X_i = x} \mathbb{E}\left[ \frac{ Y_i^{(1)} A_i}{\pi} - \frac{ Y_i^{(0)} (1-A_i)}{1-\pi} \mid \mathbf{X}_n  \right] \right]
\end{align*}

This last line uses the fact that $\frac{\sum_{i =1}^n \mathbbm{1}_{X_i = x}}{n}$ is measurable with respect to $\mathbf{X}_n$. 
Then, note that,

\begin{align*}
      \mathbbm{1}_{X_i = x} \mathbb{E}\left[  \frac{ Y_i^{(1)} A_i}{\pi} - \frac{ Y_i^{(0)} (1-A_i)}{1-\pi} \mid \mathbf{X}_n   \right] &=\mathbbm{1}_{X_i = x} \mathbb{E}\left[  \frac{ Y_i^{(1)} A_i}{\pi} - \frac{ Y_i^{(0)} (1-A_i)}{1-\pi} \mid X_i \right] && \text{\textit{iid} observations.} 
\end{align*}

Then, recall from the proof in Subsection~\ref{proof:completely-oracle}, and in particular from \eqref{eq:intermediary-result-proof-cate} that

\begin{align*}
    \mathbbm{1}_{X_i = x} \mathbb{E}\left[  \frac{ Y_i^{(1)} A_i}{\pi} - \frac{ Y_i^{(0)} (1-A_i)}{1-\pi} \mid \mathbf{X}_n   \right]
     &= \mathbbm{1}_{X_i = x} \mathbb{E}\left[\frac{ Y_i^{(1)} A_i}{\pi} - \frac{ Y_i^{(0)} (1-A_i)}{1-\pi} \mid X = x \right] && \text{Indicator \textit{forcing} $X=x$.} \\
     &=  \mathbbm{1}_{X_i = x} \tau(x) && \text{Transportability.}
\end{align*}

Therefore,
\begin{align*}
  \mathbb{E}\left[\hat \tau_{\pi, \text{\tiny T}, n}^*  \right] &= \sum_{x\in\mathds{X}}\mathbb{E}\left[ \frac{p_\text{\tiny T}(x)\mathds{1}_{Z_n(x) >0}}{\hat  p_{\text{\tiny R},n}(x)} \frac{\sum_{i =1}^n \mathbbm{1}_{X_i = x}}{n}\tau(x) \right] \\
     &= \sum_{x\in\mathds{X}}\mathbb{E}\left[ \frac{p_\text{\tiny T}(x)\mathds{1}_{Z_n(x) >0}}{\frac{\sum_{i =1}^n \mathbbm{1}_{X_i = x}}{n}} \frac{\sum_{i =1}^n \mathbbm{1}_{X_i = x}}{n}\tau(x) \right] && \text{Estimation procedure - Equation~\ref{eq-proof:dep-hat-pr}}\\  &= \sum_{x\in\mathds{X}}\mathbb{E}\left[ p_\text{\tiny T}(x) \tau(x) \mathds{1}_{Z_n(x) >0}  \right].
\end{align*}
Note that $Z_n(x) = \sum_{i =1}^n \mathbbm{1}_{X_i = x}$ is distributed as $\mathfrak{B}(n, p_{\text{\tiny R}}(x))$. This leads to the following equality, 
\begin{align*}
  \mathbb{E}\left[\hat \tau_{\pi, \text{\tiny T}, n}^*  \right] &= \sum_{x\in\mathds{X}}\mathbb{E}\left[p_\text{\tiny T}(x) \tau(x) \mathbbm{1}_{Z_n(x)>0} \right] \\
     &= \sum_{x\in\mathds{X}} p_\text{\tiny T}(x) \tau(x) \mathbb{E}\left[ \mathbbm{1}_{Z_n(x)>0} \right] \\
     &= \sum_{x\in\mathds{X}} p_\text{\tiny T}(x) \tau(x) \left(1 - \left(1 - p_{\text{\tiny R}}(x)\right)^n \right). \\
\end{align*}

\textbf{Upper bound of the bias}.\\

If $p_{\text{\tiny R}}(x) = 0$, then $p_{\text{\tiny T}}(x) = 0$ (due to the support inclusion assumption, see Assumption~\ref{a:pos}). Therefore, for all $x \in \mathds{X}, p_{\text{\tiny R}}(x) > 0$. Then, it is possible to bound the bias for any sample size $n$, noting that,
\begin{align*}
 |\mathbb{E}\left[\hat \tau_{\pi, \text{\tiny T}, n}^*  \right] - \tau| & =  \left| \sum_{x\in\mathds{X}} p_\text{\tiny T}(x) \tau(x) \left(1 - (1 - p_{\text{\tiny R}}(x)\right)^n)  - \tau \right|\\
  & =   \left| \sum_{x\in\mathds{X}} p_\text{\tiny T}(x) \tau(x) \left(1 - (1 - p_{\text{\tiny R}}(x))^n \right)  -  \sum_{x\in\mathds{X}} p_\text{\tiny T}(x) \tau(x)  \right|\\
 & =   \left|\sum_{x\in\mathds{X}} p_\text{\tiny T}(x) \tau(x)  \left(1 - p_{\text{\tiny R}}(x)\right)^n  \right|\\
 & \leq     \left(1 - \min_x p_{\text{\tiny R}}(x)\right)^n \sum_{x\in\mathds{X}} p_\text{\tiny T}(x) \left|\tau(x)\right|\\
 & \leq  \left(1 - \min_x p_{\text{\tiny R}}(x) \right)^n \mathbb{E}_{\text{\tiny T}} \left[ | \tau(X) | \right].
\end{align*}

\textbf{Variance}\\

The proof follows the same track as that of the completely oracle IPSW, conditioning on $\mathbf{X}_n$, and using the law of total variance,

\begin{equation}\label{eq:total-variance-semi-oracle}
    \operatorname{Var}\left[   \hat \tau_{\pi, \text{\tiny T}, n}^* \right] =  \operatorname{Var}\left[ \mathbb{E}\left[ \hat \tau_{\pi, \text{\tiny T}, n}^*\mid \mathbf{X}_n \right] \right] + \mathbb{E}\left[ \operatorname{Var}\left[ \hat \tau_{\pi, \text{\tiny T}, n}^* \mid \mathbf{X}_n  \right] \right].
\end{equation}

For the first inside term,
 \begin{align*}
     \mathbb{E}\left[ \hat \tau_{\pi, \text{\tiny T}, n}^*  \mid \mathbf{X}_n  \right]  &=    \mathbb{E}\left[ \sum_{x\in\mathds{X}} \frac{p_\text{\tiny T}(x) \mathds{1}_{Z_n(x)>0}}{\hat p_{\text{\tiny R},n}(x)} \frac{1}{n}\sum_{i =1}^n \mathbbm{1}_{X_i = x} \left(\frac{ Y_i^{(1)} A_i}{\pi} - \frac{ Y_i^{(0)} (1-A_i)}{1-\pi} \right)\mid  \mathbf{X}_n  \right] && \text{By definition (and SUTVA)} \\
     &= \sum_{x\in\mathds{X}} \frac{p_\text{\tiny T}(x)\mathds{1}_{Z_n(x)>0}}{\hat p_{\text{\tiny R},n}(x)} \frac{1}{n}  \sum_{i =1}^n \mathbbm{1}_{X_i = x}\mathbb{E}\left[  \left(\frac{ Y_i^{(1)} A_i}{\pi} - \frac{ Y_i^{(0)} (1-A_i)}{1-\pi} \right)\mid  \mathbf{X}_n  \right] && \text{Linearity of $\mathbb{E}[.]$}  \\
     &= \sum_{x\in\mathds{X}} \frac{p_\text{\tiny T}(x)\mathds{1}_{Z_n(x)>0}}{\hat p_{\text{\tiny R},n}(x)} \frac{1}{n}  \sum_{i =1}^n \mathbbm{1}_{X_i = x} \tau(X_i) &&    \\
      &= \sum_{x\in\mathds{X}} p_\text{\tiny T}(x) \tau(x) \mathds{1}_{Z_n(x) > 0} && \text{Equation~\ref{eq-proof:dep-hat-pr}} \\
      &= \mathbb{E}_{\text{\tiny T}} \left[ \tau(X) \mathds{1}_{Z_n(X) >0} | \mathbf{X}_n \right] && \text{Re-writing the sum as expectancy}.
     \end{align*}
where $X$ is nothing but a generic random variable, independent of $X_1, \hdots , X_n$ and distributed according to the density $p_T$. 
     
\begin{align}
\operatorname{Var} \left[ \mathbb{E}_{\text{\tiny T}}  \left[   \tau(X) \mathbbm{1}_{Z_n(X) > 0} | \mathbf{X}_n \right] \right] &= 
\operatorname{Var} \left[ \mathbb{E}_{\text{\tiny T}}  \left[   \tau(X) \right] - \mathbb{E}_{\text{\tiny T}}  \left[   \tau(X)\mathbbm{1}_{Z_n(X) = 0} | \mathbf{X}_n \right] \right] \nonumber \\
&= \operatorname{Var} \left[\tau - \mathbb{E}_{\text{\tiny T}}  \left[  \tau(X) \mathbbm{1}_{Z_n(X) = 0} | \mathbf{X}_n \right] \right] \nonumber \\
& = \operatorname{Var} \left[ \mathbb{E}_{\text{\tiny T}}  \left[  \tau(X) \mathbbm{1}_{Z_n(X) = 0} | \mathbf{X}_n \right] \right], \label{eq_proof_var_est_tau_indic}
\end{align}
as the only source of randomness comes from $ \mathbb{E}_{\text{\tiny T}}  \left[ \tau(X)  \mathbbm{1}_{Z_n(X) = 0} | \mathbf{X}_n \right] $. 
Therefore, the first inside term of \eqref{eq:total-variance-semi-oracle} corresponds to,
     \begin{equation}\label{eq:proof-semi-oracle-inter-eq-null}
         \operatorname{Var}\left[ \mathbb{E} \left[ \hat \tau_{\pi,\text{\tiny T}, n}^*\mid \mathbf{X}_n \right] \right]
          = \operatorname{Var}\left[ \mathbb{E}_{\text{\tiny T}} \left[ \tau(X) \mathds{1}_{Z_n(X)=0} | \mathbf{X}_n \right] \right].
     \end{equation}

On the other hand, 

 \begin{align*}
      \operatorname{Var}\left[ \hat \tau_{\pi,\text{\tiny T}, n}^* \mid  \mathbf{X}_n  \right] &=       \operatorname{Var}\left[ \frac{1}{n} \sum_{i = 1}^n \frac{p_\text{\tiny T}\left( X_i \right)}{\hat p_{\text{\tiny R},n}\left( X_i \right)}\left(\frac{ Y_i A_i}{\pi} - \frac{ Y_i (1-A_i)}{1-\pi} \right) \mid  \mathbf{X}_n  \right] \\
      &=  \frac{1}{n^2} \sum_{i = 1}^n \left( \frac{p_\text{\tiny T}\left( X_i \right)}{\hat p_{\text{\tiny R},n}\left( X_i \right)}\right)^2 \operatorname{Var}\left[\left(\frac{ Y_i A_i}{\pi} - \frac{ Y_i (1-A_i)}{1-\pi} \right)  \mid  \mathbf{X}_n \right] 
\\
&=  \frac{1}{n^2} \sum_{i = 1}^n \left( \frac{p_\text{\tiny T}\left( X_i \right)}{\hat p_{\text{\tiny R},n}\left( X_i \right)}\right)^2 \operatorname{Var}\left[\left(\frac{ Y_i A_i}{\pi} - \frac{ Y_i (1-A_i)}{1-\pi} \right)  \mid  X_i \right] \\
&=  \frac{1}{n^2} \sum_{i = 1}^n \left( \frac{p_\text{\tiny T}\left( X_i \right)}{\hat p_{\text{\tiny R},n}\left( X_i \right)}\right)^2 V_{\text{\tiny HT}}(X_i),
      \end{align*}
      where the last line comes from intermediary results in the completely oracle proof (see equation~\eqref{eq_explanation_for_g(x)_VHTx}), with

   \begin{equation*}
     V_{ \text{\tiny HT}}(x) := \mathbb{E}_{\text{\tiny R}}\left[ \frac{\left( Y^{(1)} \right)^2}{\pi}  \mid X_i \right]  + \mathbb{E}_{\text{\tiny R}}\left[ \frac{\left( Y^{(0)} \right)^2}{1-\pi}  \mid X_i  \right]  - \tau(X)^2.
\end{equation*}

      Then,
      
   \begin{align*}
       \mathbb{E}\left[ \operatorname{Var}\left[ \hat \tau_{\pi,\text{\tiny T}, n}^* \mid  \mathbf{X}_n  \right]  \right]&=\mathbb{E}\left[\frac{1}{n^2} \sum_{i = 1}^n \mathds{1}_{Z_n(x)=0} \left( \frac{p_\text{\tiny T}\left( X_i \right)}{\hat p_{\text{\tiny R},n}\left( X_i \right)}\right)^2 V_{\text{\tiny HT}}(X_i) \right] && \text{From previous derivations}\\
       &=\mathbb{E}\left[ \sum_{x\in\mathds{X}} \mathds{1}_{Z_n(x)=0} \left(  \frac{1}{n^2} \sum_{i = 1}^n  \mathbbm{1}_{X_i = x}  \left( \frac{p_\text{\tiny T}\left( X_i \right)}{\hat p_{\text{\tiny R},n}\left( X_i \right)}\right)^2 V_{\text{\tiny HT}}(X_i) \right) \right] && \text{Categorical $X$}\\
       &= \mathbb{E}\left[ \sum_{x\in\mathds{X}}  \frac{\mathds{1}_{Z_n(x)=0}}{n^2} \left( \frac{p_\text{\tiny T}\left( x \right)}{\hat p_{\text{\tiny R},n}\left( x \right)}\right)^2 V_{\text{\tiny HT}}(x) \left( \sum_{i = 1}^n  \mathbbm{1}_{X_i = x} \right) \right] \\
       &=   \sum_{x\in\mathds{X}}  \frac{1}{n^2} p_\text{\tiny T}\left( x \right)^2 V_{\text{\tiny HT}}(x)  \mathbb{E}\left[ \left( \frac{\mathds{1}_{Z_n(x)=0}}{\hat p_{\text{\tiny R},n}\left( x \right)}\right)^2  \left( \sum_{i = 1}^n  \mathbbm{1}_{X_i = x} \right) \right].
   \end{align*}

 Replacing $\hat p_{\text{\tiny R},n}\left( x \right)$ by its explicit expression,
   \begin{align}\label{eq:ineq-proof-so}
          \mathbb{E}\left[ \operatorname{Var}\left[ \hat \tau_{\pi,\text{\tiny T}, n}^* \mid  \mathbf{X}_n  \right] \right] & =  \frac{1}{n} \sum_{x\in\mathds{X}}   p_\text{\tiny T}\left( x \right)^2 V_{\text{\tiny HT}}(x) \mathbb{E}\left[ \left( \frac{\mathds{1}_{Z_n(x)=0}}{ \frac{1}{n}\sum_{i = 1}^n  \mathbbm{1}_{X_i = x} }\right)^2  \left( \frac{1}{n}\sum_{i = 1}^n  \mathbbm{1}_{X_i = x} \right) \right] \\
          & =  \frac{1}{n} \sum_{x\in\mathds{X}}   p_\text{\tiny T}\left( x \right)^2 V_{\text{\tiny HT}}(x) \mathbb{E}\left[  \frac{\mathds{1}_{Z_n(x)=0}}{ \frac{1}{n}\sum_{i = 1}^n  \mathbbm{1}_{X_i = x} }  \right].
   \end{align}

Recalling \eqref{eq:total-variance-semi-oracle} and \eqref{eq:proof-semi-oracle-inter-eq-null},
we have
\begin{align}
 \operatorname{Var}\left[ \hat \tau_{\pi, \text{\tiny T}, n}^*\right] & = \operatorname{Var}\left[ \mathbb{E}\left[ \hat \tau_{\pi,\text{\tiny T}, n}^*\mid \mathbf{X}_n \right] \right] +   \mathbb{E}\left[ \operatorname{Var}\left[ \hat \tau_{\pi, \text{\tiny T}, n}^*\mid  \mathbf{X}_n  \right]  \right] \nonumber \\
 & = \operatorname{Var}\left[ \mathbb{E}_{\text{\tiny T}} \left[ \tau(X) \mathds{1}_{Z_n(X)=0} | \mathbf{X}_n \right] \right] + \frac{1}{n} \sum_{x\in\mathds{X}}   p_\text{\tiny T}\left( x \right)^2 V_{\text{\tiny HT}}(x) \mathbb{E}\left[ \frac{\mathbbm{1}_{Z_n(x)>0}}{  \frac{1}{n} \sum_{i = 1}^n  \mathbbm{1}_{X_i = x} } \right]. \label{eq:inter-semi-oracle-before-convergence-variance}
\end{align}



\textbf{Upper bound on the variance}\\

According to \cite{Arnould2021Analyzing} (see Lemma S5, Supplementary Material, page 27), since $Z_{n}(x)$ is distributed as $\mathfrak{B}(n, p_{\text{\tiny R}}(x))$, we have

\begin{equation*}
    \forall x\in \mathds{X},\, 
    \mathbb{E}\left[\frac{\mathbbm{1}_{Z_{n}(x)\neq 0}}{Z_{n}(x)}\right] \leq \frac{2}{(n+1) p_\text{\tiny R}\left( x \right)}.
\end{equation*}
Besides, 
\begin{align}
\operatorname{Var}\left[ \mathbb{E}_{\text{\tiny T}} \left[ \tau(X) \mathds{1}_{Z_n(X)=0} | \mathbf{X}_n \right] \right] & =     \operatorname{Var}\left[ 
\sum_{x \in \mathds{X}} \tau(x) p_{\text{\tiny T}}(x) \mathds{1}_{Z_n(x) =0} \right] \\
& = \operatorname{Var}\left[ \begin{pmatrix} \tau(x_1) p_{\text{\tiny T}}(x_1) \\ \tau(x_2) p_{\text{\tiny T}}(x_2) \\ \vdots  \end{pmatrix}^\top \begin{pmatrix} \mathds{1}_{Z_n(x_1) =0} \\ \mathds{1}_{Z_n(x_2) =0} \\ \vdots  \end{pmatrix} \right] \\
& =  \begin{pmatrix} \tau(x_1) p_{\text{\tiny T}}(x_1) \\ \tau(x_2) p_{\text{\tiny T}}(x_2) \\ \vdots  \end{pmatrix}^\top \operatorname{Var}\left[ \begin{pmatrix} \mathds{1}_{Z_n(x_1) =0} \\ \mathds{1}_{Z_n(x_2) =0} \\ \vdots  \end{pmatrix} \right] \begin{pmatrix} \tau(x_1) p_{\text{\tiny T}}(x_1) \\ \tau(x_2) p_{\text{\tiny T}}(x_2) \\ \vdots  \end{pmatrix},
\end{align}
where the term in the middle is the covariance matrix of the vector $(\mathds{1}_{Z_n(x_1) =0}, \hdots, \mathds{1}_{Z_n(x_p) =0})$, where $p$ is the cardinality of $\mathds{X}$. Direct calculations and Cauchy-Schwarz inequality show that each term of the covariance matrix is upper bounded by 
\begin{align}
\alpha = \left( 1 - \min_{x \in \mathbb{X}} p_\text{\tiny R}(x)\right)^n 
\end{align}
Thus, letting $J = \mathbf{1}_{d \times d}$ the matrix of size $d \times d$ containing ones only, we have 
\begin{align}
\operatorname{Var}\left[ \mathbb{E}_{\text{\tiny T}} \left[ \tau(X) \mathds{1}_{Z_n(X)=0} | \mathbf{X}_n \right] \right] 
& \leq  \alpha \begin{pmatrix} |\tau(x_1)| p_{\text{\tiny T}}(x_1) \\ |\tau(x_2)| p_{\text{\tiny T}}(x_2) \\ \vdots  \end{pmatrix}^\top J \begin{pmatrix} |\tau(x_1) |p_{\text{\tiny T}}(x_1) \\ |\tau(x_2)| p_{\text{\tiny T}}(x_2) \\ \vdots  \end{pmatrix}\\
& \leq \alpha \left( \sum_{x \in \mathds{X}} |\tau(x) | p_{\text{\tiny T}}(x) \right)^2\\
& \leq \left( 1 - \min_{x \in \mathbb{X}} p_\text{\tiny R}(x)\right)^n \left(\mathds{E}_{\text{\tiny T}} \left[ |\tau(X)| \right]\right)^2.
\end{align}

Combining these inequalities with \eqref{eq:inter-semi-oracle-before-convergence-variance} yields, for all $n$,
%

\begin{align*}
\operatorname{Var}\left[   \hat \tau_{\pi, \text{\tiny T}, n}^* \right] \le \left(\mathds{E}_{\text{\tiny T}} \left[ |\tau(X)| \right]\right)^2 \left( 1 - \min_{x \in \mathbb{X}} p_\text{\tiny R}(x)\right)^n + \frac{2}{n+1} \sum_{x\in\mathds{X}}   \frac{p_\text{\tiny T}\left( x \right)^2 }{p_\text{\tiny R}\left( x \right)}V_{ \text{\tiny HT}}(x).
\end{align*}

This expression can be further simplified in,

\begin{align*}
\operatorname{Var}\left[   \hat \tau_{\pi, \text{\tiny T}, n}^* \right] \le \frac{2 V_{so}}{n+1}  +  \left( 1 - \min_{x \in \mathbb{X}} p_\text{\tiny R}(x)\right)^n \left(\mathds{E}_{\text{\tiny T}} \left[ |\tau(X)| \right]\right)^2,
\end{align*}

where 
\begin{align*}
    V_{so}:= \sum_{x\in\mathds{X}}   \frac{p_\text{\tiny T}\left( x \right)^2 }{p_\text{\tiny R}\left( x \right)}V_{ \text{\tiny HT}}(x) =  \mathbb{E}_{\text{\tiny T}}\left[ \left( \frac{p_\text{\tiny T}(X)}{p_\text{\tiny R}(X)} \right)^2  V_{ \text{\tiny HT}}(X)\right].
\end{align*}




\end{proof}

\subsubsection{Proof of Corollary~\ref{cor_asympt_semi_oracle}}\label{proof:asympt-bias-variance-semi-oracle}

\begin{proof}

 \textbf{Asymptotically unbiased}\\

 Recall the expression of the semi-oracle IPSW bias from Proposition~\ref{prop_semi_oracle_bias_variance}.
 
 \begin{align*}
\mathbb{E}\left[\hat \tau_{\pi, \text{\tiny T}, n}^*  \right]
&= \sum_{x\in\mathds{X}} p_\text{\tiny T}(x) \tau(x) \left(1 - (1 - p_{\text{\tiny R}}(x)\right)^n). 
\end{align*}

According to Assumption~\ref{a:pos}, we have 
$\forall x \in \mathds{X}$, $0 < p_{\text{\tiny R}}(x) < 1$. As a consequence,
\begin{equation*}
     \lim_{n\to\infty}  \left(1 - (1 - p_{\text{\tiny R}}(x)\right)^n = 1,
\end{equation*}

which leads to

\begin{equation*}
    \lim_{n\to\infty}  \mathbb{E}\left[\hat \tau_{\pi, \text{\tiny T}, n}^*  \right]  = \tau.
\end{equation*}

 \textbf{Asymptotic variance}\\
 
 Recall the expression of the variance of the semi-oracle IPSW from Proposition~\ref{prop_semi_oracle_bias_variance}:
 
 \begin{equation}
    n \operatorname{Var}\left[ \hat \tau_{\pi, \text{\tiny T}, n}^*\right]  = n \operatorname{Var}\left[ \mathbb{E}\left[ \hat \tau_{\pi,\text{\tiny T}, n}^*\mid \mathbf{X}_n \right] \right] +  \sum_{x\in\mathds{X}} p_\text{\tiny T}\left( x \right)^2 V_{\text{\tiny HT}}(x) \mathbb{E}_\text{\tiny R}\left[ \frac{\mathbbm{1}_{Z_n(x)>0}}{  \frac{1}{n} \sum_{i = 1}^n  \mathbbm{1}_{X_i = x} } \right].
\end{equation}
Note that the first term tends to zero since 
\begin{align*}
0 \leq n \operatorname{Var}\left[ \mathbb{E}\left[ \hat \tau_{\pi,\text{\tiny T}, n}^*\mid \mathbf{X}_n \right] \right] \leq \left(\mathds{E}_{\text{\tiny T}} \left[ |\tau(X)| \right]\right)^2  \left( 1 - \min_{x \in \mathbb{X}} p_\text{\tiny R}(x)\right)^n.
\end{align*}
Therefore, 
 \begin{equation}\label{eq:semi-oracle-before-convergence-variance-recall}
    \lim\limits_{n \to \infty} n \operatorname{Var}\left[ \hat \tau_{\pi, \text{\tiny T}, n}^*\right]  = \lim\limits_{n \to \infty}  \sum_{x\in\mathds{X}} p_\text{\tiny T}\left( x \right)^2 V_{\text{\tiny HT}}(x) \mathbb{E}_\text{\tiny R}\left[ \frac{\mathbbm{1}_{Z_n(x)>0}}{  \frac{1}{n} \sum_{i = 1}^n  \mathbbm{1}_{X_i = x} } \right].
\end{equation}
The next part of the proof consists in characterizing how the term $\mathbb{E}_\text{\tiny R}\left[ \frac{\mathbbm{1}_{Z_n(x)>0}}{  Z_n(x)/n } \right] $ converges. Let $\varepsilon>0$. Since, for all $x$, $p_{\text{\tiny R}}(x)>0$, we have 
 \begin{align}
 \mathbb{E}\left[ \frac{\mathbbm{1}_{Z_n(x)>0}}{  \frac{Z_n(x)}{n} } \right] & =  \mathbb{E}\left[ \frac{\mathbbm{1}_{Z_n(x)>0}}{  \frac{Z_n(x)}{n} } \mathbbm{1}_{|\frac{Z_n(x)}{n} - p_{\text{\tiny R}}(x)| \geq \varepsilon}\right] +   \mathbb{E}\left[ \frac{\mathbbm{1}_{Z_n(x)>0}}{  \frac{Z_n(x)}{n} } \mathbbm{1}_{| \frac{Z_n(x)}{n} - p_{\text{\tiny R}}(x)|< \varepsilon}\right].  \label{eq_proof_Chernoff}
 \end{align}
Regarding the first term in \eqref{eq_proof_Chernoff}, we have
\begin{align*}
\mathbb{E}\left[ \frac{\mathbbm{1}_{Z_n(x)>0}}{  \frac{Z_n(x)}{n} } \mathbbm{1}_{|\frac{Z_n(x)}{n} - p_{\text{\tiny R}}(x)| \geq \varepsilon}\right]
& \leq n \mathbb{P}\left[ |\frac{Z_n(x)}{n} - p_{\text{\tiny R}}(x)| \geq \varepsilon\right],
\end{align*}
since, on the event $Z_n(x) >0$, $Z_n(x) \geq 1$. Now, by Chernoff's inequality, 
\begin{align*}
\mathbb{P}\left[ |\frac{Z_n(x)}{n} - p_{\text{\tiny R}}(x)| \geq \varepsilon \right] \leq 2 \exp\left( - 2 \varepsilon^2 n \right),    
\end{align*}
which yields
\begin{align}
 \mathbb{E}\left[ \frac{\mathbbm{1}_{Z_n(x)>0}}{  \frac{Z_n(x)}{n} } \mathbbm{1}_{|\frac{Z_n(x)}{n} - p_{\text{\tiny R}}(x)| \geq \varepsilon}\right] \leq  2 n \exp\left( - 2 \varepsilon^2 n \right).  \label{eq_proof_Chernoff1}
\end{align}
Regarding the second term in equation \eqref{eq_proof_Chernoff}, since 
\begin{align*}
 \frac{\mathbbm{1}_{Z_n(x)>0}}{  \frac{Z_n(x)}{n} } \mathbbm{1}_{|\frac{Z_n(x)}{n} - p_{\text{\tiny R}}(x)|< \varepsilon}   
\end{align*}
is bounded above, for $\varepsilon < p_{\text{\tiny R}}(x)/2$ and converges in probability to $1/p_{\text{\tiny R}}(x)$, we have 
\begin{align}
\mathbb{E}\left[ \frac{\mathbbm{1}_{Z_n(x)>0}}{ \frac{Z_n(x)}{n} } \mathbbm{1}_{|\frac{Z_n(x)}{n}  - p_{\text{\tiny R}}(x)|< \varepsilon}\right] \to \frac{1}{p_{\text{\tiny R}}(x)}, \quad \textrm{as}~n \to \infty. \label{eq_proof_Chernoff2}
\end{align}
Combining \eqref{eq_proof_Chernoff1} and \eqref{eq_proof_Chernoff2}, we have 
\begin{align*}
\mathbb{E}\left[ \frac{\mathbbm{1}_{Z_n(x)>0}}{  Z_n(x)/n } \right] \to \frac{1}{p_{\text{\tiny R}}(x)}, \quad \textrm{as}~n \to \infty.    
\end{align*}

Using equation \eqref{eq:semi-oracle-before-convergence-variance-recall}, we finally obtain
   \begin{equation*}
       \lim_{n\to\infty}  n  \operatorname{Var}\left[ \hat \tau_{\pi, \text{\tiny T}, n}^*\right] = \sum_{x\in\mathds{X}} \frac{p_\text{\tiny T}\left( x \right)^2}{p_\text{\tiny R}\left( x \right)}V_{\text{\tiny HT}}(x) = \mathbb{E} \left[ \left( \frac{p_\text{\tiny T}\left( X \right)}{p_\text{\tiny R}\left( X \right)}\right)^2 V_{\text{\tiny HT}}(X)\right] := V_{\text{\tiny so}}.
   \end{equation*}
\end{proof}

 \subsubsection{Proof of Theorem~\ref{thm:semi-oracle}}\label{proof:risk-and-bound-semi-oracle}
 
 \begin{proof}
 For any estimate $\hat{\tau}$, we have 
\begin{equation*}
\mathbb{E}\left[ \left( \hat  \tau - \tau \right)^2\right] = \left( \mathbb{E}\left[ \hat  \tau \right] - \tau \right)^2 + \operatorname{Var}\left[   \hat  \tau \right].
\end{equation*}
Therefore, the risk of the semi-oracle IPSW estimate can be bounded using results from Subsection~\ref{proof:explicit-bias-variance-and-bounds-semi-oracle} (or Proposition~\ref{prop_semi_oracle_bias_variance}), and in particular the bounds on the variance and the bias,
\begin{align*}
\mathbb{E}\left[ \left( \hat  \tau - \tau \right)^2\right] & \leq \left(1 - \min_x p_\text{\tiny R}(x) \right)^{2n} \mathbb{E}_{\text{\tiny T}} \left[  |\tau(X)|\right]^2 + \frac{2 V_{so}}{n+1} +  \left( 1 - \min_{x \in \mathbb{X}} p_\text{\tiny R}(x)\right)^n \left(\mathds{E}_{\text{\tiny T}} \left[ |\tau(X)| \right]\right)^2 \\
& \leq  \frac{2 V_{so}}{n+1} + 2  \left( 1 - \min_{x \in \mathbb{X}} p_\text{\tiny R}(x)\right)^n \left(\mathds{E}_{\text{\tiny T}} \left[ |\tau(X)| \right]\right)^2,
\end{align*}
The $L^2$ consistency holds by letting $n$ tend to infinity. 

 \end{proof}

\subsection{Proofs for (estimated) IPSW $\hat \tau_{\pi,n,m}$}\label{proof:ipsw}

We first recall the definition of a fully estimated estimator introduced in Definition~\ref{def:ipsw}.

\begin{equation*}
    \hat \tau_{\pi, n,m} =   \frac{1}{n} \sum_{i = 1}^n \frac{\hat p_{\text{\tiny T},m }\left( X_i \right)}{\hat  p_{\text{\tiny R},n} (X_i)}\left(\frac{ Y_i A_i}{\pi} - \frac{ Y_i (1-A_i)}{1-\pi} \right),
\end{equation*}

where, for all $x\in \mathds{X}$, \begin{equation}\label{eq-proof:dep-hat-pr-pt}
    \hat p_{\text{\tiny R},n}\left( x \right)= \frac{\sum_{i=1}^n \mathbbm{1}_{X_i = x}}{n},\qquad  \text{and}\qquad \hat p_{\text{\tiny T},m}\left( x \right)= \frac{\sum_{i=n+1}^{n+m} \mathbbm{1}_{X_i = x}}{m}.
\end{equation}

Similar to the completely oracle estimator, this estimated IPSW can be written as,

\begin{equation*}
     \hat \tau_{\pi, n,m} = \sum_{x\in\mathds{X}} \frac{\hat p_{\text{\tiny T},m }\left( x \right) \mathbbm{1}_{Z_{n}(x) > 0} }{\hat  p_{\text{\tiny R},n}(x)} \left( \frac{1}{n}\sum_{i =1}^n \mathbbm{1}_{X_i = x} \left(\frac{ Y_i A_i}{\pi} - \frac{ Y_i (1-A_i)}{1-\pi} \right)  \right).
\end{equation*}

All the proofs below rely on this decomposition.

\subsubsection{Proof of Proposition~\ref{prop_completely_estimated}}\label{proof_prop_completely_estimated}

We prove the following, more general, proposition. 

\begin{proposition}
\label{prop:appendix_soIPSW}
Under the general setting defined in Subsection~\ref{subsec:model}, granting Assumptions~\ref{a:repres-rct}-\ref{a:pos}, the bias of the estimated IPSW is the same as that of the semi-oracle IPSW, that is,  for all $n, m$,  
\begin{align*}
  \mathbb{E}\left[\hat \tau_{\pi, n,m}  \right] - \tau \; &= \; - \sum_{x\in\mathds{X}} p_{\text{\tiny T}}(x)\, (1 - p_{\text{\tiny R}}(x))^n \,\tau(x).
\end{align*}
Moreover, under the same set of assumptions, the variance of the estimated IPSW satisfies, for all $n,m$, 
\begin{multline*}
    \operatorname{Var}\left[   \hat \tau_{\pi, n,m} \right] \,=\, \operatorname{Var}\left[\hat \tau_{\pi, \text{\tiny T},n}^* \right]  + \frac{1}{m}\left(  \operatorname{Var}_{\text{\tiny T}}\left[  \tau(X) \mathbbm{1}_{Z_n(X) \neq 0} \right] - \operatorname{Var} \left[ \mathbb{E}_{\text{\tiny T}} \left[ \tau(X) \mathbbm{1}_{Z_n(X) = 0} | \mathbf{X}_n \right] \right]  \right) \\
    +  \frac{1}{n\,m}  \sum_{x\in\mathds{X}}  V_{\text{\tiny HT}} (x) \, p_{\text{\tiny T}}(x) \,(1-p_{\text{\tiny T}}(x)) \,\mathbb{E} \left[  \frac{\mathbbm{1}_{Z_n(x) \neq 0} }{\hat p_{\text{\tiny R},n}\left( x \right)}  \right]
\end{multline*}
\begin{flalign}
\text{and}&&
     \operatorname{Var}\left[   \hat \tau_{\pi, n,m} \right]  &\le  \frac{2V_{so}}{n+1}  +  \frac{ \operatorname{Var}_{\text{\tiny T}} \left[ \tau(X)\right]}{m}  +  \frac{2}{m\left(n+1\right)}\mathbb{E}_{\text{\tiny R}}\left[ \frac{p_\text{\tiny T}\left( X \right)(1-p_\text{\tiny T}\left( X \right))}{p_\text{\tiny R}\left( X \right)^2} V_{\text{\tiny HT}}(X)\right]   \nonumber 
&&\\&&
     & \qquad +  \left(1 - \min_x p_{\text{\tiny R}}(x) \right)^{n/2}\mathbb{E}_{\text{\tiny T}} \left[   \tau(X)^2 \right]\left( 1 +  \frac{4}{m} \right).
&&
\end{flalign}
\end{proposition}

\begin{proof}

\textbf{Expression of the bias}\\

Using the exact same derivations as in Subsection~\ref{proof:explicit-bias-variance-and-bounds-semi-oracle} (Bias), \textbf{but} using the law of total expectation when conditioning on $\mathbf{X}_{n+m}\in \mathds{X}^{n+m}$ (i.e. comprising the $n$ and $m$ observations $X_1, X_2, \dots X_n, X_{n+1} \dots X_{n+m}  \in \mathds{X}$ in the trial and target population, one has, 
\begin{align*}
  \mathbb{E}\left[\hat \tau_{\pi, n,m}  \right] &= \sum_{x\in\mathds{X}}\mathbb{E}\left[ \frac{\hat p_{\text{\tiny T},m}(x) \mathbbm{1}_{Z_{n}(x) > 0}}{\hat  p_{\text{\tiny R},n}(x)} \frac{\sum_{i =1}^n \mathbbm{1}_{X_i = x}}{n}\tau(x) \right] \\
     &= \sum_{x\in\mathds{X}}\mathbb{E}\left[ \frac{\hat p_{\text{\tiny T},m}(x) \mathbbm{1}_{Z_{n}(x) > 0}}{\frac{\sum_{i =1}^n \mathbbm{1}_{X_i = x}}{n}} \frac{\sum_{i =1}^n \mathbbm{1}_{X_i = x}}{n}\tau(x) \right] && \text{Estimation procedure - Equation~\ref{eq-proof:dep-hat-pr-pt}} \\
     & = \sum_{x\in\mathds{X}}\mathbb{E} \left[\hat p_{\text{\tiny T},m}(x) \tau(x) \mathbbm{1}_{Z_{n}(x) > 0} \right].
\end{align*}


Note that $Z_n(x)$ only depend on the trial sample $\mathcal{R}$ and $\hat p_{\text{\tiny T},m}(x)$ on the observational sample. In addition, $\tau(x)$ is deterministic, therefore

\begin{align*}
  \mathbb{E}\left[\hat \tau_{\pi, n,m}  \right] &= \sum_{x\in\mathds{X}} \tau(x) \mathbb{E}\left[\hat p_{\text{\tiny T},m}(x) \right]  \mathbb{E}\left[\mathbbm{1}_{Z_{n}(x) > 0} \right].
\end{align*}

Note that $\mathbb{E} \left[ \hat p_{\text{\tiny T},m}(x) \right]  = p_{\text{\tiny T}}(x)$. Besides, according to the proof of the semi-oracle IPSW, 

%
%
\begin{align*}
    \mathbb{E} \left[\mathbbm{1}_{Z_{n}(x) > 0} \right] = \left(1 - (1 - p_{\text{\tiny R}}(x)\right)^n).
\end{align*}
Therefore, 
\begin{align*}
  \mathbb{E} \left[\hat \tau_{\pi, n,m}  \right] &= \sum_{x\in\mathds{X}} p_{\text{\tiny T}}(x)\tau(x) \left(1 - (1 - p_{\text{\tiny R}}(x)\right)^n),
\end{align*}
that is 
\begin{align*}
  \mathbb{E} \left[\hat \tau_{\pi, n,m}  \right] - \tau &= - \sum_{x\in\mathds{X}} p_{\text{\tiny T}}(x)\tau(x)  \left(1 - p_{\text{\tiny R}}(x)\right)^n.
\end{align*}

\textbf{Upper bound on the bias}\\

It is possible to bound the bias for any sample size $n$, using the exact same derivations than for the semi-oracle IPSW.\\


\textbf{Expression of the variance}\\

The proof follows a similar spirit as the proof for the completely oracle estimator, conditioning on all observations $\mathbf{X}_{n+m}$.

\begin{equation}
    \operatorname{Var}\left[   \hat \tau_{\pi, n,m} \right] =  \operatorname{Var}\left[ \mathbb{E}\left[\hat \tau_{\pi, n,m} \mid \mathbf{X}_{n+m} \right] \right] + \mathbb{E}\left[ \operatorname{Var}\left[ \hat \tau_{\pi, n,m}  \mid \mathbf{X}_{n+m}  \right] \right]. \label{eq_proof_comp_est_dec_var}
\end{equation}

\begin{align*}
    \mathbb{E}\left[\hat \tau_{\pi, n,m} \mid \mathbf{X}_{n+m} \right]  &= \mathbb{E} \left[\sum_{x\in\mathds{X}} \frac{\hat p_{\text{\tiny T},m }\left( x \right) \mathbbm{1}_{Z_{n}(x) > 0}}{\hat  p_{\text{\tiny R},n}(x)} \left( \frac{1}{n}\sum_{i =1}^n \mathbbm{1}_{X_i = x} \left(\frac{ Y_i A_i}{\pi} - \frac{ Y_i (1-A_i)}{1-\pi} \right)  \right)\mid \mathbf{X}_{n+m}  \right]  \\
    &= \sum_{x\in\mathds{X}} \mathbb{E} \left[\frac{\hat p_{\text{\tiny T},m }\left( x \right) \mathbbm{1}_{Z_{n}(x) > 0}}{\hat  p_{\text{\tiny R},n}(x)} \left( \frac{1}{n}\sum_{i =1}^n \mathbbm{1}_{X_i = x} \left(\frac{ Y_i A_i}{\pi} - \frac{ Y_i (1-A_i)}{1-\pi} \right)  \right) \mid \mathbf{X}_{n+m} \right] && \text{Linearity of $\mathbb{E}[.]$} \\
    &= \sum_{x\in\mathds{X}} \frac{\hat p_{\text{\tiny T},m }\left( x \right) \mathbbm{1}_{Z_{n}(x) > 0}}{\hat  p_{\text{\tiny R},n}(x)}  \mathbb{E} \left[ \frac{1}{n}\sum_{i =1}^n \mathbbm{1}_{X_i = x} \left(\frac{ Y_i A_i}{\pi} - \frac{ Y_i (1-A_i)}{1-\pi} \right)  \mid \mathbf{X}_{n+m}  \right].
\end{align*}

Indeed, $\hat p_{\text{\tiny R},n}(x)$  and $\hat p_{\text{\tiny T},m }(x)$ are measurable with respect to  $\mathbf{X}_{n+m}$. Pursuing the computation, we have

\begin{align*}
    \mathbb{E}\left[\hat \tau_{\pi, n,m} \mid \mathbf{X}_{n+m} \right]  &= \sum_{x\in\mathds{X}} \frac{\hat p_{\text{\tiny T},m }\left( x \right) \mathbbm{1}_{Z_{n}(x) > 0}}{\hat  p_{\text{\tiny R},n}(x)}  \mathbb{E} \left[ \frac{1}{n}\sum_{i =1}^n \mathbbm{1}_{X_i = x} \left(\frac{ Y_i A_i}{\pi} - \frac{ Y_i (1-A_i)}{1-\pi} \right)  \mid \mathbf{X}_{n+m}  \right] \\
    &= \sum_{x\in\mathds{X}} \frac{\hat p_{\text{\tiny T},m }\left( x \right) \mathbbm{1}_{Z_{n}(x) > 0}}{\hat  p_{\text{\tiny R},n}(x)}  \frac{1}{n}\sum_{i =1}^n  \mathbb{E} \left[ \mathbbm{1}_{X_i = x} \left(\frac{ Y_i A_i}{\pi} - \frac{ Y_i (1-A_i)}{1-\pi} \right)  \mid \mathbf{X}_{n+m}  \right] && \text{Linearity of $\mathbb{E}[.]$} \\
    &= \sum_{x\in\mathds{X}} \frac{\hat p_{\text{\tiny T},m }\left( x \right) \mathbbm{1}_{Z_{n}(x) > 0}}{\hat  p_{\text{\tiny R},n}(x)}  \frac{1}{n}\sum_{i =1}^n \mathbbm{1}_{X_i = x} \mathbb{E} \left[  \left(\frac{ Y_i A_i}{\pi} - \frac{ Y_i (1-A_i)}{1-\pi} \right)  \mid \mathbf{X}_{n+m}  \right] && \text{Conditioning on $\mathbf{X}_n$} \\
    &= \sum_{x\in\mathds{X}} \frac{\hat p_{\text{\tiny T},m }\left( x \right) \mathbbm{1}_{Z_{n}(x) > 0}}{\hat  p_{\text{\tiny R},n}(x)} \tau(x)  \frac{1}{n}\sum_{i =1}^n \mathbbm{1}_{X_i = x} && \text{Transportability}\\
    &= \sum_{x\in\mathds{X}}  \hat p_{\text{\tiny T},m }\left( x \right) \tau(x) \mathbbm{1}_{Z_{n}(x) > 0}, && 
\end{align*}
where $Z_n(x) = \sum_{i =1}^n \mathbbm{1}_{X_i = x}$. Then,

\begin{align*}
    \operatorname{Var}\left[ \mathbb{E}\left[\hat \tau_{\pi, n,m} \mid \mathbf{X}_{n+m} \right]  \right] &=  \operatorname{Var}\left[  \sum_{x\in\mathds{X}} \hat p_{\text{\tiny T},m }\left( x \right)\tau(x) \mathbbm{1}_{Z_{n}(x) > 0} \right] \\
    &=  \operatorname{Var}\left[  \sum_{x\in\mathds{X}} \frac{\sum_{i = n+1}^{n+m} \mathbbm{1}_{X_i = x} }{m}\tau(x) \mathbbm{1}_{Z_{n}(x) > 0} \right] \\
    &= \operatorname{Var}\left[  \frac{1}{m} \sum_{i = n+1}^{n+m}  \tau(X_i) \mathbbm{1}_{Z_{n}(x) > 0}  \right] 
\end{align*}

Note that, contrary to the semi-oracle IPSW, this term is non-null due to estimation of $\hat p_{\text{\tiny T},m}$. By the law of total variance, 
\begin{align*}
\operatorname{Var}\left[  \frac{1}{m} \sum_{i = n+1}^{n+m}  \tau(X_i) \mathbbm{1}_{Z_{n}(x) > 0}  \right] & = \mathbb{E} \left[ \operatorname{Var}\left[  \frac{1}{m} \sum_{i = n+1}^{n+m}  \tau(X_i) \mathbbm{1}_{Z_{n}(x) > 0} | \mathbf{X}_n \right] \right] \\
& \quad + \operatorname{Var} \left[ \mathbb{E} \left[  \frac{1}{m} \sum_{i = n+1}^{n+m}  \tau(X_i) \mathbbm{1}_{Z_{n}(x) > 0} | \mathbf{X}_n \right] \right] \\
& = \frac{1}{m} \mathbb{E} \left[ \operatorname{Var}\left[  \tau(X) \mathbbm{1}_{Z_{n}(x) > 0} | \mathbf{X}_n \right] \right] + \operatorname{Var} \left[ \mathbb{E} \left[   \tau(X) \mathbbm{1}_{Z_{n}(x) > 0} | \mathbf{X}_n \right] \right]\\
& = \frac{1}{m}  \operatorname{Var}\left[  \tau(X) \mathbbm{1}_{Z_{n}(x) > 0} \right] + \left( 1 - \frac{1}{m} \right) \operatorname{Var} \left[ \mathbb{E} \left[   \tau(X) \mathbbm{1}_{Z_{n}(x) > 0} | \mathbf{X}_n \right] \right],
\end{align*}
where the last line comes from the law of total variance applied to $\operatorname{Var}\left[  \tau(X) \mathbbm{1}_{Z_{n}(x) > 0}\right]$. Recalling similar derivations from the semi-oracle IPSW proof, and in particular \eqref{eq_proof_var_est_tau_indic}, one has
\begin{align*}
\operatorname{Var} \left[ \mathbb{E} \left[   \tau(X) \mathbbm{1}_{Z_{n}(x) > 0}| \mathbf{X}_n \right] \right] &=  \operatorname{Var} \left[ \mathbb{E}_{\text{\tiny T}} \left[  \tau(X) \mathbbm{1}_{Z_{n}(x) > 0} | \mathbf{X}_n \right] \right],
\end{align*}

so that
\begin{align}
    \operatorname{Var}\left[ \mathbb{E} \left[\hat \tau_{\pi, n,m} \mid \mathbf{X}_n \right]  \right] &=
    \frac{1}{m}  \operatorname{Var}\left[  \tau(X) \mathbbm{1}_{Z_{n}(x) > 0} \right] + \left( 1 - \frac{1}{m} \right) \operatorname{Var} \left[ \mathbb{E}_{\text{\tiny T}} \left[ \tau(X)  \mathbbm{1}_{Z_n(X) = 0} | \mathbf{X}_n \right] \right]. \label{eq_proof_comp_est_dec_var_term1}
\end{align}

For the other term of \eqref{eq_proof_comp_est_dec_var}, 
\begin{align*}
    \operatorname{Var}\left[ \hat \tau_{\pi, n,m}  \mid \mathbf{X}_{n+m}  \right]  &=   \operatorname{Var}\left[  \frac{1}{n} \sum_{i = 1}^n \frac{\hat p_{\text{\tiny T},m}\left( X_i \right)}{\hat p_{\text{\tiny R},n}\left( X_i \right)}\left(\frac{ Y_i A_i}{\pi} - \frac{ Y_i (1-A_i)}{1-\pi} \right)   \mid \mathbf{X}_{n+m} \right]  \\
    &= \frac{1}{n^2} \sum_{i = 1}^n \left( \frac{\hat p_{\text{\tiny T},m}\left( X_i \right)}{\hat p_{\text{\tiny R},n}\left( X_i \right)}\right)^2 V_{\text{\tiny HT}}(X_i).
      \end{align*}
   
Derivations are very similar to the semi-oracle estimator, using the fact that $\hat p_{\text{\tiny R},n}(x)$ and $\hat p_{\text{\tiny T},m }(x)$ are measurable with respect to $\mathbf{X}_{n+m}$. We have

   \begin{align*}
       \mathbb{E}\left[ \operatorname{Var}\left[ \hat \tau_{\pi, n,m}\mid  \mathbf{X}_{n+m}  \right]  \right]&=\mathbb{E} \left[\frac{1}{n^2} \sum_{i = 1}^n \left( \frac{\hat p_{\text{\tiny T},m}\left( X_i \right)}{\hat p_{\text{\tiny R},n}\left( X_i \right)}\right)^2 V_{\text{\tiny HT}}(X_i) \right] && \text{From previous derivations}\\
       &=\mathbb{E} \left[ \sum_{x\in\mathds{X}} \mathbbm{1}_{Z_{n}(x) > 0} \left(  \frac{1}{n^2} \sum_{i = 1}^n  \mathbbm{1}_{X_i = x}  \left( \frac{\hat p_{\text{\tiny T},m}\left( X_i \right)}{\hat p_{\text{\tiny R},n}\left( X_i \right)}\right)^2V_{\text{\tiny HT}}(X_i)\right) \right] && \text{Categorical $X$}\\
       &= \mathbb{E} \left[ \sum_{x\in\mathds{X}} \mathbbm{1}_{Z_{n}(x) > 0}  \frac{1}{n^2} \left( \frac{\hat p_{\text{\tiny T},m}\left( x \right)}{\hat p_{\text{\tiny R},n}\left( x \right)}\right)^2 V_{\text{\tiny HT}}(x)  \left( \sum_{i = 1}^n  \mathbbm{1}_{X_i = x} \right)  \right] \\
       &=   \sum_{x\in\mathds{X}}  \frac{1}{n^2}  V_{\text{\tiny HT}}(x) \mathbb{E} \left[ \mathbbm{1}_{Z_{n}(x) > 0}  \left( \frac{\hat p_{\text{\tiny T},m}\left( x \right)}{\hat p_{\text{\tiny R},n}\left( x \right)}\right)^2  \left( \sum_{i = 1}^n  \mathbbm{1}_{X_i = x} \right) \right] \\
       &=   \sum_{x\in\mathds{X}}  \frac{1}{n} V_{\text{\tiny HT}}(x) \mathbb{E}_\text{\tiny R}\left[ \mathbbm{1}_{Z_{n}(x) > 0} \frac{\left( \hat p_{\text{\tiny T},m}\left( x \right)\right)^2}{\hat p_{\text{\tiny R},n}\left( x \right)}   \right].
   \end{align*}

In particular, the last term can be simplified in
\begin{equation}\label{eq:ineq-proof-completely-estimated}
    \mathbb{E}\left[ \operatorname{Var}\left[ \hat \tau_{\pi, n,m}\mid  \mathbf{X}_{n+m} \right] \right] =  \sum_{x\in\mathds{X}}  \frac{1}{n}  V_{\text{\tiny HT}}(x) \mathbb{E} \left[  \left( \hat p_{\text{\tiny T},m}\left( x \right)\right)^2 \right] \mathbb{E} \left[  \frac{\mathbbm{1}_{Z_{n}(x) > 0} }{\hat p_{\text{\tiny R},n}\left( x \right)}  \right].
\end{equation}

This last derivation is possible because $\hat p_{\text{\tiny T},m}\left( x \right)$, which depends on $\mathcal{T}$, and $\hat p_{\text{\tiny R},n}\left( x \right)$, which depends on $\mathcal{R}$, are independent. The difference from the semi-oracle estimator comes from the term


\begin{align}
    \mathbb{E} \left[  \left( \hat p_{\text{\tiny T},m}\left( x \right)\right)^2 \right] &=  \mathbb{E} \left[ \left( \frac{ \sum_{i = n+1}^m  \mathbbm{1}_{X_i = x}}{m} \right)^2 \right] \nonumber \\
    &= \frac{1}{m^2}\mathbb{E} \left[ \left( \sum_{i = n+1}^m  \mathbbm{1}_{X_i = x} \right)^2  \right] \nonumber\\
    &= \frac{1}{m^2} \left( m p_{\text{\tiny T}}(x)(1-p_{\text{\tiny T}})(x) +  m^2 p_{\text{\tiny T}}^2 (x) \right) \nonumber \\
    &= \frac{p_{\text{\tiny T}}(x)(1-p_{\text{\tiny T}}(x))}{m} + p_{\text{\tiny T}}^2 (x). \label{eq_proof_comp_est_dec_var_term2}
\end{align}
Using \eqref{eq_proof_comp_est_dec_var_term1} and \eqref{eq_proof_comp_est_dec_var_term2} in \eqref{eq_proof_comp_est_dec_var}, we have
\begin{align}%
    \operatorname{Var}\left[   \hat \tau_{\pi, n,m} \right] & =  \operatorname{Var}\left[ \mathbb{E}\left[\hat \tau_{\pi, n,m} \mid \mathbf{X}_{n+m} \right] \right] + \mathbb{E}\left[ \operatorname{Var}\left[ \hat \tau_{\pi, n,m}  \mid \mathbf{X}_{n+m}  \right] \right]  \nonumber \\
    &= \frac{1}{m}  \operatorname{Var}_{\text{\tiny T}}\left[  \tau(X) \mathbbm{1}_{Z_{n}(x) > 0} \right] + \left( 1 - \frac{1}{m} \right) \operatorname{Var} \left[ \mathbb{E}_{\text{\tiny T}} \left[ \tau(X) \mathbbm{1}_{Z_n(X) = 0} | \mathbf{X}_n \right] \right]  \nonumber \\
    & \qquad +  \frac{1}{n}  \sum_{x\in\mathds{X}}  V_{\text{\tiny HT}} (x) \frac{p_{\text{\tiny T}}(x)(1-p_{\text{\tiny T}}(x))}{m}\mathbb{E} \left[  \frac{\mathbbm{1}_{Z_{n}(x) > 0} }{\hat p_{\text{\tiny R},n}\left( x \right)}  \right] + \frac{1}{n} \sum_{x\in\mathds{X}}  V_{\text{\tiny HT}} (x) p_{\text{\tiny T}}^2 (x)  \mathbb{E} \left[  \frac{\mathbbm{1}_{Z_{n}(x) > 0} }{\hat p_{\text{\tiny R},n}\left( x \right)}  \right]  \nonumber \\
    &= \frac{1}{m}\left(  \operatorname{Var}_{\text{\tiny T}}\left[  \tau(X) \mathbbm{1}_{Z_{n}(x) > 0} \right] - \operatorname{Var} \left[ \mathbb{E}_{\text{\tiny T}} \left[ \tau(X) \mathbbm{1}_{Z_n(X) = 0} | \mathbf{X}_n \right] \right]  \right)  + \operatorname{Var}\left[\hat \tau_{\pi, \text{\tiny T},n}^* \right] \nonumber \\
    & \qquad +  \frac{1}{nm}  \sum_{x\in\mathds{X}}  V_{\text{\tiny HT}} (x) p_{\text{\tiny T}}(x)(1-p_{\text{\tiny T}}(x)) \mathbb{E} \left[  \frac{\mathbbm{1}_{Z_{n}(x) > 0} }{\hat p_{\text{\tiny R},n}\left( x \right)}  \right]. \label{eq:explicit-variance-expression}
\end{align}

~\\

\textbf{Upper bound on the variance}.\\

We first bound \eqref{eq_proof_comp_est_dec_var_term1}, corresponding to

\begin{align*}
    \operatorname{Var}\left[ \mathbb{E} \left[\hat \tau_{\pi, n,m} \mid \mathbf{X}_{n+m} \right]  \right] &=
    \frac{1}{m}  \operatorname{Var}_{\text{\tiny T}}\left[  \tau(X) \mathbbm{1}_{Z_{n}(x) > 0} \right] + \left( 1 - \frac{1}{m} \right) \operatorname{Var} \left[ \mathbb{E} \left[ \tau(X) \mathbbm{1}_{Z_n(X) = 0} | \mathbf{X}_{n+m} \right] \right].
\end{align*}

We have
\begin{align*}
\operatorname{Var}_{\text{\tiny T}}\left[  \tau(X) \mathbbm{1}_{Z_n(X) \neq 0} \right]    & = \operatorname{Var}_{\text{\tiny T}}\left[  \tau(X)  -  \tau(X) \mathbbm{1}_{Z_n(X) = 0} \right] \\
& = \operatorname{Var}_{\text{\tiny T}} \left[ \tau(X)\right] - 2 \operatorname{Cov}_{\text{\tiny T}}(\tau(X), \tau(X) \mathbbm{1}_{Z_n(X) = 0}) + \operatorname{Var}_{\text{\tiny T}} \left[\tau(X) \mathbbm{1}_{Z_n(X) = 0}\right]\\
& \leq \operatorname{Var}_{\text{\tiny T}} \left[ \tau(X)\right] + 2 \left( \operatorname{Var}_{\text{\tiny T}}[\tau(X)] \operatorname{Var}_{\text{\tiny T}}\left[\tau(X) \mathbbm{1}_{Z_n(X) = 0} \right] \right)^{1/2} + \operatorname{Var}_{\text{\tiny T}} \left[ \tau(X)\mathbbm{1}_{Z_n(X) = 0}\right],
\end{align*}
with 
\begin{align*}
\operatorname{Var}_{\text{\tiny T}} \left[ \tau(X)\mathbbm{1}_{Z_n(X) = 0}\right] & \leq \mathbb{E} \left[ \tau(X)^2 \mathbbm{1}_{Z_n(X) = 0}\right] \\
& \leq \mathbb{E} \left[ \tau(X)^2 \mathbb{E} \left[\mathbbm{1}_{Z_n(X) = 0} \mid X \right]  \right]\\
& \leq \mathbb{E} \left[ \tau(X)^2 (1 - p_{\text{\tiny R}}(X))^n  \right]\\
& \leq \left(1 - \min_x p_{\text{\tiny R}}(x) \right)^n \mathbb{E}_{\text{\tiny T}} \left[ \tau(X)^2   \right].
\end{align*}
Consequently, 
\begin{align*}
\operatorname{Var}_{\text{\tiny T}}\left[  \tau(X) \mathbbm{1}_{Z_n(X) \neq 0} \right] & \leq \operatorname{Var}_{\text{\tiny T}} \left[ \tau(X)\right] + 2 \mathbb{E}_{\text{\tiny T}} \left[ \tau(X)^2 \right] \left(1 - \min_x p_{\text{\tiny R}}(x) \right)^{n/2} + \left(1 - \min_x p_{\text{\tiny R}}(x) \right)^n \mathbb{E}_{\text{\tiny T}} \left[ \tau(X)^2   \right] \\
& \leq \operatorname{Var}_{\text{\tiny T}} \left[ \tau(X)\right] + 4 \mathbb{E}_{\text{\tiny T}} \left[ \tau(X)^2 \right] \left(1 - \min_x p_{\text{\tiny R}}(x) \right)^{n/2}.
\end{align*}

One can also bound the other term of \eqref{eq_proof_comp_est_dec_var_term1} following the same derivations as the semi-oracle IPSW,

\begin{align}
     \operatorname{Var} \left[ \mathbb{E} \left[ \tau(X) \mathbbm{1}_{Z_n(X) = 0} | \mathbf{X}_{n+m} \right] \right] 
     &\le \left(\mathds{E}_{\text{\tiny T}} \left[ |\tau(X)| \right]\right)^2 \left(1 - \min_x p_{\text{\tiny R}}(x) \right)^{n}.
     \label{eq_proof_bound_var_exp_tau_ind}
\end{align}
Then, using the fact that $1-\frac{1}{m} \le 1$, 
\begin{align}
\operatorname{Var}\left[ \mathbb{E}_{\text{\tiny T}} \left[\hat \tau_{\pi, n,m} \mid \mathbf{X}_{n+m} \right]  \right] & \leq \frac{ \operatorname{Var}_{\text{\tiny T}} \left[ \tau(X)\right]}{m} + \frac{4 \mathbb{E}_{\text{\tiny T}} \left[ \tau(X)^2 \right]}{m} \left(1 - \min_x p_{\text{\tiny R}}(x) \right)^{n/2}  +  \left(\mathds{E}_{\text{\tiny T}} \left[ |\tau(X)| \right]\right)^2 \left(1 - \min_x p_{\text{\tiny R}}(x) \right)^{n} \nonumber\\
& \leq \frac{ \operatorname{Var}_{\text{\tiny T}} \left[ \tau(X)\right]}{m} +    \mathbb{E}_{\text{\tiny T}} \left[   \tau(X)^2 \right]\left( \frac{4}{m} +  1 \right) \left(1 - \min_x p_{\text{\tiny R}}(x) \right)^{n/2},\label{eq_proof_comp_est_dec_var_term1_bound}
\end{align}
using the fact that $\left(\mathds{E}_{\text{\tiny T}} \left[ |\tau(X)| \right]\right)^2 \leq  \mathds{E}_{\text{\tiny T}} \left[  \tau(X)^2 \right]$. Then, for the other term of the asymptotic variance, one can use the results from \cite{Arnould2021Analyzing} (see page 27) to bound the variance, which leads to 

\begin{align*}
         \mathbb{E} \left[ \operatorname{Var}\left[ \hat \tau_{\pi, n,m}\mid  \mathbf{X}_{n+m}  \right] \right] &=  \sum_{x\in\mathds{X}}  \frac{1}{n}  g(x) \mathbb{E} \left[  \left( \hat p_{\text{\tiny T},m}\left( x \right)\right)^2 \right] \mathbb{E} \left[  \frac{\mathbbm{1}_{Z_{n}(x) > 0} }{\frac{Z_n(x) }{n}} \right] \\
         & \le \sum_{x\in\mathds{X}}  g(x) \mathbb{E} \left[  \left( \hat p_{\text{\tiny T},m}\left( x \right)\right)^2 \right] \frac{2}{(n+1) p_\text{\tiny R}\left( x \right)} && \text{\cite{Arnould2021Analyzing} (p.27)} \\
         & = \sum_{x\in\mathds{X}}  g(x) \left( \frac{p_\text{\tiny T}\left( x \right) (1-p_\text{\tiny T}\left( x \right) ) }{m} +  p_\text{\tiny T}\left( x \right) ^2  \right)\frac{2}{(n+1) p_\text{\tiny R}\left( x \right)}
\end{align*}

Finally, using \eqref{eq_proof_comp_est_dec_var_term1_bound}, and \eqref{eq:explicit-variance-expression},

\begin{align}
     \operatorname{Var}\left[   \hat \tau_{\pi, n,m} \right]  &\le  \frac{ \operatorname{Var}_{\text{\tiny T}} \left[ \tau(X)\right]}{m} +    \mathbb{E}_{\text{\tiny T}} \left[   \tau(X)^2 \right]\left( \frac{4}{m} +  1 \right) \left(1 - \min_x p_{\text{\tiny R}}(x) \right)^{n/2}\nonumber \\
     & \qquad +  \frac{2}{n+1}\left(  \mathbb{E}_{\text{\tiny R}}\left[ \left( \frac{p_\text{\tiny T}\left( X \right)}{p_\text{\tiny R}\left( X \right)}\right)^2 V_{\text{\tiny HT}}(X)\right] + \frac{1}{m} \mathbb{E}_{\text{\tiny R}}\left[ \frac{p_\text{\tiny T}\left( X \right)(1-p_\text{\tiny T}\left( X \right))}{p_\text{\tiny R}\left( X \right)^2} V_{\text{\tiny HT}}(X)\right] \right). \label{eq_proof_bonus_ext_pi_est}
\end{align}

\end{proof}

\subsubsection{Proof of Corollary~\ref{cor_asympt_completely_estimated}}\label{proof_asympt_completely_estimated}

\textbf{Asymptotic bias}\\

The proof is exactly the same as for the semi-oracle IPSW, see Subsection~\ref{proof:asympt-bias-variance-semi-oracle}.\\

\textbf{Asymptotic variance}\\

We recall that the explicit expression of the variance is

\begin{align*}
    \operatorname{Var}\left[   \hat \tau_{\pi, n,m} \right] & = \frac{1}{m}  \operatorname{Var}_{\text{\tiny T}}\left[  \tau(X) \mathbbm{1}_{Z_{n}(x) > 0} \right] + \left( 1- \frac{1}{m} \right) \operatorname{Var} \left[ \mathbb{E} \left[ \tau(X)  \mathbbm{1}_{Z_n(X) = 0} | \mathbf{X}_n \right] \right] \\
    & \qquad +  \frac{1}{n}  \sum_{x\in\mathds{X}}  V_{\text{\tiny HT}} (x) \frac{p_{\text{\tiny T}}(x)(1-p_{\text{\tiny T}}(x))}{m}\mathbb{E} \left[  \frac{\mathbbm{1}_{Z_{n}(x) > 0} }{\hat p_{\text{\tiny R},n}\left( x \right)}  \right] + \operatorname{Var}\left[\hat \tau_{\pi, \text{\tiny T},n}^* \right].
\end{align*}

Let's consider a slightly different quantity, multiplying by $\min(n,m)$,

\begin{align*}
   \min(n,m)  \operatorname{Var}\left[   \hat \tau_{\pi, n,m} \right] & = \frac{\min(n,m)}{m}  \operatorname{Var}_{\text{\tiny T}}\left[  \tau(X) \mathbbm{1}_{Z_{n}(x) > 0} \right] + \min(n,m) \left( 1- \frac{1}{m} \right) \operatorname{Var} \left[ \mathbb{E} \left[ \tau(X)   \mathbbm{1}_{Z_n(X) = 0} | \mathbf{X}_n \right] \right] \\
    & \qquad +  \frac{\min(n,m) }{nm}  \sum_{x\in\mathds{X}}  V_{\text{\tiny HT}} (x) p_{\text{\tiny T}}(x)(1-p_{\text{\tiny T}}(x)) \mathbb{E} \left[  \frac{\mathbbm{1}_{Z_{n}(x) > 0} }{\hat p_{\text{\tiny R},n}\left( x \right)}  \right] + \min(n,m) \operatorname{Var}\left[\hat \tau_{\pi, \text{\tiny T},n}^* \right].
\end{align*}

Now, we study an asymptotic regime where $n$ and $m$ can grow toward infinity but at different paces. Let $  \lim\limits_{n,m\to\infty} \frac{m}{n} = \lambda \in [0,\infty]$,where $\lambda$ characterizes the regime.

\textbf{Case 1}:
If $\lambda \in [1, \infty]$, one can replace $\min(n,m) $ by $n$, so that

\begin{align*}
 \lim\limits_{n,m\to\infty}  n \operatorname{Var}\left[   \hat \tau_{\pi, n,m} \right] & = \lim\limits_{n,m\to\infty} \left( \underbrace{\frac{n}{m}}_{\frac{1}{\lambda}} \operatorname{Var}_{\text{\tiny T}}\left[  \tau(X) \mathbbm{1}_{Z_{n}(x) > 0} \right]+ n \left( 1 - \frac{1}{m} \right) \operatorname{Var} \left[ \mathbb{E} \left[ \tau(X)   \mathbbm{1}_{Z_n(X) = 0} | \mathbf{X}_n \right] \right] \right) \\
 &\qquad + \underbrace{\lim\limits_{n,m\to\infty} \left( \frac{1 }{m}  \sum_{x\in\mathds{X}}  V_{\text{\tiny HT}} (x) p_{\text{\tiny T}}(x)(1-p_{\text{\tiny T}}(x)) \mathbb{E} \left[  \frac{\mathbbm{1}_{Z_{n}(x) > 0}}{\hat p_{\text{\tiny R},n}\left( x \right)}  \right]\right)}_{=0} \\
 &\qquad + \underbrace{\lim\limits_{n,m\to\infty} \left( n \operatorname{Var}\left[\hat \tau_{\pi, \text{\tiny T},n}^* \right]\right)}_{ =V_{\text{so}}}, && \text{Corollary~\ref{cor_asympt_semi_oracle}}
\end{align*}

where we also used from former proof, \eqref{eq_proof_Chernoff1} and \eqref{eq_proof_Chernoff2} stating that
\begin{align*}
\mathbb{E} \left[ \frac{\mathbbm{1}_{Z_n(x)>0}}{  Z_n(x)/n } \right] \to \frac{1}{p_{\text{\tiny R}}(x)}, \quad \textrm{as}~n \to \infty. 
\end{align*}

Recalling \eqref{eq_proof_bound_var_exp_tau_ind},

\begin{equation*}
   0 \le   \operatorname{Var} \left[ \mathbb{E} \left[ \tau(X) \mathbbm{1}_{Z_n(X) = 0} | \mathbf{X}_n \right] \right]   \le \left(\mathds{E}_{\text{\tiny T}} \left[ |\tau(X)| \right]\right)^2  \left(1 - \min_x p_{\text{\tiny R}}(x) \right)^{n},
\end{equation*}

due to the exponential convergence one has,

\begin{align*}
     \lim\limits_{n \to\infty} n \operatorname{Var} \left[ \mathbb{E} \left[ \tau(X)  \mathbbm{1}_{Z_n(X) = 0} | \mathbf{X}_n \right] \right]   &= 0,
\end{align*}

and therefore,

\begin{align}\label{eq_convergence_exponential_term}
     \lim\limits_{n,m \to\infty} n \left( 1 - \frac{1}{m} \right)   \operatorname{Var} \left[ \mathbb{E} \left[ \tau(X) \mathbbm{1}_{Z_n(X) = 0} | \mathbf{X}_n \right] \right]   &= 0.
\end{align}

Besides, 
\begin{align*}
    \lim\limits_{n\to\infty} \operatorname{Var}_{\text{\tiny T}}\left[  \tau(X) \mathbbm{1}_{Z_n(X) \neq 0} \right] &= \operatorname{Var}_{\text{\tiny T}}\left[  \tau(X) \right],
\end{align*}
To summarize, if  $\lambda \in [1, \infty]$, one can conclude that 

\begin{align}\label{eq_lambda_above_1}
    \lim\limits_{n,m\to\infty}  n \operatorname{Var}\left[   \hat \tau_{\pi, n,m} \right] = \frac{\operatorname{Var}\left[  \tau(X) \right]}{\lambda} + V_{so}.
\end{align}
\\

\textbf{Case 2}: If $\lambda \in [0, 1]$, one can replace $\min(n,m) $ by $m$, so that 

\begin{align*}
   \lim\limits_{n,m\to\infty}  m  \operatorname{Var}\left[   \hat \tau_{\pi, n,m} \right] & =   \lim\limits_{n,m\to\infty}  \operatorname{Var}_{\text{\tiny T}}\left[  \tau(X) \mathbbm{1}_{Z_n(X) \neq 0} \right] +   \lim\limits_{n,m\to\infty}  m \left( 1 - \frac{1}{m} \right) \operatorname{Var} \left[ \mathbb{E} \left[ \tau(X)  \mathbbm{1}_{Z_n(X) = 0} | \mathbf{X}_n \right] \right] \\
    & \qquad +   \lim\limits_{n,m\to\infty}  \frac{1}{n}  \sum_{x\in\mathds{X}}  V_{\text{\tiny HT}} (x) p_{\text{\tiny T}}(x)(1-p_{\text{\tiny T}}(x)) \mathbb{E} \left[  \frac{\mathbbm{1}_{Z_n(x) \neq 0} }{\hat p_{\text{\tiny R},n}\left( x \right)}  \right] \\
    & \qquad \qquad +   \lim\limits_{n,m\to\infty} \lambda \sum_{x\in\mathds{X}}  V_{\text{\tiny HT}} (x) p_{\text{\tiny T}}^2 (x)  \mathbb{E} \left[  \frac{\mathbbm{1}_{Z_n(x) \neq 0} }{\hat p_{\text{\tiny R},n}\left( x \right)}  \right].
\end{align*}

In particular,

\begin{align*}
    \lim\limits_{n,m\to\infty}  \operatorname{Var}_{\text{\tiny T}}\left[  \tau(X) \mathbbm{1}_{Z_n(X) \neq 0} \right] &= \operatorname{Var}_{\text{\tiny T}}\left[  \tau(X) \right].
\end{align*}

As above, we have 
\begin{align*}
     \lim\limits_{n,m \to\infty} m \left( 1 - \frac{1}{m} \right)   \operatorname{Var} \left[ \mathbb{E} \left[ \tau(X)  \mathbbm{1}_{Z_n(X) = 0} | \mathbf{X}_n \right] \right]   &= 0,
\end{align*}

because,

\begin{align*}
      0 \le    m \left( 1 - \frac{1}{m} \right) \operatorname{Var} \left[ \mathbb{E} \left[ \tau(X) \mathbbm{1}_{Z_n(X) = 0} | \mathbf{X}_n \right] \right] \le n  \left( 1 - \frac{1}{m} \right) \operatorname{Var} \left[ \mathbb{E} \left[ \tau(X)  \mathbbm{1}_{Z_n(X) = 0} | \mathbf{X}_n \right] \right].
\end{align*}

In addition, \eqref{eq_proof_Chernoff1} and \eqref{eq_proof_Chernoff2} ensure that

\begin{align*}
    \lim\limits_{n,m\to\infty} \lambda \sum_{x\in\mathds{X}}  V_{\text{\tiny HT}} (x) p_{\text{\tiny T}}^2 (x)  \mathbb{E} \left[  \frac{\mathbbm{1}_{Z_n(x) \neq 0} }{\hat p_{\text{\tiny R},n}\left( x \right)}  \right] &= \lambda V_{so}.
\end{align*}

As an intermediary conclusion, if $\lambda \in [0, 1]$,

\begin{align}\label{eq_lambda_below_1}
     \lim\limits_{n,m\to\infty}  \min(n,m)  \operatorname{Var}\left[   \hat \tau_{\pi, n,m} \right] & = \operatorname{Var}\left[  \tau(X) \right] + \lambda V_{so}   
\end{align}
\\

\textbf{General conclusion}:
It is possible to gather equations \eqref{eq_lambda_above_1} and \eqref{eq_lambda_below_1} in one single conclusion. Therefore, letting $ \lim\limits_{n,m\to\infty} m/n = \lambda \in [0,\infty]$, the asymptotic variance of estimated IPSW satisfies
\begin{align*}
 \lim\limits_{n,m\to\infty} \min(n,m) \operatorname{Var}\left[   \hat \tau_{\pi, n,m} \right] = \min(1, \lambda) \left( \frac{\operatorname{Var}\left[ \tau(X) \right]}{\lambda} + V_{so} \right).
\end{align*}

\subsubsection{Proof of Theorem~\ref{thm:ipsw}}\label{proof_thm_ipsw}

 \begin{proof}
 For any estimate $\hat{\tau}$, we have 
\begin{equation*}
\mathbb{E}\left[ \left( \hat  \tau - \tau \right)^2\right] = \left( \mathbb{E}\left[ \hat  \tau \right] - \tau \right)^2 + \operatorname{Var}\left[   \hat  \tau \right].
\end{equation*}
Therefore, the risk of the (estimated) IPSW estimate can be bounded using results from Subsection~\ref{proof_prop_completely_estimated} (or Proposition~\ref{prop_completely_estimated}), and in particular the bounds on the variance and the bias,
\begin{align*}
\mathbb{E}\left[ \left( \hat  \tau - \tau \right)^2\right] &\leq \left(1 - \min_x p_{\text{\tiny R}}(x) \right)^{2n} \left(\mathbb{E}_{\text{\tiny T}}[|\tau(X)|]\right)^2 +   \frac{ \operatorname{Var}_{\text{\tiny T}} \left[ \tau(X)\right]}{m} +   \left(1 - \min_x p_R(x) \right)^{n/2} \left( \frac{4 \mathbb{E}_{\text{\tiny T}} \left[ \tau(X)^2 \right]}{m} + \tau \right) \nonumber \\
     & \qquad +  \frac{2}{n+1}\left(  \mathbb{E}_{\text{\tiny R}}\left[ \left( \frac{p_\text{\tiny T}\left( X \right)}{p_\text{\tiny R}\left( X \right)}\right)^2 V_{\text{\tiny HT}}(X)\right] + \frac{1}{m} \mathbb{E}_{\text{\tiny R}}\left[ \frac{p_\text{\tiny T}\left( X \right)(1-p_\text{\tiny T}\left( X \right))}{p_\text{\tiny R}\left( X \right)^2} V_{\text{\tiny HT}}(X)\right] \right)  \\
     & \le \frac{2V_{so}}{n+1} + \frac{\operatorname{Var}_{\text{\tiny T}}\left[\tau(X) \right]}{m} + \frac{2}{m(n+1)}  \mathbb{E}_{\text{\tiny R}}\left[ \frac{p_\text{\tiny T}\left( X \right)(1-p_\text{\tiny T}\left( X \right))}{p_\text{\tiny R}\left( X \right)^2} V_{\text{\tiny HT}}(X)\right] \\
& \quad  + \left(1 - \min_x p_{\text{\tiny R}}\left(x\right)\right)^{n/2}  \left(1 + \mathbb{E}_{\text{\tiny T}}[\tau(X)^2] + \frac{4 \mathbb{E}_{\text{\tiny T}} \left[ \tau(X)^2 \right]}{m} \right) .
\end{align*}



The $L^2$ consistency holds by letting $n$ \underline{and} $m$ tend to infinity. 

 \end{proof}

\subsection{Estimated IPSW with estimated $\hat{\pi}_n(x)$ }

\subsubsection{Proof of Proposition~\ref{prop:when-estimating-pi}}\label{proof:cor-when-estimating-pi}

We prove the more general following proposition.

\begin{proposition}[IPSW's properties when also estimating $\pi$]

Under the general setting defined in Subsection~\ref{subsec:model}, granting Assumptions~\ref{a:repres-rct}-\ref{a:pos}, the bias of the estimated IPSW with estimated $\hat \pi_n$ (see Definition~\ref{def:procedure-for-densities-and-pi}) is given by
\begin{align*}
    \mathbb{E} \left[ \hat \tau_{n, m} \right]- \tau \,&=\,\sum_{x\in\mathds{X}} p_{\text{\tiny T}}(x) \,\mathbb{E}\left[ Y^{(0)} \mid X=x \right]  \biggl( 1 - p_{\text{\tiny R}}(x) \left(1 - \pi(x)\right) \biggr)^n \\
    & \quad - \sum_{x\in\mathds{X}} p_{\text{\tiny T}}(x) \, \mathbb{E} \left[ Y^{(1)} \mid X = x \right]\bigl(  1- p_{\text{\tiny R}}(x)\,\pi(x)  \bigr)^n \\
    &  \leq \left(  1- \min_x \left( (1 - \tilde \pi(x)) p_{\text{\tiny R}}(x) \right) \right)^n \left(   \mathbb{E}_\text{\tiny T} \left[ |  Y^{(1)}   | \right]+  \mathbb{E}_\text{\tiny T} \left[| Y^{(0)}  | \right]  \right),
\end{align*}
where $\tilde \pi(x) = \max ( \pi(x), 1 - \pi(x))$.
Besides, the variance of the estimated IPSW with estimated $\hat \pi_n$ satisfies, for all $n$,
\begin{align*}
 \operatorname{Var}\left[    \hat \tau_{n, m} \right] &  =   \frac{1}{m}  \operatorname{Var}\left[   \tau(X)  - C_n(X) \right] + \left( 1 - \frac{1}{m} \right) \operatorname{Var} \left[ \mathbb{E} \left[ C_n(X) | \mathbf{X}_n \right] \right] \\
& \quad +    \sum_{x \in \mathbb{X}} \left( \frac{p_{\text{\tiny T}}(x)(1-p_{\text{\tiny T}}(x))}{m} + p_{\text{\tiny T}}^2 (x)  \right) \mathbb{E} \left[  \frac{\mathds{1}_{Z_n(x) > 0}}{Z_n(x)} \right] V_{\text{\tiny DM},n}(x), 
\end{align*}
\begin{flalign*}
\quad\qquad\text{where}&&
    C_n(X) = \mathbb{E} \left[ Y^{(1)} \mid X  \right] (1-\pi(X))^{Z_n(X)}  - \mathbb{E} \left[ Y^{(0)} \mid X \right]\pi(X)^{Z_n(X)}.&&
\end{flalign*}
\begin{flalign*}
\text{Furthermore,} &&
\operatorname{Var}\left[    \hat \tau_{n, m} \right]  \leq &\frac{2\,  \tilde V_{so}}{n+1}   + \frac{\operatorname{Var}\left[   \tau(X) \right] }{m}   +  \frac{2}{(n+1)m}  \mathbb{E}_{\text{\tiny R}}\left[ \frac{p_\text{\tiny T}\left( X \right)(1-p_\text{\tiny T}\left( X \right))}{p_\text{\tiny R}\left( X \right)^2} V_{\text{\tiny DM}}(X)\right] 
&&\\&&
&\quad + 2 \left( 1 + \frac{3}{m} \right) \left( 1 - \min_x \left( (1 - \tilde \pi(x)^2) p_{\text{\tiny R}}(x) \right) \right)^{n/2} \mathbb{E} \left[ (Y^{(1)})^2 + (Y^{(0)})^2 \right],&&
\\
\qquad\text{where} &&
\tilde \pi(x) = \max &\left(\pi(x), 1 - \pi(x) \right) \quad \text{and} \quad
    \tilde V_{so}:=\mathbb{E}_{\text{\tiny R}}\left[ \left( \frac{p_\text{\tiny T}\left( X \right)}{p_\text{\tiny R}\left( X \right)}\right)^2 V_{\text{\tiny DM},n}(X)\right].&&
\end{flalign*}
\end{proposition}

\begin{proof}

\textbf{Bias}\\

We start by computing the bias of 
\begin{align*}
 \mathbb{E} \left[ \hat \tau_{n, m} \right] & = \mathbb{E} \left[ \frac{1}{n} \sum_{i=1}^{n} \frac{\hat p_{\text{\tiny T}, m}(X_i)}{\hat p_{\text{\tiny R}, n}(X_i)} Y_i \left( \frac{A_i}{\hat \pi_n(X_i)} - \frac{1-A_i}{1-\hat \pi_n(X_i)} \right) \right] \\
 &= \frac{1}{n} \mathbb{E} \left[ \mathbb{E} \left[  \sum_{x \in \mathbb{X}}  \frac{ \hat p_{\text{\tiny T},m}(x) \mathbbm{1}_{Z_{n}(x) > 0}}{\hat p_{\text{\tiny R}, n}(x)}  \sum_{i=1}^n \mathds{1}_{X_i = x}  Y_i \left( \frac{A_i}{\hat \pi_n(x)} - \frac{1-A_i}{1-\hat \pi_n(x)} \right) \mid \mathbf{X}_n, \mathbf{A}_n, \mathbf{Y}_n\right]  \right]  \\
  &= \frac{1}{n} \mathbb{E} \left[  \sum_{x \in \mathbb{X}} \mathbbm{1}_{Z_{n}(x) > 0} \frac{  \mathbb{E} \left[ \hat p_{\text{\tiny T},m}(x) \mid \mathbf{X}_n, \mathbf{A}_n, \mathbf{Y}_n  \right] }{\hat p_{\text{\tiny R}, n}(x)}  \sum_{i=1}^n \mathds{1}_{X_i = x}  Y_i \left( \frac{A_i}{\hat \pi_n(x)} - \frac{1-A_i}{1-\hat \pi_n(x)} \right)  \right]  \\
 & = \frac{1}{n} \mathbb{E} \left[ \sum_{x \in \mathbb{X}}  \frac{ p_{\text{\tiny T}}(x) \mathbbm{1}_{Z_{n}(x) > 0}}{\hat p_{\text{\tiny R}, n}(x)}  \sum_{i=1}^n \mathds{1}_{X_i = x}  Y_i \left( \frac{A_i}{\hat \pi_n(x)} - \frac{1-A_i}{1-\hat \pi_n(x)} \right) \right].
\end{align*}

This derivation is possible as $\hat p_{\text{\tiny T}, m}$ is estimated on a different data set than the trial.

Using SUTVA (Assumption~\ref{a:trial-internal-validity}), one can replace observed outcomes by potential outcomes, and
\begin{align*}
 \mathbb{E} \left[ \hat \tau_{n, m} \right] & =   \frac{1}{n} \mathbb{E} \left[ \sum_{x \in \mathbb{X}}  \frac{ p_{\text{\tiny T}}(x)\mathbbm{1}_{Z_{n}(x) > 0}}{\hat p_{\text{\tiny R}, n}(x)}  \sum_{i=1}^n \mathds{1}_{X_i = x}  \left( \frac{ Y_i^{(1)} A_i}{\hat \pi_n(x)} - \frac{Y_i^{(0)}(1-A_i)}{1-\hat \pi_n(x)} \right) \right]\\
 & = \frac{1}{n} \mathbb{E} \left[ \sum_{x \in \mathbb{X}}  \frac{ p_{\text{\tiny T}}(x)\mathbbm{1}_{Z_{n}(x) > 0}}{\hat p_{\text{\tiny R}, n}(x)}  \sum_{i=1}^n \mathds{1}_{X_i = x}   \mathbb{E} \left[  \left( \frac{Y_i^{(1)} A_i}{\hat \pi_n(x)} - \frac{Y_i^{(0)} (1-A_i)}{1-\hat \pi_n(x)}\right) | \mathbf{X}_n, \mathbf{Y}_n^{(1)}, \mathbf{Y}_n^{(0)} \right]  \right].
 \end{align*}
Let us consider, for any fixed $x \in \mathbb{X}$, 
\begin{align*}
 \mathbb{E} \left[  \frac{Y_i^{(1)} A_i}{\hat \pi_n(x)}  | \mathbf{X}_n, \mathbf{Y}_n^{(1)}, \mathbf{Y}_n^{(0)}  \right] 
 =  Y_i^{(1)} \mathbb{E} \left[  \frac{A_i}{\hat \pi_n(x)}  | \mathbf{X}_n \right].
\end{align*}
Up to reordering the $X_i$'s, we have
\begin{align*}
\mathbb{E} \left[  \frac{A_i}{\hat \pi_n(x)}  | \mathbf{X}_n \right]  & =  \mathbb{E} \left[  \frac{A_i}{ \frac{\sum_{j=1}^{Z_n(x)} A_j}{Z_n(x)}  }  | \mathbf{X}_n \right]\\
 & = Z_n(x) \mathbb{E} \left[  \frac{A_i}{\sum_{j=1}^{Z_n(x)} A_j}  | \mathbf{X}_n \right]\\
& = Z_n(x) \pi(x) \mathbb{E} \left[  \frac{1}{1 + \sum_{j=2}^{Z_n(x)} A_j}  | \mathbf{X}_n \right].
\end{align*}
The last rows uses the law of total probability.
According to Lemma 11 $(i)$ in \cite{Biau2012AnalysisRandomForests}, and considering $B_n(x) \sim \mathfrak{B}(n, p)$, for any $x \in \mathds{X}$, 
\begin{align*}
\mathbb{E} \left[  \frac{1}{1 + B_n(x)} \right] = \frac{1}{(n+1)p} - \frac{(1-p)^{n+1}}{(n+1)p}.
\end{align*}
Since, conditional on $\mathbf{X}_n$, $\sum_{j=2}^{Z_n(x)} A_j$ is distributed as $ \mathfrak{B}(Z_n(x)-1, \pi(x))$, 
\begin{align*}
\mathbb{E} \left[  \frac{A_i}{\hat \pi_n(x)}  | \mathbf{X}_n \right] & = Z_n(x) \pi(x)  \left( \frac{1}{Z_n(x)\pi(x)} - \frac{(1-\pi(x))^{Z_n(x)}}{Z_n(x) \pi(x)} \right) \\
& =  1 - (1-\pi(x))^{Z_n(x)}.
\end{align*}
Similarly, 
\begin{align*}
 \mathbb{E} \left[  \frac{(1-A_i)}{1 - \hat \pi_n(x)}  | \mathbf{X}_n \right] = 1 - \pi(x)^{Z_n(x)}. 
\end{align*}
Consequently, 
\begin{align*}
\mathbb{E} \left[  \left( \frac{Y_i^{(1)} A_i}{\hat \pi_n(x)} - \frac{Y_i^{(0)} (1-A_i)}{1-\hat \pi_n(x)}\right) | \mathbf{X}_n, \mathbf{Y}_n^{(1)}, \mathbf{Y}_n^{(0)} \right] & =   Y_i^{(1)}\left(   1 - (1-\pi(x))^{Z_n(x)}\right) -  Y_i^{(0)}\left( 1 - \pi(x)^{Z_n(x)}\right) \\
&=  \left( Y_i^{(1)} - Y_i^{(0)} \right) - Y_i^{(1)}  (1-\pi(x))^{Z_n(x)} +  Y_i^{(0)}\pi(x)^{Z_n(x)}.
\end{align*}

Therefore,

\begin{align}
     \mathbb{E} \left[ \hat \tau_{n, m} \right] & =    \frac{1}{n} \mathbb{E} \left[ \sum_{x \in \mathbb{X}}  \frac{ p_{\text{\tiny T}}(x)\mathbbm{1}_{Z_{n}(x) > 0}}{\hat p_{\text{\tiny R}, n}(x)}  \sum_{i=1}^n \mathds{1}_{X_i = x}  ( Y_i^{(1)} - Y_i^{(0)})  \right] \label{eq_proof_completely_first_term} \\
     &\qquad  +  \frac{1}{n} \mathbb{E} \left[ \sum_{x \in \mathbb{X}}  \frac{ p_{\text{\tiny T}}(x)\mathbbm{1}_{Z_{n}(x) > 0}}{\hat p_{\text{\tiny R}, n}(x)}  \sum_{i=1}^n \mathds{1}_{X_i = x}  Y_i^{(0)}\pi(x)^{Z_n(x)} \right] \label{eq_proof_completely_second_term} \\
     &\qquad \qquad  -  \frac{1}{n} \mathbb{E} \left[ \sum_{x \in \mathbb{X}}  \frac{ p_{\text{\tiny T}}(x)\mathbbm{1}_{Z_{n}(x) > 0}}{\hat p_{\text{\tiny R}, n}(x)}  \sum_{i=1}^n \mathds{1}_{X_i = x}   Y_i^{(1)}  (1-\pi(x))^{Z_n(x)}  \right] .\label{eq_proof_completely_third_term}
\end{align}

On one hand, considering \eqref{eq_proof_completely_first_term},

\begin{align*}
      \frac{1}{n} \mathbb{E} \left[ \sum_{x \in \mathbb{X}}  \frac{ p_{\text{\tiny T}}(x)\mathbbm{1}_{Z_{n}(x) > 0}}{\hat p_{\text{\tiny R}, n}(x)}  \sum_{i=1}^n \mathds{1}_{X_i = x}  ( Y_i^{(1)} - Y_i^{(0)})  \right]&=   \frac{1}{n}\sum_{x \in \mathbb{X}}   \mathbb{E} \left[ \frac{ p_{\text{\tiny T}}(x)\mathbbm{1}_{Z_{n}(x) > 0}}{\hat p_{\text{\tiny R}, n}(x)}  \sum_{i=1}^n \mathds{1}_{X_i = x}  ( Y_i^{(1)} - Y_i^{(0)})  \right] \\
      &=  \sum_{x \in \mathbb{X}}   \mathbb{E} \left[ p_{\text{\tiny T}}(x)\mathds{1}_{Z_n(x) > 0} \tau(x)  \right],
\end{align*}

corresponding to the bias in the semi-oracle and the estimated IPSW. Indeed, we recall from the semi-oracle IPSW proof that,

\begin{align*}
     \sum_{x \in \mathbb{X}}   \mathbb{E} \left[ p_{\text{\tiny T}}(x)\mathds{1}_{Z_n(x) > 0} \tau(x)  \right]&=   \sum_{x \in \mathbb{X}}  p_{\text{\tiny T}}(x)(1-(1- p_{\text{\tiny R}}(x))^n)\tau(x).
\end{align*}

On the other hand, considering \eqref{eq_proof_completely_second_term},

\begin{align*}
    \frac{1}{n} \mathbb{E} \left[ \sum_{x \in \mathbb{X}}  \frac{ p_{\text{\tiny T}}(x)\mathbbm{1}_{Z_{n}(x) > 0}}{\hat p_{\text{\tiny R}, n}(x)}  \sum_{i=1}^n \mathds{1}_{X_i = x}  Y_i^{(0)}\pi(x)^{Z_n(x)} \right]  &=  \frac{1}{n} \mathbb{E} \left[ \sum_{x \in \mathbb{X}}  \frac{ p_{\text{\tiny T}}(x)\mathbbm{1}_{Z_{n}(x) > 0}}{\hat p_{\text{\tiny R}, n}(x)}  \sum_{i=1}^n \mathds{1}_{X_i = x}  \mathbb{E} \left[ Y_i^{(0)} \mid X_i = x\right] \pi(x)^{Z_n(x)} \right] \\
    &=\mathbb{E} \left[ \sum_{x \in \mathbb{X}}p_{\text{\tiny T}}(x) \mathds{1}_{Z_n(x) > 0}\mathbb{E} \left[ Y_i^{(0)} \mid X_i = x\right] \pi(x)^{Z_n(x)} \right]\\
    &=\sum_{x \in \mathbb{X}}p_{\text{\tiny T}}(x) \mathbb{E} \left[ Y_i^{(0)} \mid X_i = x\right]   \mathbb{E} \left[  \mathds{1}_{Z_n(x) > 0}\pi(x)^{Z_n(x)} \right]\\ 
    &=\sum_{x \in \mathbb{X}}p_{\text{\tiny T}}(x) \mathbb{E} \left[ Y_i^{(0)} \mid X_i = x\right]   \mathbb{E} \left[  \left( 1 - \mathds{1}_{Z_n(x) = 0} \right)\pi(x)^{Z_n(x)} \right]\\
    &=\sum_{x \in \mathbb{X}}p_{\text{\tiny T}}(x) \mathbb{E} \left[ Y_i^{(0)} \mid X_i = x\right] \left(   \mathbb{E} \left[ \pi(x)^{Z_n(x)} \right]  - \mathbb{E} \left[ \mathds{1}_{Z_n(x) = 0}\right]\right).\\
\end{align*}

Now, note that $\mathbb{P} \left[ Z_n(x)=0 \right] = (1 - p_{\text{\tiny R}}(x))^n$ and
\begin{align*}
  \mathbb{E} \left[  \pi(x)^{Z_n(x)}   \right] & = \prod_{j=1}^{n}   \mathbb{E} \left[ \pi(x)^{\mathds{1}_{X_i =x}}   \right]\\
  & = \left(\pi(x) p_{\text{\tiny R}}(x) + (1 - p_{\text{\tiny R}}(x)) \right)^n.\\
 & = \left( 1 - p_{\text{\tiny R}}(x) \left(1 - \pi(x)\right) \right)^n.
\end{align*}

Therefore,

\begin{align*}
    \frac{1}{n} \mathbb{E} \left[ \sum_{x \in \mathbb{X}}  \frac{ p_{\text{\tiny T}}(x) \mathbbm{1}_{Z_{n}(x) > 0}}{\hat p_{\text{\tiny R}, n}(x)}  \sum_{i=1}^n \mathds{1}_{X_i = x}  Y_i^{(0)}\pi(x)^{Z_n(x)} \right] 
    &=\sum_{x \in \mathbb{X}}p_{\text{\tiny T}}(x)    \mathbb{E} \left[ Y_i^{(0)} \mid X_i = x \right] \big(  \left( 1 - p_{\text{\tiny R}}(x) \left(1 - \pi(x)\right) \right)^n - \left(1 - p_{\text{\tiny R}}(x)\right)^n \big).\\
\end{align*}

Similarly, considering \eqref{eq_proof_completely_third_term},
\begin{align*}
     -  \frac{1}{n} \mathbb{E} \left[ \sum_{x \in \mathbb{X}}  \frac{ p_{\text{\tiny T}}(x) \mathbbm{1}_{Z_{n}(x) > 0}}{\hat p_{\text{\tiny R}, n}(x)}  \sum_{i=1}^n \mathds{1}_{X_i = x}   Y_i^{(1)}  (1-\pi)^{Z_n(x)}  \right] &= - \sum_{x \in \mathbb{X}}p_{\text{\tiny T}}(x) \mathbb{E} \left[ Y_i^{(1)} \mid X_i = x\right] \big( \left(  1- p_{\text{\tiny R}}(x)\pi(x)  \right)^n - (1 - p_{\text{\tiny R}}(x))^n\big).
\end{align*}

Finally, the bias of the estimated IPSW with estimated treatment proportion is given by 

\begin{align*}
\mathbb{E} \left[ \hat \tau_{n, m} \right]- \tau &=  - \sum_{x\in\mathds{X}} p_{\text{\tiny T}}(x)\tau(x)  \left(1 - p_{\text{\tiny R}}(x)\right)^n  \\ 
& \qquad +  \sum_{x\in\mathds{X}} p_{\text{\tiny T}}(x) \left(\mathbb{E}\left[ Y_i^{(1)} \mid X_i = x\right]  - \mathbb{E}\left[ Y_i^{(0)} \mid X_i = x\right]  \right) (1 - p_{\text{\tiny R}}(x))^n \\
& \qquad \qquad +  \sum_{x\in\mathds{X}} p_{\text{\tiny T}}(x) \mathbb{E}\left[ Y_i^{(0)} \mid X_i = x\right]  \left( 1 - p_{\text{\tiny R}}(x) \left(1 - \pi(x)\right) \right)^n \\
& \qquad \qquad \qquad - \sum_{x\in\mathds{X}} p_{\text{\tiny T}}(x) \mathbb{E}\left[ Y_i^{(1)} \mid X_i = x \right]\left(  1- p_{\text{\tiny R}}(x)\pi(x)  \right)^n ,
\end{align*}

such that,

\begin{align*}
    \mathbb{E} \left[ \hat \tau_{n, m} \right]- \tau &=\sum_{x\in\mathds{X}} p_{\text{\tiny T}}(x) \mathbb{E}\left[ Y_i^{(0)} \mid X_i = x\right]  \left( 1 - p_{\text{\tiny R}}(x) \left(1 - \pi(x)\right) \right)^n \\
& \quad  - \sum_{x\in\mathds{X}} p_{\text{\tiny T}}(x) \mathbb{E}\left[ Y_i^{(1)} \mid X_i = x \right]\left(  1- p_{\text{\tiny R}}(x)\pi(x)  \right)^n \,.
\end{align*}

Consequently, the bias of the IPSW estimator with estimated $\hat \pi_n$ can be upper bounded via
\begin{align*}
    \left| \mathbb{E} \left[ \hat \tau_{n, m} \right]- \tau \right| & \leq  \sum_{x\in\mathds{X}} p_{\text{\tiny T}}(x) \left|\mathbb{E}\left[ Y^{(0)} \mid X = x\right] \right|  \left( 1 - p_{\text{\tiny R}}(x) \left(1 - \pi(x)\right) \right)^n \\
& \quad + \sum_{x\in\mathds{X}} p_{\text{\tiny T}}(x) \left|\mathbb{E}\left[ Y^{(1)} \mid X = x \right] \right|\left(  1- p_{\text{\tiny R}}(x)\pi(x)  \right)^n \\
    & \leq \left(  1- \min_x \left( (1 - \tilde \pi(x)) p_{\text{\tiny R}}(x) \right) \right)^n  \mathbb{E}_\text{\tiny T} \left[ \left|\mathbb{E}\left[ Y^{(1)} \mid X \right] \right| + \left|\mathbb{E}\left[ Y^{(0)} \mid X \right] \right| \right],
\end{align*}
where $\tilde \pi(x) = \max ( \pi(x), 1 - \pi(x))$.

\textbf{Variance}\\

As above, we have
\begin{equation}
    \operatorname{Var}\left[    \hat \tau_{n, m} \right] =  \operatorname{Var}\left[ \mathbb{E} \left[ \hat \tau_{n, m} \mid \mathbf{X}_{m+n} \right] \right] + \mathbb{E} \left[ \operatorname{Var}\left[  \hat \tau_{n, m}  \mid \mathbf{X}_{m+n}  \right] \right]. \label{eq_proof_comp_est_probest_dec_var}
\end{equation}
Let us examine the first term. We have
\begin{align*}
\mathbb{E} \left[ \hat \tau_{n, m} \mid \mathbf{X}_{m+n} \right]  & = \mathbb{E} \left[ \frac{1}{n} \sum_{i=1}^{n} \frac{\hat p_{\text{\tiny T}, m}(X_i)}{\hat p_{\text{\tiny R}, n}(X_i)} Y_i \left( \frac{A_i}{\hat \pi_n(X_i)} - \frac{1-A_i}{1-\hat \pi_n(X_i)} \right) \mid \mathbf{X}_{m+n} \right]\\
& = \frac{1}{n} \sum_{i=1}^{n} \frac{\hat p_{\text{\tiny T}, m}(X_i)}{\hat p_{\text{\tiny R}, n}(X_i)} \mathbb{E} \left[    \frac{Y_i^{(1)} A_i}{\hat \pi_n(X_i)} - \frac{Y_i^{(0)} (1-A_i)}{1-\hat \pi_n(X_i)}  \mid \mathbf{X}_{m+n} \right]\\
& = \frac{1}{n} \sum_{i=1}^{n} \frac{\hat p_{\text{\tiny T}, m}(X_i)}{\hat p_{\text{\tiny R}, n}(X_i)} \mathbb{E} \left[   \mathbb{E} \left[ \frac{Y_i^{(1)} A_i}{\hat \pi_n(X_i)} - \frac{Y_i^{(0)} (1-A_i)}{1-\hat \pi_n(X_i)} \mid \mathbf{X}_{m+n}, \mathbf{Y}_{n}^{(1)}, \mathbf{Y}_{n}^{(0)} \right] \mid \mathbf{X}_{m+n} \right].
\end{align*}
A similar computation as the one used in the derivation of the bias above shows that
\begin{align*}
\mathbb{E} \left[ \frac{Y_i^{(1)} A_i}{\hat \pi_n(X_i)} - \frac{Y_i^{(0)} (1-A_i)}{1-\hat \pi_n(X_i)} \mid \mathbf{X}_{m+n}, \mathbf{Y}_{n}^{(1)}, \mathbf{Y}_{n}^{(0)} \right] =   \left( Y_i^{(1)} - Y_i^{(0)} \right) - Y_i^{(1)}  (1-\pi(X_i))^{Z_n(X_i)} +  Y_i^{(0)}\pi(X_i)^{Z_n(X_i)}, 
\end{align*}
which leads to 
\begin{align*}
\mathbb{E} \left[ \hat \tau_{n, m} \mid \mathbf{X}_{m+n} \right] &= \frac{1}{n} \sum_{i=1}^{n} \frac{\hat p_{\text{\tiny T}, m}(X_i)}{\hat p_{\text{\tiny R}, n}(X_i)} \mathbb{E} \left[  \left( Y_i^{(1)} - Y_i^{(0)} \right) - Y_i^{(1)}  (1-\pi(X_i))^{Z_n(X_i)} +  Y_i^{(0)}\pi(X_i)^{Z_n(X_i)}  \mid \mathbf{X}_{m+n} \right]\\
& = \frac{1}{n} \sum_{i=1}^{n} \frac{\hat p_{\text{\tiny T}, m}(X_i)}{\hat p_{\text{\tiny R}, n}(X_i)} \left( \tau(X_i) - \mathbb{E} \left[ Y_i^{(1)} \mid X_i \right] (1-\pi(X_i))^{Z_n(X_i)}  + \mathbb{E} \left[ Y_i^{(0)} \mid X_i \right]\pi(X_i)^{Z_n(X_i)} \right).
\end{align*}
Rewriting the previous sum yields
 \begin{align*}
\mathbb{E} \left[ \hat \tau_{n, m} \mid \mathbf{X}_{m+n} \right] 
& = \frac{1}{n} \sum_{x \in \mathbb{X}} \mathbbm{1}_{Z_{n}(x) > 0} \sum_{i=1}^n \mathds{1}_{X_i=x} \frac{\hat p_{\text{\tiny T}, m}(X_i)}{\hat p_{\text{\tiny R}, n}(X_i)}  \bigg( \tau(X_i) - \mathbb{E} \left[ Y_i^{(1)} \mid X_i \right] (1-\pi(X_i))^{Z_n(X_i)} \\
& \qquad + \mathbb{E} \left[ Y_i^{(0)} \mid X_i \right]\pi(X_i)^{Z_n(X_i)} \bigg)\\
& =  \sum_{x \in \mathbb{X}} \mathbbm{1}_{Z_{n}(x) > 0} \hat p_{\text{\tiny T}, m}(x) \left( \tau(x) - \mathbb{E} \left[ Y_i^{(1)} \mid X_i = x \right] (1-\pi(x))^{Z_n(x)}  + \mathbb{E} \left[ Y_i^{(0)} \mid X_i = x \right]\pi(x)^{Z_n(x)} \right)\\
& = \frac{1}{m}   \sum_{i=n+1}^{n+m}     U_n(X_i), 
\end{align*}
where 
\begin{align*}
U_n(X_i) := \left( \tau(X_i) - \mathbb{E} \left[ Y_i^{(1)} \mid X_i  \right] (1-\pi(X_i))^{Z_n(X_i)}  + \mathbb{E} \left[ Y_i^{(0)} \mid X_i \right]\pi(X_i)^{Z_n(X_i)} \right).
\end{align*}
By the law of total variance, 
\begin{align*}
& \operatorname{Var}\left[   \mathbb{E} \left[ \hat \tau_{n, m} \mid \mathbf{X}_{m+n} \right] \right] \\
 = & \operatorname{Var}\left[  \frac{1}{m}   \sum_{i=n+1}^{n+m}  U_n(X_i)  \right] \\
 = & \mathbb{E} \left[ \operatorname{Var}\left[  \frac{1}{m}   \sum_{i=n+1}^{n+m}  U_n(X_i) | \mathbf{X}_n \right] \right]   + \operatorname{Var} \left[ \mathbb{E} \left[  \frac{1}{m}   \sum_{i=n+1}^{n+m}  U_n(X_i) | \mathbf{X}_n \right] \right] \\
 = & \frac{1}{m} \mathbb{E} \left[ \operatorname{Var}\left[    U_n(X) | \mathbf{X}_n \right] \right] + \operatorname{Var} \left[ \mathbb{E} \left[    U_n(X) | \mathbf{X}_n \right] \right]\\
 = & \frac{1}{m}  \operatorname{Var}\left[    U_n(X) \right] + \left( 1 - \frac{1}{m} \right) \operatorname{Var} \left[ \mathbb{E} \left[    U_n(X) | \mathbf{X}_n \right] \right],
\end{align*}
where the last line comes from the law of total variance applied to $\operatorname{Var}\left[   U_n(X) \right]$. Since
\begin{align*}
& \operatorname{Var} \left[ \mathbb{E} \left[    U_n(X) | \mathbf{X}_n \right] \right] \\
= & \operatorname{Var} \left[ \mathbb{E} \left[  \left( \tau(X) - \mathbb{E} \left[ Y^{(1)} \mid X  \right] (1-\pi(X))^{Z_n(X)}  + \mathbb{E} \left[ Y^{(0)} \mid X \right]\pi(X)^{Z_n(X)} \right) | \mathbf{X}_n \right] \right] \\
= & \operatorname{Var} \left[ \mathbb{E} \left[  \mathbb{E} \left[ Y^{(0)} \mid X \right]\pi(X)^{Z_n(X)} - \mathbb{E} \left[ Y^{(1)} \mid X  \right] (1-\pi(X))^{Z_n(X)}  | \mathbf{X}_n \right] \right],
\end{align*}
as the only source of randomness comes from $ Z_n(X)$ (and not from $\tau(X)$),
we have
\begin{align}
\label{proof_first_var_pi_est}
    \operatorname{Var}\left[ \mathbb{E} \left[\hat \tau_{n,m} \mid \mathbf{X}_{m+n} \right]  \right] &=\frac{1}{m}  \operatorname{Var}\left[   \tau(X)  - C_n(X) \right] + \left( 1 - \frac{1}{m} \right) \operatorname{Var} \left[ \mathbb{E} \left[ C_n(X) | \mathbf{X}_n \right] \right], 
\end{align}
where 
\begin{align*}
    C_n(X) = \mathbb{E} \left[ Y^{(1)} \mid X  \right] (1-\pi(X))^{Z_n(X)}  - \mathbb{E} \left[ Y^{(0)} \mid X \right]\pi(X)^{Z_n(X)}.
\end{align*}

Regarding the other term, and first re-writing $ \hat \tau_{n, m}$,
\begin{align*}
 \hat \tau_{n, m} & =  \frac{1}{n} \sum_{i=1}^{n} \frac{\hat p_{\text{\tiny T}, m}(X_i)}{\hat p_{\text{\tiny R}, n}(X_i)}   \left( \frac{A_i Y_i^{(1)}}{\hat \pi_n(X_i)} - \frac{(1-A_i)Y_i^{(0)}}{1-\hat \pi_n(X_i)} \right)\\
 & = \sum_{i=1}^{n} \frac{\hat p_{\text{\tiny T}, m}(X_i)}{Z_n(X_i)}   \left( \frac{A_i Y_i^{(1)}}{\hat \pi_n(X_i)} - \frac{(1-A_i)Y_i^{(0)}}{1-\hat \pi_n(X_i)} \right)\\
 & = \sum_{x \in \mathbb{X}} \frac{\hat p_{\text{\tiny T}, m}(x) \mathbbm{1}_{Z_{n}(x) > 0}}{Z_n(x)}  \sum_{i=1}^{n} \mathds{1}_{X_i=x}   \left( \frac{A_i Y_i^{(1)}}{\hat \pi_n(x)} - \frac{(1-A_i)Y_i^{(0)}}{1-\hat \pi_n(x)} \right).
\end{align*}
Hence, 
\begin{align*}
 & \operatorname{Var} \left[ \hat \tau_{n, m} \mid \mathbf{X}_{n+m}\right] \\
 = &  \operatorname{Var} \left[ \sum_{x \in \mathbb{X}} \frac{\hat p_{\text{\tiny T}, m}(x) \mathbbm{1}_{Z_{n}(x) > 0}}{Z_n(x)}  \sum_{i=1}^{n} \mathds{1}_{X_i=x}   \left( \frac{A_i Y_i^{(1)}}{\hat \pi_n(x)} - \frac{(1-A_i)Y_i^{(0)}}{1-\hat \pi_n(x)} \right) \mid \mathbf{X}_{n+m}\right] \\
 = & \sum_{x \in \mathbb{X}} (\hat p_{\text{\tiny T}, m}(x))^2 \operatorname{Var} \left[  \frac{\mathbbm{1}_{Z_{n}(x) > 0}}{Z_n(x)}  \sum_{i=1}^{n} \mathds{1}_{X_i=x}   \left( \frac{A_i Y_i^{(1)}}{\hat \pi_n(x)} - \frac{(1-A_i)Y_i^{(0)}}{1-\hat \pi_n(x)} \right) \mid \mathbf{X}_{n+m}\right] \\
 & + \sum_{x, y \in \mathbb{X}, x \neq y} \operatorname{Cov} \bigg[  \frac{\hat p_{\text{\tiny T}, m}(x) \mathbbm{1}_{Z_{n}(x) > 0}}{Z_n(x)}  \sum_{i=1}^{n} \mathds{1}_{X_i=x}   \left( \frac{A_i Y_i^{(1)}}{\hat \pi_n(x)} - \frac{(1-A_i)Y_i^{(0)}}{1-\hat \pi_n(x)} \right), \\
 & \qquad \qquad \qquad \frac{\hat p_{\text{\tiny T}, m}(y) \mathbbm{1}_{Z_{n}(y) > 0}}{Z_n(y)}  \sum_{j=1}^{n} \mathds{1}_{X_j=y}   \left( \frac{A_j Y_j^{(1)}}{\hat \pi_n(y)} - \frac{(1-A_j)Y_j^{(0)}}{1-\hat \pi_n(y)} \right)  \mid \mathbf{X}_{n+m}\bigg].
\end{align*}
Note that the term 
\begin{align*}
& \operatorname{Var} \left[  \frac{\mathbbm{1}_{Z_{n}(x) > 0}}{Z_n(x)}  \sum_{i=1}^{n} \mathds{1}_{X_i=x}   \left( \frac{A_i Y_i^{(1)}}{\hat \pi_n(x)} - \frac{(1-A_i)Y_i^{(0)}}{1-\hat \pi_n(x)} \right) \mid \mathbf{X}_{n+m}\right],
\end{align*}
corresponds to the variance of the difference-in-means estimator on the strata $X=x$ (where $n$ is replaced by $Z_n(x)$) and therefore equals 
\begin{align*}
V_{\text{\tiny DM},n}(x)\mathds{1}_{Z_n(x) > 0}/Z_n(x), 
\end{align*}
where 
\begin{align*}
V_{\text{\tiny DM},n}(x)
 =  \frac{1}{Z_n(x)} \operatorname{Var} \left[    \sum_{i=1}^{n} \mathds{1}_{X_i=x}   \left( \frac{A_i Y_i^{(1)}}{\hat \pi_n(x)} - \frac{(1-A_i)Y_i^{(0)}}{1-\hat \pi_n(x)} \right) \mid \mathbf{X}_{n+m} \right].
\end{align*}

Consequently, 
\begin{align*}
 & \operatorname{Var} \left[ \hat \tau_{n, m} \mid \mathbf{X}_{n+m}\right] \\
 = & \sum_{x \in \mathbb{X}} \frac{(\hat p_{\text{\tiny T}, m}(x))^2 V_{\text{\tiny DM},n}(x) \mathds{1}_{Z_n(x) > 0}}{Z_n(x)} \\
 & + \sum_{x, y \in \mathbb{X}, x \neq y} \frac{\hat p_{\text{\tiny T}, m}(x)}{Z_n(x)} \frac{\hat p_{\text{\tiny T}, m}(y)}{Z_n(y)} \sum_{i, j} \mathds{1}_{X_i=x} \mathds{1}_{X_j=y} \operatorname{Cov} \left[        \left( \frac{A_i Y_i^{(1)}}{\hat \pi_n(x)} - \frac{(1-A_i)Y_i^{(0)}}{1-\hat \pi_n(x)} \right),      \left( \frac{A_j Y_j^{(1)}}{\hat \pi_n(y)} - \frac{(1-A_j)Y_j^{(0)}}{1-\hat \pi_n(y)} \right)  \mid \mathbf{X}_{n+m}\right].
\end{align*}
Note that for $x\neq y $, $\hat \pi_n(x) \indep \hat \pi_n(y)$. Consequently, for $i \neq j$,
\begin{align*}
 \left( \frac{A_i Y_i^{(1)}}{\hat \pi_n(x)} - \frac{(1-A_i)Y_i^{(0)}}{1-\hat \pi_n(x)} \right) \indep  \left( \frac{A_j Y_j^{(1)}}{\hat \pi_n(x)} - \frac{(1-A_j)Y_j^{(0)}}{1-\hat \pi_n(x)} \right).
\end{align*}
Consequently, 
\begin{align*}
 & \operatorname{Var} \left[ \hat \tau_{n, m} \mid \mathbf{X}_{n+m}\right] = \sum_{x \in \mathbb{X}} \frac{(\hat p_{\text{\tiny T}, m}(x))^2 V_{\text{\tiny DM},n}(x) \mathds{1}_{Z_n(x) > 0}}{Z_n(x)},
 \end{align*}
and, taking the expectation with respect to $\mathbf{X}_{n+m}$, we have
\begin{align}
\mathbb{E} \left[ \operatorname{Var} \left[ \hat \tau_{n, m} \mid \mathbf{X}_{n+m}\right] \right] & = \sum_{x \in \mathbb{X}} \mathbb{E} \left[ (\hat p_{\text{\tiny T}, m}(x))^2 \right] \mathbb{E} \left[  \frac{\mathds{1}_{Z_n(x) > 0}}{Z_n(x)} \right] V_{\text{\tiny DM},n}(x) \nonumber \\
& = \sum_{x \in \mathbb{X}} \left( \frac{p_{\text{\tiny T}}(x)(1-p_{\text{\tiny T}}(x))}{m} + p_{\text{\tiny T}}^2 (x)  \right) \mathbb{E} \left[  \frac{\mathds{1}_{Z_n(x) > 0}}{Z_n(x)} \right] V_{\text{\tiny DM},n}(x). \label{proof_sec_var_pi_est}
\end{align}

Gathering \eqref{proof_first_var_pi_est} and \eqref{proof_sec_var_pi_est}, we finally obtain, 
\begin{align*}
& \operatorname{Var}\left[    \hat \tau_{n, m} \right]\\
& =  \frac{1}{m}  \operatorname{Var}\left[   \tau(X)  - C_n(X) \right] + \left( 1 - \frac{1}{m} \right) \operatorname{Var} \left[ \mathbb{E} \left[ C_n(X) | \mathbf{X}_n \right] \right] \\
& \quad +    \sum_{x \in \mathbb{X}} \left( \frac{p_{\text{\tiny T}}(x)(1-p_{\text{\tiny T}}(x))}{m} + p_{\text{\tiny T}}^2 (x)  \right) \mathbb{E} \left[  \frac{\mathds{1}_{Z_n(x) > 0}}{Z_n(x)} \right] V_{\text{\tiny DM},n}(x),
\end{align*}
where 
\begin{align*}
    C_n(X) = \mathbb{E} \left[ Y^{(1)} \mid X  \right] (1-\pi(X))^{Z_n(X)}  - \mathbb{E} \left[ Y^{(0)} \mid X \right]\pi(X)^{Z_n(X)}.
\end{align*}
Note that, by Jensen's inequality,  
\begin{align*}
 & \quad \operatorname{Var} \left[ \mathbb{E} \left[  C_n(X)  | \mathbf{X}_n \right] \right] \\
 & \leq \mathbb{E} \left[  C_n(X)^2 \right]\\
 & \leq 2 \mathbb{E} \left[ \mathbb{E} \left[ Y^{(1)} \mid X  \right]^2 (1-\pi(X))^{2Z_n(X)} \right] + 2 \mathbb{E} \left[ \mathbb{E} \left[ Y^{(0)} \mid X \right]^2\pi(X)^{2Z_n(X)}\right] \\
 & \leq 2 \mathbb{E} \left[ \mathbb{E} \left[ Y^{(1)} \mid X  \right]^2 \mathbb{E} \left[ (1-\pi(X))^{2Z_n(X)} \mid X \right] \right] + 2 \mathbb{E} \left[ \mathbb{E} \left[ Y^{(0)} \mid X \right]^2 \mathbb{E} \left[ \pi(X)^{2Z_n(X)} \mid X \right] \right] \\
 & \leq 2 \mathbb{E} \left[ \mathbb{E} \left[ Y^{(1)} \mid X  \right]^2  \left( 1 - \left(1 - \pi(X)^2 \right) p_{\text{\tiny R}}(X) \right)^n \right] + 2 \mathbb{E} \left[ \mathbb{E} \left[ Y^{(0)} \mid X \right]^2 \left( 1 - \left(1 - (1-\pi(X))^2 \right) p_{\text{\tiny R}}(X) \right)^n \right] \\
 & \leq 2 \left( 1 - \min_x \left( (1 - \tilde \pi(x)^2) p_{\text{\tiny R}}(x) \right) \right)^n   \mathbb{E} \left[ (Y^{(1)})^2 + (Y^{(0)})^2 \right],
\end{align*}
where $\tilde \pi(x) = \max ( \pi(x), 1 - \pi(x) )$, and we have used the fact that 
\begin{align*}
\mathbb{E} \left[ \pi(X)^{2Z_n(X)} \mid X \right] & = \left( \pi(X)^2 p_{\text{\tiny R}}(X) + 1 - p_{\text{\tiny R}}(X) \right)^n \\
& =  \left( 1 - \left(1 - \pi(X)^2 \right) p_{\text{\tiny R}}(X) \right)^n.
\end{align*}

Besides, we have
\begin{align*}
 \operatorname{Var}\left[   \tau(X)  - C_n(X) \right] \leq \operatorname{Var}_{\text{\tiny T}} \left[ \tau(X)\right] + 2 \left( \operatorname{Var}_{\text{\tiny T}}[\tau(X)] \operatorname{Var}_{\text{\tiny T}}\left[C_n(X) \right] \right)^{1/2} + \operatorname{Var}_{\text{\tiny T}} \left[ C_n(X) \right],   
\end{align*}
where
\begin{align*}
\operatorname{Var}_{\text{\tiny T}} \left[ C_n(X) \right] & \leq \mathbb{E} \left[ C_n(X)^2 \right] \\
& \leq 2 \left( 1 - \min_x \left( (1 - \tilde \pi(x)^2) p_{\text{\tiny R}}(x) \right) \right)^n   \mathbb{E} \left[ (Y^{(1)})^2 + (Y^{(0)})^2 \right].
\end{align*}
Consequently, 
\begin{align*}
 \operatorname{Var}\left[   \tau(X)  - C_n(X) \right] & \leq \operatorname{Var}_{\text{\tiny T}} \left[ \tau(X)\right] + 4 \left( 1 - \min_x \left( (1 - \tilde \pi(x)^2) p_{\text{\tiny R}}(x) \right) \right)^{n/2} \mathbb{E} \left[ (Y^{(1)})^2 + (Y^{(0)})^2 \right] \\
 & \quad + 2 \left( 1 - \min_x \left( (1 - \tilde \pi(x)^2) p_{\text{\tiny R}}(x) \right) \right)^n   \mathbb{E} \left[ (Y^{(1)})^2 + (Y^{(0)})^2 \right]\\
 & \leq \operatorname{Var}_{\text{\tiny T}} \left[ \tau(X)\right] + 6 \left( 1 - \min_x \left( (1 - \tilde \pi(x)^2) p_{\text{\tiny R}}(x) \right) \right)^{n/2} \mathbb{E} \left[ (Y^{(1)})^2 + (Y^{(0)})^2 \right].
\end{align*}

Finally, 
\begin{align*}
\operatorname{Var}\left[    \hat \tau_{n, m} \right]  & \leq   \frac{2}{n+1}   \mathbb{E}_{\text{\tiny R}}\left[ \left( \frac{p_\text{\tiny T}\left( X \right)}{p_\text{\tiny R}\left( X \right)}\right)^2 V_{\text{\tiny DM}}(X)\right] + \frac{\operatorname{Var}\left[   \tau(X) \right] }{m} +  \frac{2}{(n+1)m}  \mathbb{E}_{\text{\tiny R}}\left[ \frac{p_\text{\tiny T}\left( X \right)(1-p_\text{\tiny T}\left( X \right))}{p_\text{\tiny R}\left( X \right)^2} V_{\text{\tiny DM},n}(X)\right]  \\
& \quad + 2 \left( 1 + \frac{3}{m} \right) \left( 1 - \min_x \left( (1 - \tilde \pi(x)^2) p_{\text{\tiny R}}(x) \right) \right)^{n/2} \mathbb{E} \left[ (Y^{(1)})^2 + (Y^{(0)})^2 \right] .
\end{align*}

\end{proof}

\subsubsection{Proof of Corollary~\ref{cor_asympt_completely_estimated_pi_estimated}}\label{proof:cor_asympt_completely_estimated_pi_estimated}
\begin{proof}

The proof follows exactly the same structure as that of the proof of Corollary~\ref{cor_asympt_completely_estimated}.

\begin{proof}
We recall the explicit expression of the variance of $\hat \tau_{n, m}$,
\begin{align*}
& \operatorname{Var}\left[    \hat \tau_{n, m} \right]\\
& =  \frac{1}{m}  \operatorname{Var}\left[   \tau(X)  - C_n(X) \right] + \left( 1 - \frac{1}{m} \right) \operatorname{Var} \left[ \mathbb{E} \left[ C_n(X) | \mathbf{X}_n \right] \right] \\
& \quad +    \sum_{x \in \mathbb{X}} \left( \frac{p_{\text{\tiny T}}(x)(1-p_{\text{\tiny T}}(x))}{m} + p_{\text{\tiny T}}^2 (x)  \right) \mathbb{E} \left[  \frac{\mathds{1}_{Z_n(x) > 0}}{Z_n(x)} \right] V_{\text{\tiny DM}}(x),
\end{align*}
where 
\begin{align*}
    C_n(X) = \mathbb{E} \left[ Y^{(1)} \mid X  \right] (1-\pi(X))^{Z_n(X)}  - \mathbb{E} \left[ Y^{(0)} \mid X \right]\pi(X)^{Z_n(X)}.
\end{align*}

Recall that using \eqref{eq_proof_Chernoff1} and \eqref{eq_proof_Chernoff2}, one has
\begin{align*}
\lim _{n \rightarrow \infty}  \mathbb{E} \left[ \frac{\mathbbm{1}_{Z_n(x)>0}}{  Z_n(x)/n } \right] = \frac{1}{p_{\text{\tiny R}}(x)},
\end{align*}

and we also have
\begin{equation*}
    \lim _{n \rightarrow \infty} \operatorname{Var}_{\text{\tiny T}}\left[\tau(X) - C_n(X)\right]=\operatorname{Var}_{\text{\tiny T}}[\tau(X)] = \operatorname{Var}[\tau(X)].
\end{equation*}

Finally, note that the term $\operatorname{Var} \left[ \mathbb{E} \left[ C_n(X) | \mathbf{X}_n \right] \right]$ can be bounded by a term proportional to $(1-\operatorname{min}(\pi, 1-\pi))^n$, so that the convergence toward $0$ it at an exponential pace with $n$.\\

Multiplying the explicit variance by $\operatorname{min}(n,m)$ one has,

\begin{align*}
& \operatorname{min}(n,m) \operatorname{Var}\left[    \hat \tau_{n, m} \right]\\
& =  \frac{\operatorname{min}(n,m)}{m}  \operatorname{Var}\left[   \tau(X)  - C_n(X) \right] + \operatorname{min}(n,m) \left( 1 - \frac{1}{m} \right) \operatorname{Var} \left[ \mathbb{E} \left[ C_n(X) | \mathbf{X}_n \right] \right] \\
& \quad +    \frac{\operatorname{min}(n,m)}{n} \sum_{x \in \mathbb{X}} \left( \frac{p_{\text{\tiny T}}(x)(1-p_{\text{\tiny T}}(x))}{m} + p_{\text{\tiny T}}^2 (x)  \right) \mathbb{E} \left[  \frac{\mathds{1}_{Z_n(x) > 0}}{Z_n(x)/n} \right] V_{\text{\tiny DM}}(x).
\end{align*}

Now, we study an asymptotic regime where $n$ and $m$ can grow toward infinity but at different paces. Let $  \lim\limits_{n,m\to\infty} \frac{m}{n} = \lambda \in [0,\infty]$,where $\lambda$ characterizes the regime.\\

\textbf{Case 1}:
If $\lambda \in [1, \infty]$, one can replace $\min(n,m) $ by $n$, so that 
\begin{align*}
\operatorname{min}(n,m) \operatorname{Var}\left[    \hat \tau_{n, m} \right] & =n \operatorname{Var}\left[    \hat \tau_{n, m} \right] \\
& =  \frac{1}{\lambda}  \operatorname{Var}\left[   \tau(X)  - C_n(X) \right] +  \left( n - \frac{1}{\lambda} \right) \operatorname{Var} \left[ \mathbb{E} \left[ C_n(X) | \mathbf{X}_n \right] \right] \\
& \quad +   \sum_{x \in \mathbb{X}} \left( \frac{p_{\text{\tiny T}}(x)(1-p_{\text{\tiny T}}(x))}{m} + p_{\text{\tiny T}}^2 (x)  \right) \mathbb{E} \left[  \frac{\mathds{1}_{Z_n(x) > 0}}{Z_n(x)} \right] V_{\text{\tiny DM}}(x),
\end{align*}

such that

\begin{align*}
     \lim _{n,m \rightarrow \infty} n \operatorname{Var}\left[    \hat \tau_{n, m} \right] &=  \frac{ \operatorname{Var}\left[   \tau(X) \right] }{\lambda} + \tilde V_{so},
\end{align*}

where

\begin{align*}
    \tilde V_{so}:= \mathbb{E}_{\text{\tiny R}}\left[ \left( \frac{p_{\text{\tiny T}}(X)}{p_{\text{\tiny R}}(X)} \right)^2V_{\text{\tiny DM}, \infty}(X) \right].
\end{align*}

\textbf{Case 2}:
If $\lambda \in [0, 1]$, one can replace $\min(n,m) $ by $m$, so that,

\begin{align*}
& \min(n,m) \operatorname{Var}\left[    \hat \tau_{n, m} \right] = m \operatorname{Var}\left[    \hat \tau_{n, m} \right] \\
& =   \operatorname{Var}\left[   \tau(X)  - C_n(X) \right] + m \left( 1 - \frac{1}{m} \right) \operatorname{Var} \left[ \mathbb{E} \left[ C_n(X) | \mathbf{X}_n \right] \right] \\
& \quad +  \lambda  \sum_{x \in \mathbb{X}} \left( \frac{p_{\text{\tiny T}}(x)(1-p_{\text{\tiny T}}(x))}{m} + p_{\text{\tiny T}}^2 (x)  \right) \mathbb{E} \left[  \frac{\mathds{1}_{Z_n(x) > 0}}{Z_n(x)/n} \right] V_{\text{\tiny DM}}(x).
\end{align*}

Because $m \le n$, then

\begin{align*}
    m \left( 1 - \frac{1}{m} \right) \operatorname{Var} \left[ \mathbb{E} \left[ C_n(X) | \mathbf{X}_n \right] \right] &\le  n \left( 1 - \frac{1}{m} \right) \operatorname{Var} \left[ \mathbb{E} \left[ C_n(X) | \mathbf{X}_n \right] \right] ,
\end{align*}

so that 
\begin{align*}
      \lim _{n,m \rightarrow \infty}  m \left( 1 - \frac{1}{m} \right) \operatorname{Var} \left[ \mathbb{E} \left[ C_n(X) | \mathbf{X}_n \right] \right] &\le  \lim _{n,m \rightarrow \infty}  n \left( 1 - \frac{1}{m} \right) \operatorname{Var} \left[ \mathbb{E} \left[ C_n(X) | \mathbf{X}_n \right] \right] =0.
\end{align*}

Finally,
\begin{align*}
 \lim _{n,m \rightarrow \infty} m \operatorname{Var}\left[    \hat \tau_{n, m} \right]
& =   \operatorname{Var}\left[   \tau(X)  \right] +  \lambda  \tilde V_{so}.
\end{align*}

\end{proof}

\end{proof}

\subsubsection{Proof of Theorem~\ref{thm:ipsw_pi_est}}\label{proof:thm:ipsw_pi_est}
\begin{proof}

According to Proposition~\ref{prop:when-estimating-pi}, the bias of the IPSW estimator with estimated $\hat \pi_n$ can be upper bounded via
\begin{align*}
    \left| \mathbb{E} \left[ \hat \tau_{n, m} \right]- \tau \right| & \leq  \sum_{x\in\mathds{X}} p_{\text{\tiny T}}(x) \left|\mathbb{E}\left[ Y^{(0)} \mid X = x\right] \right|  \left( 1 - p_{\text{\tiny R}}(x) \left(1 - \pi(x)\right) \right)^n \\
& \quad + \sum_{x\in\mathds{X}} p_{\text{\tiny T}}(x) \left|\mathbb{E}\left[ Y^{(1)} \mid X = x \right] \right|\left(  1- p_{\text{\tiny R}}(x)\pi(x)  \right)^n \\
    & \leq \left(  1- \min_x \left( (1 - \tilde \pi(x)) p_{\text{\tiny R}}(x) \right) \right)^n  \mathbb{E}_\text{\tiny T} \left[ \left|\mathbb{E}\left[ Y^{(1)} \mid X \right] \right| + \left|\mathbb{E}\left[ Y^{(0)} \mid X \right] \right| \right].
\end{align*}
Therefore, the risk of the (estimated) IPSW estimate with estimated $\hat \pi_n$ satisfies,
\begin{align*}
& \quad \mathbb{E}\left[ \left( \hat  \tau_{n,m} - \tau \right)^2\right] \\
&\leq \left(  1- \min_x \left( (1 - \tilde \pi(x)) p_{\text{\tiny R}}(x) \right) \right)^{2n}  \mathbb{E}_\text{\tiny T} \left[ \left|\mathbb{E}\left[ Y^{(1)} \mid X \right] \right| + \left|\mathbb{E}\left[ Y^{(0)} \mid X \right] \right| \right]^2 +  \frac{2}{n+1}   \mathbb{E}_{\text{\tiny R}}\left[ \left( \frac{p_\text{\tiny T}\left( X \right)}{p_\text{\tiny R}\left( X \right)}\right)^2 V_{\text{\tiny DM}}(X)\right]  \\
& \quad + \frac{\operatorname{Var}\left[   \tau(X) \right] }{m}  +  \frac{2}{(n+1)m}  \mathbb{E}_{\text{\tiny R}}\left[ \frac{p_\text{\tiny T}\left( X \right)(1-p_\text{\tiny T}\left( X \right))}{p_\text{\tiny R}\left( X \right)^2} V_{\text{\tiny DM}}(X)\right]  \\
& \quad +  2 \left( 1 + \frac{3}{m} \right) \left( 1 - \min_x \left( (1 - \tilde \pi(x)^2) p_{\text{\tiny R}}(x) \right) \right)^{n/2} \mathbb{E} \left[ (Y^{(1)})^2 + (Y^{(0)})^2 \right]\\
 & \le  \frac{2}{n+1}   \mathbb{E}_{\text{\tiny R}}\left[ \left( \frac{p_\text{\tiny T}\left( X \right)}{p_\text{\tiny R}\left( X \right)}\right)^2 V_{\text{\tiny DM}}(X)\right] + \frac{\operatorname{Var}\left[   \tau(X) \right] }{m} + \frac{2}{m(n+1)}  \mathbb{E}_{\text{\tiny R}}\left[ \frac{p_\text{\tiny T}\left( X \right)(1-p_\text{\tiny T}\left( X \right))}{p_\text{\tiny R}\left( X \right)^2} V_{\text{\tiny DM}}(X)\right]  \\
& \quad   +  2 \left( 2 + \frac{3}{m} \right) \left( 1 - \min_x \left( (1 - \tilde \pi(x)) p_{\text{\tiny R}}(x) \right) \right)^{n/2} \mathbb{E} \left[ (Y^{(1)})^2 + (Y^{(0)})^2 \right].
\end{align*}

\end{proof}

\section{Extended adjustment set}\label{appendix:extended-adjustment-set}

\subsection{Proof of Corollary~\ref{proposition:adding-shifted-covariates}}\label{proof_adding-shifted-covariates}

\begin{proof}

According to Corollary~\ref{cor_asympt_semi_oracle}, we have
\begin{align}
\lim_{n\to\infty} n \operatorname{Var}\left[\hat \tau_{\text{\tiny T},n}^*(X) \right] 
    &= V_{\text{so}}, \label{eq_proof_add_set_shifted1}
\end{align}
where
\begin{equation*}
    V_{\text{so}}:= \mathbb{E}_\text{\tiny R}\left[ \left(\frac{p_\text{\tiny T}(X)}{p_\text{\tiny R}(X)} \right)^2V_{ \text{\tiny HT}}(X)\right],
\end{equation*}
with
\begin{equation*}
    V_{ \text{\tiny HT}}(x) = \mathbb{E}_{\text{\tiny R}}\left[ \frac{\left( Y^{(1)} \right)^2}{\pi} \mid X = x\right]  + \mathbb{E}_{\text{\tiny R}}\left[ \frac{\left( Y^{(0)} \right)^2}{1-\pi} \mid X = x\right]  - \tau(x)^2.
\end{equation*}
Since, by assumption, $V$ is composed of covariates that are not treatment effect modifiers, using Definition~\ref{def:V-is-not-treat-effect-modifier}, we have, for all $(x,v)$,
\begin{align}
 V_{ \text{\tiny HT}}(x,v) = V_{ \text{\tiny HT}}(x).
 \label{eq_proof_add_set_shifted2}
\end{align}
Now, considering the set $(X,V)$ instead of $X$ in the expression~\eqref{eq_proof_add_set_shifted1} leads to  
\begin{align*}
    \lim_{n\to\infty} n \operatorname{Var}_\text{\tiny R}\left[\hat \tau_{\text{\tiny T},n}^*(X,V) \right] 
    &=   \mathbb{E}_\text{\tiny R}\left[ \left(\frac{p_\text{\tiny T}(X,V)}{p_\text{\tiny R}(X,V)} \right)^2V_{ \text{\tiny HT}}(X,V)\right] \\
    & = \sum_{x,v\in \mathcal{X}, \mathcal{V}} \frac{p_\text{\tiny T}^2(x,v)}{p_\text{\tiny R}(x,v)} 
    V_{ \text{\tiny HT}}(x,v)\\
     &=   \sum_{x,v\in \mathcal{X}, \mathcal{V}} \frac{p_\text{\tiny T}^2(x,v)}{p_\text{\tiny R}(x,v)} V_{ \text{\tiny HT}}(x) && \text{Equation.~\eqref{eq_proof_add_set_shifted2}} \\
     &=   \sum_{x,v\in \mathcal{X}, \mathcal{V}} \frac{p_\text{\tiny T}^2(x)p_\text{\tiny T}^2(v)}{p_\text{\tiny R}(x)p_\text{\tiny R}(v)} V_{ \text{\tiny HT}}(x)   && \text{$V \indep X$} \\
    &=  \left( \sum_{v \in \mathcal{V}} \frac{p_{\text{\tiny T}}(v)^2}{p_{\text{\tiny R}}(v)} \right)  \sum_{x\in \mathcal{X}} \frac{p_\text{\tiny T}^2(x)}{p_\text{\tiny R}(x)} V_{ \text{\tiny HT}}(x) \\
&=  \left( \sum_{v \in \mathcal{V}} \frac{p_{\text{\tiny T}}(v)^2}{p_{\text{\tiny R}}(v)} \right)   \lim_{n\to\infty} n \operatorname{Var}_\text{\tiny R} \left[\hat \tau_{\text{\tiny T},n}^*(X) \right],
    \end{align*} 
Now, note that 
\begin{align*}
\sum_{v \in \mathcal{V}} \frac{p_{\text{\tiny T}}(v)^2}{p_{\text{\tiny R}}(v)}  &= \mathbb{E}_{\text{\tiny R}} \left[ \frac{p_{\text{\tiny T}}(V)^2}{p_{\text{\tiny R}}(V)^2} \right] \\
& \geq \left( \mathbb{E}_{\text{\tiny R}} \left[ \frac{p_{\text{\tiny T}}(V)}{p_{\text{\tiny R}}(V)} \right] \right)^2 \\
& \geq \left( \sum_{v \in \mathcal{V}} p_{\text{\tiny T}}(v) \right)^2\\
& \geq 1,
\end{align*}
where the first inequality results from Jensen's inequality. Consequently, 
\begin{align*}
\lim_{n\to\infty} n \operatorname{Var}_\text{\tiny R}\left[\hat \tau_{\text{\tiny T},n}^*(X,V) \right] \geq \lim_{n\to\infty}n\operatorname{Var}_\text{\tiny R} \left[\hat \tau_{\text{\tiny T},n}^*(X) \right].
\end{align*}





\end{proof}

\subsection{Proof of Corollary~\ref{proposition:adding-treat-effect-modifier-covariates}}\label{proof_adding-treat-effect-modifier-covariates}

\begin{proof}

By the law of total variance, we have, for all $x$, 
\begin{align}
V_{ \text{\tiny DM}}(x)  =  \mathbb{E} \left[V_{ \text{\tiny DM}}(x,V) \right] +  \operatorname{Var}\left[\tau(x,V)\right].
\label{proof_eq_add_cov_unshifted_total_variance}
\end{align}
Indeed, according to the law of total variance, for all random variables $Z, X_1, X_2$, we have, a.s., 
\begin{align*}
\operatorname{Var}\left[ Z \mid X_1\right] = \mathbb{E} \left[ \operatorname{Var}\left[  Z \mid X_1, X_2 \right]\mid X_1 \right] +  \operatorname{Var} \left[ \mathbb{E}\left[  Z \mid X_1, X_2 \right]\mid X_1 \right].
\end{align*}
Letting $X_1 = X, X_2 = V$ and $Z = (YA/\pi) - (Y(1-A)/\pi)$ yields equation~\eqref{proof_eq_add_cov_unshifted_total_variance}.
Now, we can write
\begin{align*}
    \lim_{n\to\infty} n \operatorname{Var}_\text{\tiny R}\left[\hat \tau_{\text{\tiny T},n}^*(X,V) \right]
    &=   \mathbb{E}_\text{\tiny R}\left[ \left(\frac{p_\text{\tiny T}(X,V)}{p_\text{\tiny R}(X,V)} \right)^2V_{ \text{\tiny DM}}(X,V)\right] \\
    & = \sum_{x,v\in \mathcal{X}, \mathcal{V}} \frac{p_\text{\tiny T}^2(x,v)}{p_\text{\tiny R}(x,v)} 
    V_{ \text{\tiny DM}}(x,v)\\
     &=   \sum_{x,v\in \mathcal{X}, \mathcal{V}} \frac{p_\text{\tiny T}^2(x)p_\text{\tiny T}^2(v)}{p_\text{\tiny R}(x)p_\text{\tiny R}(v)} V_{ \text{\tiny DM}}(x,v)   && \text{$V \indep X$} \\
     &=   \sum_{x  \mathcal{X} } \frac{p_\text{\tiny T}^2(x)}{p_\text{\tiny R}(x)} \sum_{v\in  \mathcal{V}} \frac{p_\text{\tiny T}^2(v)}{p_\text{\tiny R}(v)} V_{ \text{\tiny DM}}(x,v)   &&   \\
     &=   \sum_{x  \mathcal{X} } \frac{p_\text{\tiny T}^2(x)}{p_\text{\tiny R}(x)} \sum_{v\in  \mathcal{V}} p_\text{\tiny T}(v) V_{ \text{\tiny HT}}(x,v)   &&  \text{by Definition~\ref{def:V-is-not-shifted}} \\
     &=   \sum_{x  \mathcal{X} } \frac{p_\text{\tiny T}^2(x)}{p_\text{\tiny R}(x)} \left( V_{ \text{\tiny DM}}(x) - \operatorname{Var}\left[\tau(x,V)\right] \right)   && \text{Equation~\eqref{proof_eq_add_cov_unshifted_total_variance}}   \\
    &= \lim_{n\to\infty} n \operatorname{Var}_\text{\tiny R}\left[\hat \tau_{\text{\tiny T},n}^*(X) \right] - 
    \mathbb{E}_\text{\tiny R} \left[ \frac{p_\text{\tiny T}(X)}{p_\text{\tiny R}(X)} \operatorname{Var}\left[ \tau(X,V) \mid X \right] \right], 
    \end{align*} 
which concludes the proof. 
\end{proof}

\section{Semi-synthetic simulation's data preparation}\label{appendix:additional-info-semi-synthetic-simulation}

\subsection{Context}
The semi-synthetic simulation is made of real world data, a trial called CRASH-3 \citep{crash3protocol, crash32019} and an observational data base called Traumabase. 
The covariates of both data sources are used to generate the true distribution from which the simulated data are generated. This part details the pre-treatment performed on the covariates, which is contained in the \texttt{R} notebook entitled \texttt{Prepare-semi-synthetic-simulation.Rmd}.
As explained in the main document, in this semi-synthetic simulation we only consider six baseline covariates:

\begin{itemize}
    \item Glasgow Coma Scale score\footnote{The GCS is a neurological scale which aims to assess a person's consciousness. The lower the score, the higher the severity of the trauma.} (GCS) (categorical);
    \item Gender (categorical);

    \item Pupil reactivity (categorical);
        \item Age (continuous);
    \item Systolic blood pressure (continuous);
    \item Time-to-treatment (continuous), being the time between the trauma and the administration of the treatment.
\end{itemize}

As three covariates out of 6 are continuous, we categorize them to obtain a completely categorical data. The time-to-treatment is categorized in 4 levels, systolic blood pressure in 3 levels, and age in 3 levels.
To further reduce the number of categories, and follow the CRASH-3 trial stratification, the Glasgow score is also gathered in 3 levels, from severe to moderately injured individuals, based on their Glasgow score.

\paragraph{CRASH-3 trial}

The CRASH-3 trial data contains information on $12,737$ individuals. Over the six covariates of interest and the $12,737$ individuals, $108$ values are missing. We imputed them using the \texttt{R} package \texttt{missRanger}.

\paragraph{Traumabase observational data}
The complete Traumabase data contains $20,037$ observations, but when keeping only the individuals suffering from Traumatic Brain Injury (TBI) as it is the case in the CRASH-3 trial, only $8,289$ observations could be kept. Many data are missing, in particular $2,660$ missing values for $8,289$ individuals and along $5$ baseline covariates considered. We impute them with the \texttt{R} package \texttt{missRanger}, using $35$ other available baseline covariates.
Because the time to treatment is not observed in the Traumabase this covariate is generated following a beta law, and considering a shifted distribution compared to the trial, in particular toward lower time-to-treatment values than in the trial.

\paragraph{Ensuring overlap}
When binding the two data sets, we had to ensure that the support inclusion assumption (Assumption~\ref{a:pos}) was verified. 
Out of the $586$ modalities present in the target data, only $192$ are also present in the trial data. Therefore only these observations are kept, such that the observational sample finally contains $8,058$ observations ($8,289$ at the beginning). All the observations in the trial are kept as there is no restriction for the trial to contain a larger support as presented in Assumption~\ref{a:pos}.

\subsubsection{Covariate shift vizualization}

For each of the six baseline and categorical considered, visualization of the covariate shift between the two data source is represented on Figures~\ref{fig:semi-synthetic-age}, \ref{fig:semi-synthetic-blood}, \ref{fig:semi-synthetic-gender}, \ref{fig:semi-synthetic-glasgow}, \ref{fig:semi-synthetic-eye}, and \ref{fig:semi-synthetic-ttt}.

\begin{figure}[!h]
    \centering
    \includegraphics[width=0.4\textwidth]{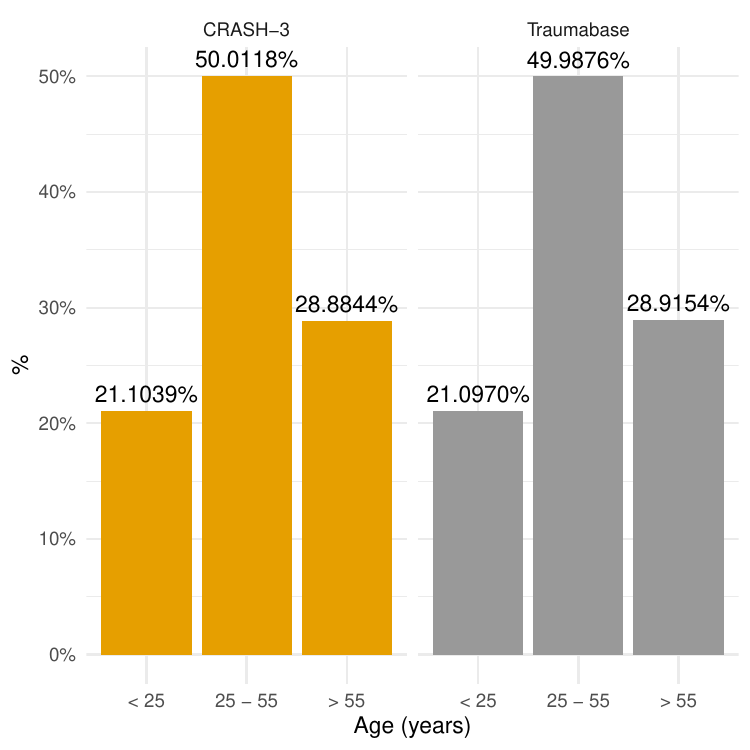}
    \caption{Bar plot of categorized age in the semi-synthetic simulation}
    \label{fig:semi-synthetic-age}
\end{figure}

\begin{figure}[!h]
    \centering
    \includegraphics[width=0.4\textwidth]{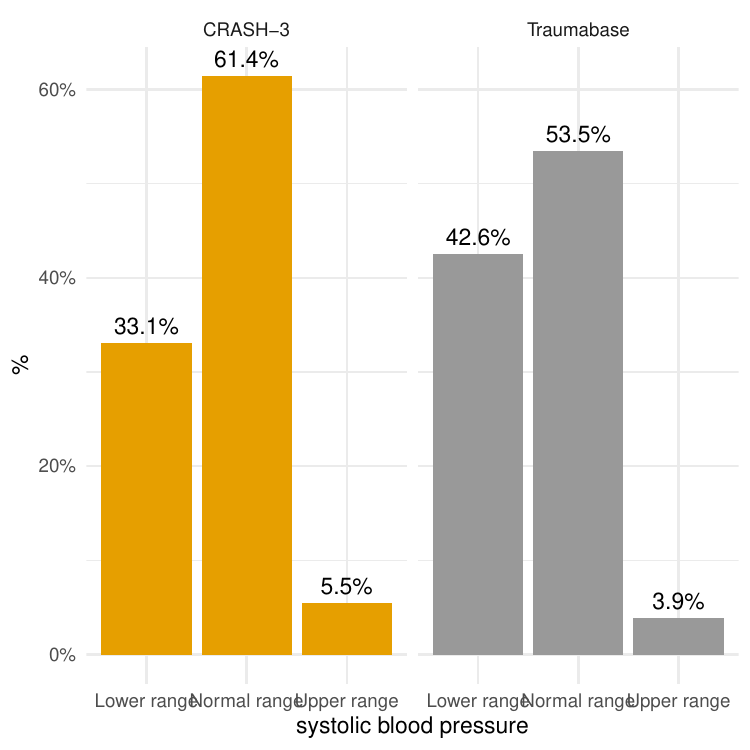}
    \caption{Bar plot of categorized systolic blood pressure in the semi-synthetic simulation}
    \label{fig:semi-synthetic-blood}
\end{figure}

\begin{figure}[!h]
    \centering
    \includegraphics[width=0.4\textwidth]{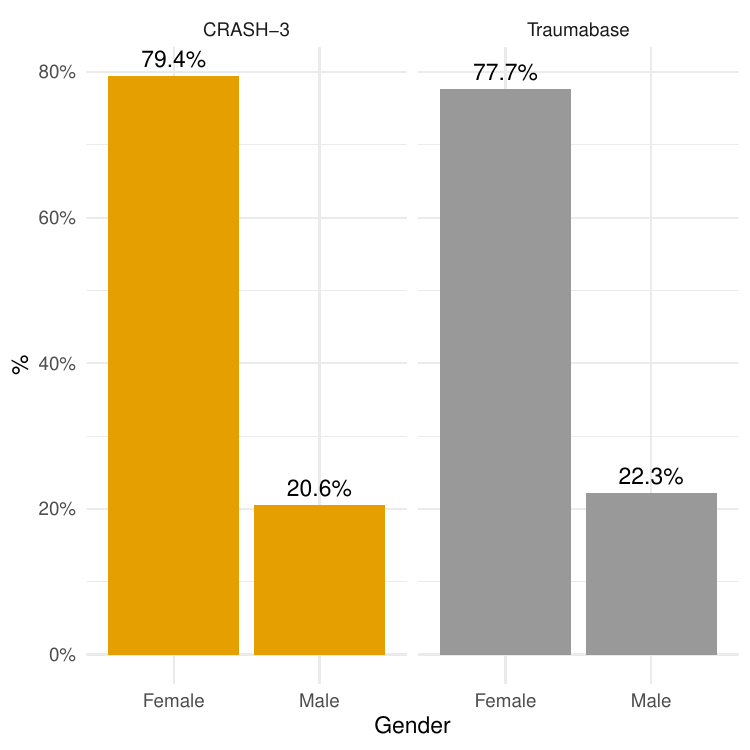}
    \caption{Bar plot of gender in the semi-synthetic simulation}
    \label{fig:semi-synthetic-gender}
\end{figure}

\begin{figure}[!h]
    \centering
    \includegraphics[width=0.4\textwidth]{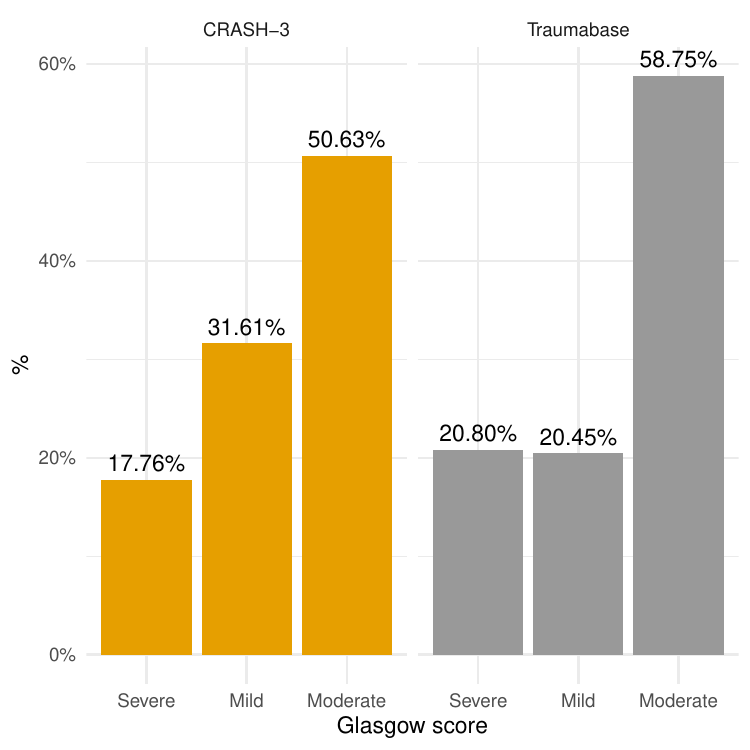}
    \caption{Bar plot of the glasgow score in the semi-synthetic simulation}
    \label{fig:semi-synthetic-glasgow}
\end{figure}

\begin{figure}[!h]
    \centering
    \includegraphics[width=0.4\textwidth]{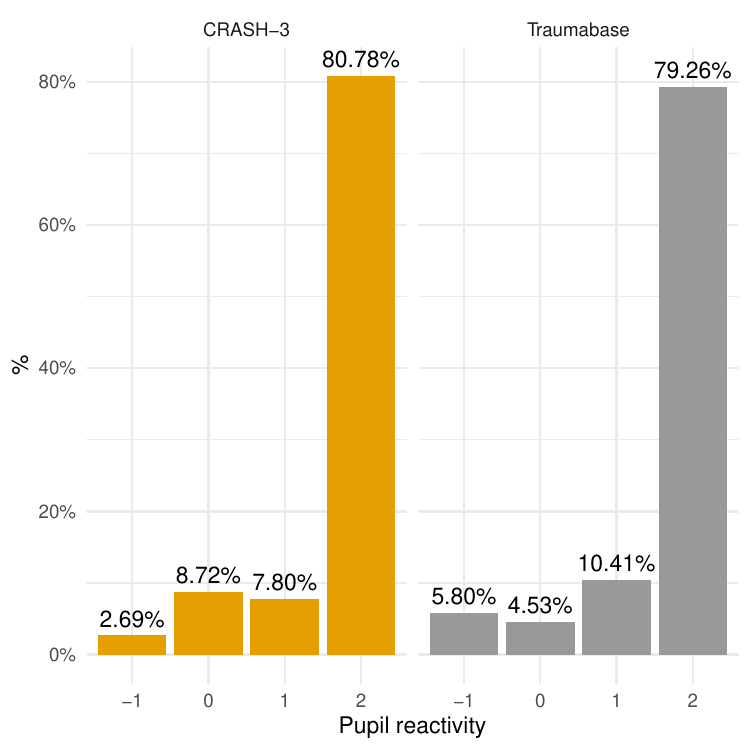}
    \caption{Bar plot of pupil reactivity ($-1$ encoding not able to measure) in the semi-synthetic simulation}
    \label{fig:semi-synthetic-eye}
\end{figure}

\begin{figure}[!h]
    \centering
    \includegraphics[width=0.4\textwidth]{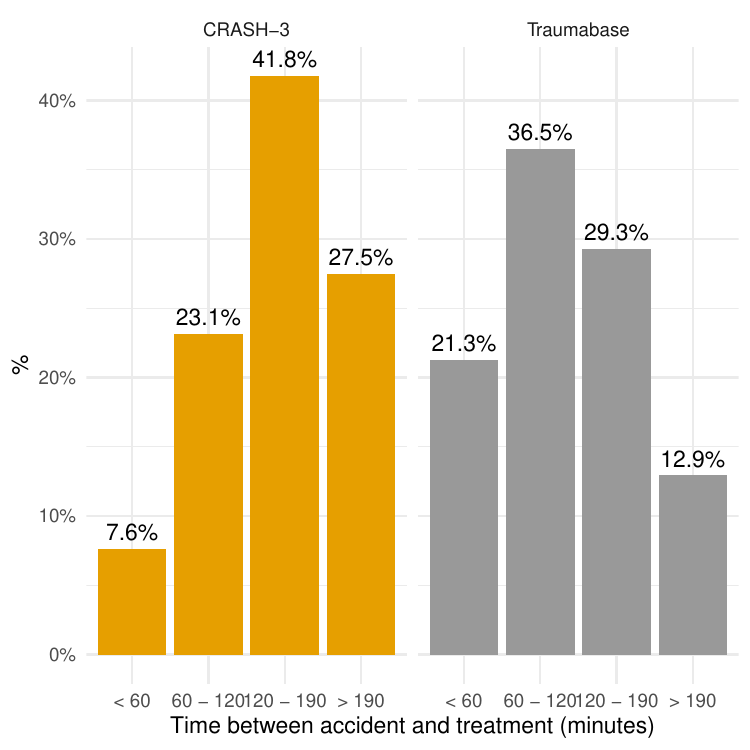}
    \caption{Bar plot of categorized time-to-treatment in the semi-synthetic simulation}
    \label{fig:semi-synthetic-ttt}
\end{figure}

\subsection{Synthetic outcome model}
As detailed above, for now the covariate support reflects a true situation, where only the time-to-treatment covariate was created as it is missing in the target population sample \citep{Colnet2021Sensitivity}.

For the purpose of simulation, the outcome model is completely synthetic, and for each strata a number is affected, from $1$ to the number of strata, starting to the lowest category (for example youngest strata, or lowest Glasgow score, or lower systolic blood pressure), to the highest one.

Doing so, the outcome model considered is such as,

\begin{align*}
    Y & = 10 -\texttt{Glasgow} +\left( \texttt{if Girl:}-5 \texttt{  else:} 0 \right)\\
    &\qquad + A \left( 15(6-\texttt{TTT}) + 3*(\texttt{Systolic.blood.pressure}-1)^2 \right) + \varepsilon_{\texttt{TTT}},
\end{align*}

where $\varepsilon_{\texttt{TTT}}$ is a random Gaussian noise with a standard deviation depending on the value of the covariate \texttt{TTT}. In particular if the treatment is given later, then the noise is stronger.

\newpage
\section{Useful results about RCTs under a Bernoulli design}\label{appendix:useful-results-rct}


Here we recall the definition of a Bernoulli trial (see Definition~\ref{def:bernoulli-trial}) and results such as variance expression of the Horvitz-Thomson and difference-in-means estimators under this design. We also provide details about variance inequality between the variance of the Horvitz-Thomson compared to the variance of the difference-in-means.
In the literature we have not found detailed derivations about the finite sample bias and variance of the difference-in-means \underline{under a Bernoulli design}. 
Extensively detailed derivations are available in Chapter two of \cite{imbens2015causal}, but for a completely randomized design.
Also note that in this work we assume a superpopulation framework, and a large part of the existing literature focuses on inference on a finite population. Indeed, when considering a finite sample, bias and variance of the Horvitz-Thomson and difference-in-means are not the same as when inferring the superpopulation treatment effect \citep{SplawaNeyman1990Translation, imbens2011experimental, miratrix2013adjusting, Harshaw2021VarianceEstimation}.\\

Note that all the results in this section considers one population, and not two populations with two distributions (target and randomized), therefore no index is placed on the expectation. When the following results on RCTs are used in the main paper and/or in the proofs, we use the index $R$ in the expectation as the trial in the main paper is sample according to $P_{\text{\tiny R}}$.

\subsection{Bernoulli trial}

A Bernoulli trial is a trial where the treatment assignment vector, being $\boldsymbol{A} = (A_1, \dots, A_n)$ follows a Bernoulli law with a constant probability. More formally,

\begin{definition}[Assignment mechanism for a Bernoulli Trial]\label{def:bernoulli-trial}
If the assignment mechanism is a Bernoulli trial with a probability $\pi$, then

$$\forall i,\, \mathbb{P}[A_i] = \pi,$$

and considering a sample for $n$ units,
$$\mathbb{P}\left[\mathbf{A} \mid i \in \mathcal{R}\right] = \prod_{i=1}^{n}\left[\pi^{A_{i}} \cdot\left(1-\pi\right)^{1-A_{i}}\right],$$
where $\mathbf{A}$ denotes the vector of treatment allocation for the trial sample $\mathcal{R}$.
\end{definition}

In this design the treatment allocation is independent of all other treatment allocations. 
A disadvantage of such design is the fact that there is always a small probability that all units receive the treatment or no treatment. 
This is why other designs are possible, such as the so-called completely randomized design, where the number of treated units is selected prior to treatment allocation (usually $n/2$ units are given treatment). The interest is to ensure a balanced group of treated and controls, and avoid a possible pathological case of high unbalance between the number of treated and control individuals. 

Mathematically, treating the situation of a completely randomized design is different than a Bernoulli design, as in the former the probability of treatement is not independent between units, for example

\begin{equation*}
  \forall i,j \in \mathcal{R},\,  \mathbb{P}_{\text{\tiny Comp. rand.}}\left[A_i = 1\mid A_j = 1 \right] \neq     \mathbb{P}_{\text{\tiny Comp. rand.}}\left[A_i = 1\right] = \pi.
\end{equation*}

\subsection{Horvitz-Thomson's}

The Horvitz-Thomson estimator is unbiased and has an explicit finite sample variance.

\begin{lemma}[Finite sample bias and variance of the Horvitz-Thomson estimator]\label{lemma:HT-unbiased-and-variance}
Assuming trial internal validity (Assumption~\ref{a:trial-internal-validity}), then 

\begin{equation*}
 \forall n, \quad    \mathbb{E}[\hat{\tau}_{\text{\tiny HT}}] - \tau = 0,
\end{equation*}
and 
\begin{equation*}
 \forall n, \quad  n\operatorname{Var}\left[ \hat{\tau}_{\text{\tiny HT},n}  \right] = \mathbb{E}\left[ \frac{\left( Y^{(1)} \right)^2}{\pi}\right]  + \mathbb{E}\left[ \frac{\left( Y^{(0)} \right)^2}{1-\pi}  \right]  - \tau^2.
\end{equation*}

\end{lemma}

Note that the following proof can be extended to any $\pi(x)$ depending on baseline covariates, and therefore extends to the oracle IPW in the causal inference literature.

\begin{proof}

\textbf{Bias}\\

\begin{align*}
    \mathbb{E}[\hat{\tau}_{\text{\tiny HT}}] &= \frac{\mathbb{E}\left[A_iY_i^{(1)} \right]}{\pi} - \frac{\mathbb{E}\left[(1-A_i)Y_i^{(0)} \right]}{1-\pi} && \text{Linearity \& SUTVA}\\
     &= \frac{\mathbb{E}\left[A_i\right]\mathbb{E}\left[Y_i^{(1)} \right]}{\pi} - \frac{\mathbb{E}\left[(1-A_i)\right]\mathbb{E}\left[Y_i^{(0)} \right]}{1-\pi} && \text{Randomization}\\
        &= \frac{\pi\mathbb{E}\left[Y_i^{(1)} \right]}{\pi} - \frac{(1-\pi)\mathbb{E}\left[Y_i^{(0)} \right]}{1-\pi} && \text{Def. of $\pi$ - Bernoulli design}\\
    &= \tau, && \text{Linearity.}
\end{align*}

\textbf{Variance}\\

\begin{align}
    \operatorname{Var}\left[ \hat{\tau}_{\text{\tiny HT},n}\right] &=    \operatorname{Var}\left[\frac{1}{n} \sum_{i=1}^n \frac{A_i Y_i}{\pi} - \frac{(1-A_i) Y_i}{1-\pi}  \right]  \nonumber \\
    &= \frac{1}{n^2 }  \operatorname{Var}\left[  \sum_{i=1}^n \frac{A_i Y^{(1)}_i}{\pi} - \frac{(1-A_i) Y^{(0)}_i}{1-e}  \right]   && \text{Assumption~\ref{a:trial-internal-validity}} \nonumber \\
    &= \frac{1}{n}  \operatorname{Var}\left[  \frac{A Y^{(1)}}{\pi} - \frac{(1-A) Y^{(0)}}{1-\pi} \right].  && \text{\textit{iid}} \nonumber 
\end{align}

Then,

\begin{align}
    \operatorname{Var}\left[ \hat{\tau}_{\text{\tiny HT},n}  \right] &=  \frac{1}{n} \left(\operatorname{Var}\left[  \frac{A Y^{(1)}}{\pi}   \right] + \operatorname{Var}\left[  \frac{(1-A) Y^{(0)}}{1-\pi}  \right] -2\, \operatorname{Cov}\left[ \frac{A Y^{(1)}}{\pi} , \frac{(1-A) Y^{(0)}}{1-\pi}   \right]  \right). \label{proof_variance_oracle_IPW}
\end{align}

The first two terms can be simplified, noting that
    
    \begin{align*}
    \mathbb{E}\left[ \left( \frac{AY^{(1)}}{\pi} \right)^2  \right]
    &=\mathbb{E}\left[\mathds{1}_{\left\{A_{i}=1\right\}} \left(\frac{ Y^{(1)}}{\pi}  \right)^2\right] && \text{A is binary} \\
      &=\mathbb{E}\left[\frac{\left( Y^{(1)} \right)^2}{\pi^2}  \right]\mathbb{E}_{\text{\tiny R}}\left[\mathds{1}_{\left\{A_{i}=1\right\}}  \right] &&\text{Randomization of trial} \\
     &=\mathbb{E}\left[\frac{\left( Y^{(1)} \right)^2}{\pi}  \right] &&\text{Definition of $\pi$} 
\end{align*}

Similarly, $$\mathbb{E}\left[ \left( \frac{(1-A)Y^{(0)}}{1-\pi}  \right)^2 \right] =  \mathbb{E}\left[\frac{\left( Y^{(0)} \right)^2}{1-\pi}\right] .$$

So,

\begin{align*}
    \operatorname{Var}\left[  \frac{A Y^{(1)}}{\pi}  \right] &= \mathbb{E}\left[ \left( \frac{AY^{(1)}}{\pi} \right)^2   \right] -  \mathbb{E}\left[ \frac{A Y^{(1)}}{\pi}  \right]^2 \\
    &= \mathbb{E}\left[ \frac{\left( Y^{(1)} \right)^2}{\pi}    \right] -  \mathbb{E}\left[ Y^{(1)}  \right]^2 .
\end{align*}

Similarly,

$$\operatorname{Var}\left[  \frac{(1-A) Y^{(0)}}{1-\pi} \right] =\mathbb{E}\left[ \frac{\left( Y^{(0)} \right)^2}{1-\pi}   \right] -  \mathbb{E}\left[ Y^{(0)} \right]^2 .$$

The third term in equation \eqref{proof_variance_oracle_IPW} can also be decomposed, so that,

\begin{align*}
 \operatorname{Cov}\left[ \frac{A Y^{(1)}}{\pi} , \frac{(1-A) Y^{(0)}}{1-\pi}  \right] &= \mathbb{E}_{\text{\tiny R}}\left[ \left(\frac{A Y^{(1)}}{\pi} - \mathbb{E}\left[ Y^{(1)}   \right] \right) \left(\frac{(1-A) Y^{(0)}}{1-\pi} - \mathbb{E}_{\text{\tiny R}}\left[ Y^{(0)}   \right] \right)   \right]\\   
 &= \mathbb{E}_{\text{\tiny R}}\left[ \underbrace{\frac{A Y^{(1)}}{\pi} \frac{(1-A) Y^{(0)}}{1-\pi} }_\textrm{$=0$}   \right] - \mathbb{E}_{\text{\tiny R}}\left[ Y^{(0)}  \right]\mathbb{E}_{\text{\tiny R}}\left[ Y^{(1)}  \right].
\end{align*}

Finally,

\begin{align*}
      n\operatorname{Var}\left[ \hat{\tau}_{\text{\tiny HT},n}   \right] 
      &= \mathbb{E}_{\text{\tiny R}}\left[ \frac{\left( Y^{(1)} \right)^2}{\pi} \right]  + \mathbb{E}_{\text{\tiny R}}\left[ \frac{\left( Y^{(0)} \right)^2}{1-\pi}   \right]  - \tau^2 := V_{ \text{\tiny HT}}  .
\end{align*}

\end{proof}

\subsection{General results about the Difference-in-means}

First, note that the Difference-in-Means estimator (in Definition~\ref{def:difference-in-means}) can be re-written as,
\begin{equation*}
    \hat{\tau}_{\text{\tiny DM},n} = \frac{1}{n}\sum_{i=1}^n \frac{A_i Y_i}{\frac{\sum_{i=1}^n A_i}{n}} -  \frac{1}{n}\sum_{i=1}^n \frac{(1-A_i) Y_i}{\frac{\sum_{i=1}^n 1-A_i}{n}},
\end{equation*}

which corresponds to the Horvitz-Thomson where the probability to be treated is estimated with the data.

This estimator is always defined, even if due to the Bernoulli design it possible that all observations were allocated treatment or control. For example, if all units are given control, then

\begin{align*}
    \sum_{i=1}^n A_i &= 0,
\end{align*}

and because for all $i$, $A_i = 0$, the ratio $ \frac{1}{n}\sum_{i=1}^n \frac{A_i Y_i}{\frac{\sum_{i=1}^n A_i}{n}}$ is defined and equal to $\frac{0}{0}=0$ by convention.

\begin{lemma}[Finite sample and large sample properties of the difference-in-means estimator]\label{lemma:DM-bias-and-variance}
Assuming trial internal validity (Assumption~\ref{a:trial-internal-validity}), then 

\begin{align*}
     \forall n,\quad   \mathbb{E}\left[    \hat{\tau}_{\text{\tiny DM},n}\right] - \tau &=    \pi^n\mathbb{E}\left[  Y_i^{(0)}\right] - (1-\pi)^n\mathbb{E}\left[  Y_i^{(1)}\right],
\end{align*}
and 
\begin{align*}
      \forall n,\quad  \operatorname{Var}\left[ \hat{\tau}_{\text{\tiny DM},n} \right] &= \frac{1}{n} \left( \mathbb{E}\left[  \frac{\mathds{1}_{\hat \pi > 0} }{\hat \pi}\right] \operatorname{Var}\left[Y_i^{(1)}\right] +  \mathbb{E}\left[  \frac{\mathds{1}_{(1-\hat \pi) > 0} }{1-\hat \pi}\right]  \operatorname{Var}\left[Y_i^{(0)}\right] \right) + D_n,
\end{align*}

where $D_n = \mathbb{E}\left[  Y_i^{(1)}\right]^2 (1-\pi)^n +  \mathbb{E}\left[  Y_i^{(0)}\right]^2  \pi^n - \left(\mathbb{E}\left[  Y_i^{(1)}\right] (1-\pi)^n  +  \mathbb{E}\left[  Y_i^{(0)}\right]  \pi^n \right)^2 $.\\

Asymptotically, the difference-in-means is unbiased

\begin{align*}
     \lim_{n\to\infty} \mathbb{E}\left[    \hat{\tau}_{\text{\tiny DM},n}\right] = \tau,
\end{align*}

and has the following variance
\begin{align*}
\lim_{n\to\infty}  n\operatorname{Var}\left[ \hat{\tau}_{\text{\tiny DM},n} \right] =  \frac{  \operatorname{Var}\left[ Y^{(1)}\right] }{\pi}+  \frac{  \operatorname{Var}\left[ Y^{(0)}\right] }{1-\pi} :=  V_{ \text{\tiny DM}, \infty}.
\end{align*}

\end{lemma}

The difference-in-means under a Bernoulli design has a finite sample bias due to the possibility of a sample where everyone receive treatments or control. But the bias is exponentially decreasing with $n$. Also note that,

\begin{align*}
   D_n &=  \mathcal{O}\left(\max(\pi, 1-\pi)^n\right)
\end{align*}

The asymptotic variance of the difference-in-means is the variance usually reported in textbooks, and corresponds to the finite sample of the Difference-in-Means estimator under a completely randomized trial. Note that we could also show that the Difference-in-Means is asymptotically normally distributed, for example using M-estimation technics \citep{Stefanski2002Mestimation}. As this result is not used in this paper, we do not detail the proof.\\

Note that for a \underline{completely randomized design}, the difference-in-means is unbiased and its finite sample variance is,

\begin{align*}
       \operatorname{Var}\left[ \hat{\tau}_{\text{\tiny DM},n} \right] &=  \frac{  \operatorname{Var}\left[ Y^{(1)}\right] }{n_1}+  \frac{  \operatorname{Var}\left[ Y^{(0)}\right] }{n_0},
\end{align*}

where $n_1$ is the number of treated units ($\sim \pi n$) and $n_0$ is the number of control units ($\sim (1-\pi) n$). 
This formula is extensively used in the literature, but under a Bernoulli design this formula is true only in large sample as detailed in Lemma~\ref{lemma:DM-bias-and-variance}.

\begin{proof}

\textbf{Bias} \\

One can use the law of total expectation, conditioning on the treatment assignment vector denoted $\mathbf{A}$, 

\begin{align*}
       \mathbb{E}\left[    \hat{\tau}_{\text{\tiny DM}}\right]& =  \mathbb{E}\left[     \mathbb{E}\left[   \hat{\tau}_{\text{\tiny DM}} \mid \mathbf{A}\right] \right] \\
      &=  \mathbb{E}\left[  \frac{\frac{1}{n} \sum_{i=1}^n A_i }{\frac{1}{n}\sum_{i=1}^n A_i}  \mathbb{E}\left[  Y_i^{(1)}\mid\mathbf{A}\right]  - \frac{\frac{1}{n} \sum_{i=1}^n (1-A_i)}{\frac{1}{n} \sum_{i=1}^n (1-A_i)} \mathbb{E}\left[  Y_i^{(0)}\mid\mathbf{A}\right]   \right] \\
      &=\mathbb{E}\left[ \frac{ \frac{1}{n} \sum_{i=1}^n A_i }{\frac{1}{n}\sum_{i=1}^n A_i}  \mathbb{E}\left[  Y_i^{(1)}\right]  - \frac{\frac{1}{n} \sum_{i=1}^n (1-A_i) }{\frac{1}{n}\sum_{i=1}^n (1-A_i)} \mathbb{E}\left[  Y_i^{(0)}\right]   \right] && \text{$\{Y_i^{(1)}, Y_i^{(0)} \} \indep A_i $} \\
      &= \mathbb{E}\left[ \mathds{1}_{\sum_{i=1}^n A_i > 0}  \mathbb{E}\left[  Y_i^{(1)}\right] -  \mathds{1}_{\sum_{i=1}^n 1-A_i > 0} \mathbb{E}\left[  Y_i^{(0)}\right]   \right] \\
      &=   \mathbb{E}\left[  Y_i^{(1)}\right] \mathbb{E}\left[ \mathds{1}_{\sum_{i=1}^n A_i > 0} \right]  - \mathbb{E}\left[ \mathds{1}_{\sum_{i=1}^n 1-A_i > 0}    \right] \mathbb{E}\left[  Y_i^{(0)}\right] \\
      &= \left(1- (1-\pi)^n \right) \mathbb{E}\left[  Y_i^{(1)}\right] - \left(1- \pi^n \right) \mathbb{E}\left[  Y_i^{(0)}\right] \\
      &=  \mathbb{E}\left[Y_i^{(1)}   -Y_i^{(0)}  \right] - (1-\pi)^n\mathbb{E}\left[  Y_i^{(1)}\right]  + \pi^n\mathbb{E}\left[  Y_i^{(0)}\right] \\
      &= \tau -  (1-\pi)^n\mathbb{E}\left[  Y_i^{(1)}\right]  + \pi^n\mathbb{E}\left[  Y_i^{(0)}\right],
\end{align*}

where the second row uses linearity of expectation and the conditioning on $\mathbf{A}$. To summarize, the difference-in-means has a finite sample bias,

\begin{align*}
        \mathbb{E}\left[    \hat{\tau}_{\text{\tiny DM},n}\right] - \tau &=    \pi^n\mathbb{E}\left[  Y_i^{(0)}\right] - (1-\pi)^n\mathbb{E}\left[  Y_i^{(1)}\right].
\end{align*}

\textbf{Variance} \\

Using the law of total variance, and conditioning on the treatment assignment vector $\mathbf{A}$, one has
\begin{align*}
    \operatorname{Var}\left[ \hat{\tau}_{\text{\tiny DM}} \right] &=    \operatorname{Var}\left[\mathbb{E}\left[   \hat{\tau}_{\text{\tiny DM}} \mid \mathbf{A}\right] \right]  +   \mathbb{E}\left[ \operatorname{Var}\left[ \hat{\tau}_{\text{\tiny DM}} \mid \mathbf{A} \right] \right].
\end{align*}

Recall from derivations about the bias that, 

\begin{align*}
     \mathbb{E}\left[  \hat{\tau}_{\text{\tiny DM}} \mid \mathbf{A}\right] &=  \mathds{1}_{\sum_{i=1}^n A_i > 0}  \mathbb{E}\left[  Y_i^{(1)}\right] -  \mathds{1}_{\sum_{i=1}^n 1-A_i > 0} \mathbb{E}\left[  Y_i^{(0)}\right].
\end{align*}

Note that if the number of treated was fixed, we would have $\mathbb{E}\left[  \hat{\tau}_{\text{\tiny DM}} \mid \mathbf{A}\right]  = \tau$, and therefore, $\operatorname{Var}\left[\mathbb{E}\left[   \hat{\tau}_{\text{\tiny DM}} \mid \mathbf{A}\right] \right] = 0$.

Here, one has,

\begin{align*}
  \operatorname{Var}\left[\mathbb{E}\left[   \hat{\tau}_{\text{\tiny DM}} \mid \mathbf{A}\right] \right] &=     \operatorname{Var}\left[  \mathds{1}_{\sum_{i=1}^n A_i > 0}  \mathbb{E}\left[  Y_i^{(1)}\right] -  \mathds{1}_{\sum_{i=1}^n 1-A_i > 0} \mathbb{E}\left[  Y_i^{(0)}\right]  \right]\\
  &=  \mathbb{E}\left[  Y_i^{(1)}\right]^2   \operatorname{Var}\left[\mathds{1}_{\sum_{i=1}^n A_i > 0}  \right] + \mathbb{E}\left[  Y_i^{(0)}\right]^2   \operatorname{Var}\left[\mathds{1}_{\sum_{i=1}^n 1-A_i > 0}  \right] \\
  &\qquad -2   \mathbb{E}\left[  Y_i^{(1)}\right]  \mathbb{E}\left[  Y_i^{(0)}\right] \operatorname{Cov}\left[\mathds{1}_{\sum_{i=1}^n A_i > 0} , \mathds{1}_{\sum_{i=1}^n 1-A_i > 0}  \right].
\end{align*}

Besides,

\begin{align*}
      \operatorname{Var}\left[\mathds{1}_{\sum_{i=1}^n A_i > 0}  \right]  &=  \mathbb{E}\left[\mathds{1}_{\sum_{i=1}^n A_i > 0} ^2 \right] -  \mathbb{E}\left[\mathds{1}_{\sum_{i=1}^n A_i > 0} \right]^2 \\
      &= (1-\pi)^n\left( 1- (1-\pi)^n\right),
\end{align*}

and similarly,

\begin{align*}
          \operatorname{Var}\left[\mathds{1}_{\sum_{i=1}^n 1-A_i > 0}  \right]  &=  \pi^n\left( 1- \pi^n\right).
\end{align*}

On the other hand,

\begin{align*}
    \operatorname{Cov}\left[\mathds{1}_{\sum_{i=1}^n A_i > 0} , \mathds{1}_{\sum_{i=1}^n 1-A_i > 0}  \right] &= \mathbb{E}\left[ \left(\mathds{1}_{\sum_{i=1}^n A_i > 0} - \left(1- (1-\pi)^n \right)  \right)  \left( \mathds{1}_{\sum_{i=1}^n 1-A_i > 0}- 1-\pi^n\right)  \right] \\
    &= \mathbb{E}\left[\mathds{1}_{\sum_{i=1}^n A_i > 0} \mathds{1}_{\sum_{i=1}^n 1-A_i > 0} \right] - \left(1- (1-\pi)^n \right)\left(1-\pi^n\right) \\
    &= 1-(1-\pi)^n-\pi^n - \left(1- \pi^n - (1-\pi)^n - \pi^n(1-\pi)^n \right)\\
    &=  \pi^n(1-\pi)^n,
\end{align*}

such that,

\begin{align*}
  \operatorname{Var}\left[\mathbb{E}\left[   \hat{\tau}_{\text{\tiny DM}} \mid \mathbf{A}\right] \right] &=   \mathbb{E}\left[  Y_i^{(1)}\right]^2  (1-\pi)^n\left( 1- (1-\pi)^n\right)+  \mathbb{E}\left[  Y_i^{(0)}\right]^2  \pi^n\left( 1- \pi^n\right)-2 \mathbb{E}\left[  Y_i^{(1)}\right]  \mathbb{E}\left[  Y_i^{(0)}\right] \pi^n(1-\pi)^n \\
  &=  \mathbb{E}\left[  Y_i^{(1)}\right]^2 (1-\pi)^n +  \mathbb{E}\left[  Y_i^{(0)}\right]^2  \pi^n - \left(\mathbb{E}\left[  Y_i^{(1)}\right] (1-\pi)^n  +  \mathbb{E}\left[  Y_i^{(0)}\right]  \pi^n \right)^2 \\
 & \leq \mathbb{E}\left[  Y_i^{(1)}\right]^2 (1-\pi)^n +  \mathbb{E}\left[  Y_i^{(0)}\right]^2  \pi^n \\
  & \leq \left( \mathbb{E}\left[  Y^{(1)}\right]^2 + \mathbb{E}\left[  Y^{(0)}\right]^2\right) \max(\pi, 1 - \pi)^n.
\end{align*}

Now,

\begin{align*}
    \operatorname{Var}\left[ \hat{\tau}_{\text{\tiny DM}} \mid \mathbf{A} \right]  &=  \operatorname{Var}\left[ \frac{1}{n}  \sum_{i=1}^n \left( \frac{A_i Y_i^{(1)}}{\hat \pi}  - \frac{(1-A_i)Y_i^{(0)}  }{1-\hat \pi} \right) \mid \mathbf{A} \right] \\
    &= \frac{1}{n}  \operatorname{Var}\left[   \frac{A_i Y_i^{(1)}}{\hat \pi}  - \frac{(1-A_i)Y_i^{(0)}  }{1-\hat \pi} \mid \mathbf{A} \right] && \text{iid}\\
    &=  \frac{1}{n}  \left(  \operatorname{Var}\left[ \frac{A_i Y_i^{(1)}}{\hat \pi}  \mid \mathbf{A}\right] + \operatorname{Var}\left[ \frac{(1-A_i)Y_i^{(0)}  }{1-\hat \pi}  \mid \mathbf{A}\right]  - 2 \operatorname{Cov}\left[ \frac{A_i Y_i^{(1)}}{\hat \pi} , \frac{(1-A_i)Y_i^{(0)}  }{1-\hat \pi} \mid \mathbf{A} \right]\right).
\end{align*}

Now, developing the covariance term, it is possible to show that,

\begin{align*}
    \operatorname{Cov}\left[ \frac{A_i Y_i^{(1)}}{\hat \pi} , \frac{(1-A_i)Y_i^{(0)}  }{1-\hat \pi} \mid \mathbf{A} \right] &= - \mathbb{E}\left[ \frac{(1-A_i)Y_i^{(0)}  }{1-\hat \pi} \mid \mathbf{A} \right] \mathbb{E}\left[ \frac{A_iY_i^{(1)}  }{\hat \pi} \mid \mathbf{A} \right] \\
    &=  -  \frac{(1-A_i)\mathbb{E}\left[Y_i^{(0)}  \mid \mathbf{A} \right] }{1-\hat \pi} \frac{A_i\mathbb{E}\left[ Y_i^{(1)}  \mid \mathbf{A} \right] }{\hat \pi} && \text{Linearity and conditioned on $\mathbf{A}$} \\
    &= 0. && \text{$A_i(1-A_i)=0$} 
\end{align*}

Now, also using linearity of expectation, and the fact that we conditioned on $\mathbf{A}$, one has

\begin{align*}
    \operatorname{Var}\left[ \hat{\tau}_{\text{\tiny DM}} \mid \mathbf{A} \right]  &=  \frac{1}{n}  \left(  \left( \frac{A_i}{\hat \pi}\right)^2 \operatorname{Var}\left[Y_i^{(1)} \mid \mathbf{A}\right]   +  \left( \frac{1-A_i}{1- \hat \pi}\right)^2 \operatorname{Var}\left[Y_i^{(0)} \mid \mathbf{A}\right]  \right) \\
     &=  \frac{1}{n}  \left(  \left( \frac{A_i}{\hat \pi}\right)^2 \operatorname{Var}\left[Y_i^{(1)}\right]   +  \left( \frac{1-A_i}{1- \hat \pi}\right)^2 \operatorname{Var}\left[Y_i^{(0)}\right]  \right), && \text{using $\{Y_i^{(1)}, Y_i^{(0)} \} \indep A_i $.} 
\end{align*}

Taking the expecation of the previous term leads to,

\begin{align*}
     \mathbb{E}\left[ \operatorname{Var}\left[ \hat{\tau}_{\text{\tiny DM}} \mid \mathbf{A} \right] \right]&= \mathbb{E}\left[ \frac{1}{n}  \left(  \left( \frac{A_i}{\hat \pi}\right)^2 \operatorname{Var}\left[Y_i^{(1)}\right]   +  \left( \frac{1-A_i}{1- \hat \pi}\right)^2 \operatorname{Var}\left[Y_i^{(0)}\right]  \right)\right] \\
     &=  \frac{1}{n} \left(\mathbb{E}\left[ \left( \frac{A_i}{\hat \pi}\right)^2\right]  \operatorname{Var}\left[Y_i^{(1)}\right] + \frac{1}{n} \mathbb{E}\left[ \left( \frac{1-A_i}{1-\hat \pi}\right)^2\right]  \operatorname{Var}\left[Y_i^{(0)}\right] \right), && \text{by linearity.}
\end{align*}

Note that,

\begin{align*}
     \mathbb{E}\left[ \left( \frac{A_i}{\hat \pi}\right)^2\right]  &=   \mathbb{E}\left[  \frac{A_i}{\left(\hat \pi\right)^2}\right] \\
     &= \frac{1}{n}\left(  \mathbb{E}\left[  \frac{A_1}{\hat \pi^2}\right] +  \mathbb{E}\left[  \frac{A_2}{\hat \pi^2}\right]  + \dots +  \mathbb{E}\left[  \frac{A_n}{\hat \pi^2}\right] \right) \\
     &= \mathbb{E}\left[  \frac{\hat \pi}{\hat \pi^2}\right] \\
     &= \mathbb{E}\left[  \frac{\mathds{1}_{\hat \pi > 0} }{\hat \pi}\right],
\end{align*}

so that

\begin{align*}
         \mathbb{E}\left[ \operatorname{Var}\left[ \hat{\tau}_{\text{\tiny DM}} \mid \mathbf{A} \right] \right] &=  \frac{1}{n} \left( \mathbb{E}\left[  \frac{\mathds{1}_{\hat \pi > 0} }{\hat \pi}\right] \operatorname{Var}\left[Y_i^{(1)}\right] +  \mathbb{E}\left[  \frac{\mathds{1}_{(1-\hat \pi) > 0} }{1-\hat \pi}\right]  \operatorname{Var}\left[Y_i^{(0)}\right] \right).
\end{align*}

Coming back to the law of total variance, one has,

\begin{align*}
     \operatorname{Var}\left[ \hat{\tau}_{\text{\tiny DM}} \right] &=    \operatorname{Var}\left[\mathbb{E}\left[   \hat{\tau}_{\text{\tiny DM}} \mid \mathbf{A}\right] \right]  +   \mathbb{E}\left[ \operatorname{Var}\left[ \hat{\tau}_{\text{\tiny DM}} \mid \mathbf{A} \right] \right] \\
     &= \mathbb{E}\left[  Y_i^{(1)}\right]^2 (1-\pi)^n +  \mathbb{E}\left[  Y_i^{(0)}\right]^2  \pi^n - \left(\mathbb{E}\left[  Y_i^{(1)}\right] (1-\pi)^n  +  \mathbb{E}\left[  Y_i^{(0)}\right]  \pi^n \right)^2  \\
     & \qquad + \frac{1}{n} \left( \mathbb{E}\left[  \frac{\mathds{1}_{\hat \pi > 0} }{\hat \pi}\right] \operatorname{Var}\left[Y_i^{(1)}\right] +  \mathbb{E}\left[  \frac{\mathds{1}_{(1-\hat \pi) > 0} }{1-\hat \pi}\right]  \operatorname{Var}\left[Y_i^{(0)}\right] \right)
\end{align*}

In particular, for any sample size,

\begin{align*}
      \operatorname{Var}\left[ \hat{\tau}_{\text{\tiny DM}} \right] &= \frac{1}{n} \left( \mathbb{E}\left[  \frac{\mathds{1}_{\hat \pi > 0} }{\hat \pi}\right] \operatorname{Var}\left[Y_i^{(1)}\right] +  \mathbb{E}\left[  \frac{\mathds{1}_{(1-\hat \pi) > 0} }{1-\hat \pi}\right]  \operatorname{Var}\left[Y_i^{(0)}\right] \right) + \mathcal{O}\left(\max(\pi, 1-\pi)^n \right), 
\end{align*}

and more particularly,
\begin{align*}
     \lim_{n\to\infty}  n \operatorname{Var}\left[ \hat{\tau}_{\text{\tiny DM}} \right] &= \frac{  \operatorname{Var}\left[ Y^{(1)}\right] }{\pi}+  \frac{  \operatorname{Var}\left[ Y^{(0)}\right] }{1-\pi}:=  V_{ \text{\tiny DM}, \infty}.
\end{align*}

\end{proof}

\subsection{Variance inequality between a Horvitz-Thomson and difference-in-means}\label{proof:variance-inequality}

Recall that 
\begin{align*}
V_{\text{\tiny DM},n}(x)
 =  \frac{\mathds{1}_{Z_n(x) >0}}{Z_n(x)} \operatorname{Var} \left[    \sum_{i=1}^{n} \mathds{1}_{X_i=x}   \left( \frac{A_i Y_i^{(1)}}{\hat \pi_n(x)} - \frac{(1-A_i)Y_i^{(0)}}{1-\hat \pi_n(x)} \right) \mid \mathbf{X}_{n+m} \right].
\end{align*}

\begin{lemma}
\label{lem_inequality_vdm_vht}
We have, for all $x$, 
\begin{align*}
& ~~ \mathbb{E} \left[ V_{\text{\tiny DM},n}(x) \right]  \leq V_{ \text{\tiny HT}}(x) - \alpha(x)^2 + \beta(x) n^{-1/4} + o(n^{-1/4}),
\end{align*}
with 
\begin{align*}
    \alpha(x)=  \sqrt{\frac{1-\pi(x)}{\pi(x)}} \mathbb{E}_{ \text{\tiny R}}[Y^{(1)} | X=x] + \sqrt{\frac{\pi(x)}{1-\pi(x)}} \mathbb{E}_{\text{\tiny  R}}[Y^{(0)} | X=x],
\end{align*}
and some $\beta(x)$ independent of $n$. 
\end{lemma}

\begin{proof}

Following the same proof as that of Lemma~\ref{lemma:DM-bias-and-variance}, noticing that  
\begin{align*}
 V_{\text{\tiny DM},n}(x)
 & = \frac{1}{Z_n(x)}  \operatorname{Var} \left[    \sum_{i=1}^{n} \mathds{1}_{X_i=x}   \left( \frac{A_i Y_i^{(1)}}{\hat \pi_n(x)} - \frac{(1-A_i)Y_i^{(0)}}{1-\hat \pi_n(x)} \right) \mid \mathbf{X}_{n+m} \right] \\
 &= \frac{1}{Z_n(x)}  \operatorname{Var} \left[ Z_n(X) \hat{\tau}_{\text{\tiny DM}}(X) | X=x, Z_n(x) \right]\\
 & =  Z_n(x) \operatorname{Var} \left[ \hat{\tau}_{\text{\tiny DM}}(X) | X=x, Z_n(x) \right],
\end{align*}
we have
\begin{align*}
V_{\text{\tiny DM},n}(x)
 &=  \mathbb{E}\left[  \frac{\mathds{1}_{\hat \pi(x) > 0} }{\hat \pi(x)} | Z_n(x)\right] \operatorname{Var}\left[Y^{(1)} | X=x\right] +  \mathbb{E}\left[  \frac{\mathds{1}_{(1-\hat \pi(x)) > 0} }{1-\hat \pi(x)} | Z_n(x) \right]  \operatorname{Var}\left[Y^{(0)} | X=x\right]   \\
 & \qquad + Z_n(x) \mathbb{E}\left[  Y^{(1)} | X=x\right]^2 (1-\pi(x))^{Z_n(x)} +  Z_n(x) \mathbb{E}\left[  Y^{(0)} | X=x\right]^2  \pi(x)^{Z_n(x)} \\
 & \qquad - Z_n(x) \left(\mathbb{E}\left[  Y^{(1)} | X=x\right] (1-\pi(x))^{Z_n(x)}  +  \mathbb{E}\left[  Y^{(0)} | X=x\right]  \pi(x)^{Z_n(x)} \right)^2.
\end{align*}

According to Lemma~\ref{lemma:ineq-binomial-pi-hat}, letting $ \alpha = 1/4$ and $C_{1/4, \pi(x)}=1+2\left(\frac{32}{\pi(x)^2}\right)^{4}$, 
\begin{align*}
    & ~~ V_{\text{\tiny DM},n}(x) \\
    & \le Z_n(x) \mathbb{E}\left[  Y^{(1)} | X = x\right]^2 (1-\pi(x))^{Z_n(x)} +  Z_n(x)\mathbb{E}\left[  Y^{(0)} | X = x \right]^2  \pi(x)^{Z_n(x)} \\
    & \qquad - Z_n(x) \left(\mathbb{E}\left[  Y^{(1)} | X = x\right] (1-\pi(x))^{Z_n(x)}  +  \mathbb{E}\left[  Y^{(0)} | X = x \right]  \pi(x)^{Z_n(x)} \right)^2  \\
      & \qquad +   \left( \frac{1+C_{1/4, \pi(x)} Z_n(x)^{-\frac{1}{4}}}{\pi(x)} \right)\operatorname{Var}\left[Y^{(1)} | X = x\right] +  \left(\frac{1+C_{1/4, 1-\pi(x)} Z_n^{-\frac{1}{4}}}{1-\pi(x)} \right)\operatorname{Var}\left[Y^{(0)} | X = x\right] \\
      & \leq V_{\text{\tiny DM}, \infty}(x) + Z_n(x)^{-1/4} \left( \frac{C_{1/4, \pi(x)}}{\pi(x) } \operatorname{Var}\left[Y^{(1)} | X = x\right] + \frac{C_{1/4, 1-\pi(x)}}{1-\pi(x) } \operatorname{Var}\left[Y^{(1)} | X = x\right] \right) \\ 
      & \qquad + Z_n(x) \left(\mathbb{E}\left[  Y^{(1)} | X = x\right]^2 (1-\pi(x))^{Z_n(x)} +  \mathbb{E}\left[  Y^{(0)} | X = x \right]^2  \pi(x)^{Z_n(x)} \right)\\
    & \qquad - Z_n(x) \left(\mathbb{E}\left[  Y^{(1)} | X = x\right] (1-\pi(x))^{Z_n(x)}  +  \mathbb{E}\left[  Y^{(0)} | X = x \right]  \pi(x)^{Z_n(x)} \right)^2,
\end{align*}
where
\begin{align*}
V_{ \text{\tiny DM}, \infty}(x) =  \frac{  \operatorname{Var}\left[ Y^{(1)} | X=x\right] }{\pi(x)}+  \frac{  \operatorname{Var}\left[ Y^{(0)} | X=x\right] }{1-\pi(x)}.
\end{align*}
Besides, we have
\begin{align*} 
V_{ \text{\tiny HT}}(x)  & =  \mathbb{E}\left[ \frac{\left( Y^{(1)} \right)^2}{\pi(x)} | X=x\right]  + \mathbb{E}\left[ \frac{\left( Y^{(0)} \right)^2}{1-\pi(x)} | X=x \right]  - \tau(x)^2 \\
& =  \frac{ \operatorname{Var}\left[Y^{(1)}| X = x\right] }{\pi(x)} + \frac{ \operatorname{Var}\left[Y^{(0)}| X = x \right] }{1-\pi(x)} + \frac{1}{\pi(x)} \left( \mathbb{E}\left[  Y^{(1)}  | X=x\right] \right)^2  + \frac{1}{1 - \pi(x)} \left( \mathbb{E}\left[  Y^{(0)}  | X=x\right] \right)^2    - \tau(x)^2 \\
& = V_{ \text{\tiny DM}, \infty}(x)  + \frac{1}{\pi(x)} \left( \mathbb{E}\left[  Y^{(1)}  | X=x\right] \right)^2  + \frac{1}{1 - \pi(x)} \left( \mathbb{E}\left[  Y^{(0)}  | X=x\right] \right)^2    - \tau(x)^2.
\end{align*}
Noting that,
\begin{equation*}
    \tau(x)^2 = \left( \mathbb{E}\left[  Y^{(1)} -  Y^{(0)} | X=x \right]\right)^2 =  \mathbb{E}\left[  Y^{(1)} | X=x \right]^2 + \mathbb{E}\left[  Y^{(0)} | X=x \right]^2 - 2\mathbb{E}\left[  Y^{(1)} | X=x \right] \mathbb{E}\left[  Y^{(0)} | X=x \right],
\end{equation*}
allows us to obtain,
\begin{align*}
    V_{ \text{\tiny HT}}(x) &= V_{ \text{\tiny DM}, \infty}(x)  - \left(1-\frac{1}{\pi(x)}\right)\mathbb{E}\left[  Y^{(1)} \right]^2- \left(1-\frac{1}{1-\pi(x)}\right)\mathbb{E}\left[  Y^{(0)} | X=x\right]^2 + 2\mathbb{E}\left[  Y^{(1)} | X=x\right] \mathbb{E}\left[  Y^{(0)} | X=x \right]\\
     &=  V_{ \text{\tiny DM}, \infty}(x) + \left( \sqrt{\frac{1-\pi(x)}{\pi(x)}} \mathbb{E}_{ \text{\tiny R}}[Y^{(1)} | X=x] + \sqrt{\frac{\pi(x)}{1-\pi(x)}} \mathbb{E}_{\text{\tiny  R}}[Y^{(0)} | X=x]\right)^2.
\end{align*}
Letting 
\begin{align*}
    \alpha(x) = \left( \sqrt{\frac{1-\pi(x)}{\pi(x)}} \mathbb{E}_{ \text{\tiny R}}[Y^{(1)} | X=x] + \sqrt{\frac{\pi(x)}{1-\pi(x)}} \mathbb{E}_{\text{\tiny  R}}[Y^{(0)} | X=x]\right)^2,
\end{align*}
we have
\begin{align*}
    & ~~ V_{\text{\tiny DM},n}(x) \\
    & \leq V_{ \text{\tiny HT}}(x) - \alpha(x) + Z_n(x)^{-1/4} \left( \frac{C_{1/4, \pi(x)}}{\pi(x) } \operatorname{Var}\left[Y^{(1)} | X = x\right] + \frac{C_{1/4, 1-\pi(x)}}{1-\pi(x) } \operatorname{Var}\left[Y^{(1)} | X = x\right] \right) \\ 
      & \qquad + Z_n(x) \left(\mathbb{E}\left[  Y^{(1)} | X = x\right]^2 (1-\pi(x))^{Z_n(x)} +  \mathbb{E}\left[  Y^{(0)} | X = x \right]^2  \pi(x)^{Z_n(x)} \right). 
\end{align*}
Taking the expectation on both sides with respect to $Z_n(x)$, we have
\begin{align*}
& ~~ \mathbb{E} \left[ V_{\text{\tiny DM},n}(x) \mathds{1}_{Z_n(x) >0}\right] \\   
& \leq V_{ \text{\tiny HT}}(x) - \alpha(x) + \mathbb{E}\left[ \frac{\mathds{1}_{Z_n(x) >0}}{Z_n(x)^{1/4}} \right] \left( \frac{C_{1/4, \pi(x)}}{\pi(x) } \operatorname{Var}\left[Y^{(1)} | X = x\right] + \frac{C_{1/4, 1-\pi(x)}}{1-\pi(x) } \operatorname{Var}\left[Y^{(1)} | X = x\right] \right) \\ 
      & \qquad + n \left(\mathbb{E}\left[  Y^{(1)} | X = x\right]^2 +  \mathbb{E}\left[  Y^{(0)} | X = x \right]^2 \right) \mathbb{E} \left[ \tilde{\pi}(x)^{Z_n(x)} \mathds{1}_{Z_n(x)>0} \right].
\end{align*}
where $\tilde{\pi}(x) = \max(\pi(x), 1 - \pi(x)).$ Simple calculations show that 
\begin{align*}
\left(\mathbb{E}\left[ \frac{\mathds{1}_{Z_n(x) >0}}{Z_n(x)^{1/4}} \right]\right)^4 & \leq     \mathbb{E}\left[ \frac{\mathds{1}_{Z_n(x) >0}}{Z_n(x)} \right]  \\
& \leq \frac{2}{(n+1) p_\text{\tiny R}\left( x \right)}
\end{align*}
according to Jensen inequality and \cite{Arnould2021Analyzing} (Lemma S5, Supplementary Material, page 27). Besides 
\begin{align*}
\mathbb{E} \left[ \tilde{\pi}(x)^{Z_n(x)} \mathds{1}_{Z_n(x)>0} \right] & \leq \prod_{i=1}^n \mathbb{E} \left[ \tilde{\pi}(x)^{\mathds{1}_{X_i=x}}\right] \\
& \leq  \left( 1 - p_R(x) + \tilde{\pi}(x) p_R(x) \right)^n.
\end{align*}
Hence, 
\begin{align*}
& ~~ \mathbb{E} \left[ V_{\text{\tiny DM},n}(x) \mathds{1}_{Z_n(x) >0}\right] \\  & \leq V_{ \text{\tiny HT}}(x) - \alpha(x) + \left( \frac{2}{(n+1) p_\text{\tiny R}\left( x \right)}\right)^{1/4} \left( \frac{C_{1/4, \pi(x)}}{\pi(x) } \operatorname{Var}\left[Y^{(1)} | X = x\right] + \frac{C_{1/4, 1-\pi(x)}}{1-\pi(x) } \operatorname{Var}\left[Y^{(1)} | X = x\right] \right) \\ 
      & \qquad + n \left(\mathbb{E}\left[  Y^{(1)} | X = x\right]^2 +  \mathbb{E}\left[  Y^{(0)} | X = x \right]^2 \right)  \left( 1 - p_R(x) + \tilde{\pi}(x) p_R(x) \right)^n\\
& \leq V_{ \text{\tiny HT}}(x) - \alpha(x) + \beta(x) n^{-1/4} + o(n^{-1/4}).
\end{align*}

\end{proof}

In this work we use an inequality to compare the variance of the Horvitz-Thomson with the variance of the difference-in-means under a Bernoulli design. We propose two inequalities, one for the finite sample and one for the asymptotic variance. The result in finite sample depends on another equality on Binomial law, and in particular $\hat \pi$, that we detail in Lemma~\ref{lemma:ineq-binomial-pi-hat}.

\begin{lemma}[Inequality on $\hat \pi$]\label{lemma:ineq-binomial-pi-hat}
Consider a Bernoulli trial (Definition~\ref{def:bernoulli-trial}) and the estimated propensity score $\hat \pi$ defined as,

\begin{equation*}
    \hat \pi = \frac{\sum_{i=1}^n A_i}{n}.
\end{equation*}

Then, for all $n\ge1$ and for all $\alpha \in (0, \frac{1}{2})$,

\begin{equation*}
    \mathbb{E}\left[\frac{\mathds{1}_{\hat{\pi}>0}}{\hat{\pi}}\right] \leq \frac{1+C_{\alpha, \pi} n^{-\alpha}}{\pi},
\end{equation*}

where $C_{\alpha, \pi}=1+2\left(\frac{16}{\pi^2(1-2 \alpha)}\right)^{\frac{2}{1-2 \alpha}}$.

\end{lemma}

\begin{proof}
Let $\varepsilon > 0$. (and later in the proof, we will more precisely posit $\varepsilon=\frac{\pi}{4} n^{-\alpha}$ with $\alpha \in (0, \frac{1}{2})$)

The law of total expectation leads to,

\begin{align*}
     \mathbb{E}\left[\frac{\mathds{1}_{\hat{\pi}>0}}{\hat{\pi}}\right]  &=  \mathbb{E}\left[\frac{\mathds{1}_{\hat{\pi}>0}}{\hat{\pi}} \mathds{1}_{|\hat \pi - \pi|<\varepsilon}\right]  +  \mathbb{E}\left[\frac{\mathds{1}_{\hat{\pi}>0}}{\hat{\pi}} \mathds{1}_{|\hat \pi - \pi|\ge\varepsilon}\right].
\end{align*}

For the first term,

\begin{align*}
   \mathbb{E}\left[\frac{\mathds{1}_{\hat{\pi}>0}}{\hat{\pi}} \mathds{1}_{|\hat \pi - \pi|<\varepsilon}\right] &\le \frac{1}{\pi-\varepsilon}  \mathbb{E}\left[\mathds{1}_{\hat{\pi}>0} \mathds{1}_{|\hat \pi - \pi|<\varepsilon}\right] \\
   &\le \frac{1}{\pi-\varepsilon},
\end{align*}

and for the second term,

\begin{align*}
      \mathbb{E}\left[\frac{\mathds{1}_{\hat{\pi}>0}}{\hat{\pi}} \mathds{1}_{|\hat \pi - \pi|\ge\varepsilon}\right] &\le n   \mathbb{E}\left[\mathds{1}_{\hat{\pi}>0}\mathds{1}_{|\hat \pi - \pi|\ge\varepsilon}\right] \\
      &\le n \mathbb{P}\left( |\hat \pi - \pi| \ge \varepsilon  \right) \\
      &\le 2 n e^{-2 \varepsilon^2 n}.
\end{align*}

The last row is obtained through Chernoff's inequality in a similar manned as in the proof for the semi-oracle (see \eqref{eq_proof_Chernoff}). As a consequence, and gathering the two previous inequalities,

\begin{align*}
    \mathbb{E}\left[\frac{\mathds{1}_{\hat{\pi}>0}}{\hat{\pi}}\right] &\leq \frac{1}{\pi-\varepsilon}+2 n e^{-2 \varepsilon^2 n} \\
    &= \frac{1}{\pi}\frac{1}{1-\frac{\varepsilon}{\pi}}+2 n e^{-2 \varepsilon^2 n}.
\end{align*}

One can show using function analysis, that, for all $0\le x < \frac{1}{2}$, we have

\begin{align*}
    \frac{1}{1-x} \le 1 + \frac{x}{1-2x}.
\end{align*}

Then, as soon as $\varepsilon$ is small enough, then $\frac{\varepsilon}{\pi} < \frac{1}{2}$, so that,
\begin{align*}
    \mathbb{E}\left[\frac{\mathds{1}_{\hat{\pi}>0}}{\hat{\pi}}\right] & \leq \frac{1}{\pi}\frac{1}{1-\frac{\varepsilon}{\pi}}+2 n e^{-2 \varepsilon^2 n} \\
    &\leq \frac{1}{\pi} \left( 1 +\frac{\frac{\varepsilon}{\pi}}{1-2\frac{\varepsilon}{\pi}}\right) +2 n e^{-2 \varepsilon^2 n}.
\end{align*}

Letting $\varepsilon=\frac{\pi}{4} n^{-\alpha}$ with $\alpha \in (0, \frac{1}{2})$, we have

\begin{align*}
\mathbb{E}\left[\frac{\mathds{1}_{\hat{\pi}>0}}{\hat{\pi}}\right] & \leq \frac{1}{\pi}   + \frac{1}{4\pi} \frac{ n^{-\alpha}}{1-\frac{n^{-\alpha}}{2}} + 2 n e^{-\frac{\pi^2}{8}n^{1-2\alpha}} 
\end{align*}

Now, using the fact that 

\begin{align*}
    \forall x \ge 1, \forall \alpha \in (0,\frac{1}{2}), \quad x^2 e^{-\frac{\pi^2}{8} x^{1-2 \alpha}} &\le  \underbrace{\left( \frac{16}{\pi^2(1-2\alpha)} \right)^{\frac{2}{1-2\alpha}}}_{C_{\alpha,\pi}},
\end{align*}
allows to have 
\begin{align*}
\mathbb{E}\left[\frac{\mathds{1}_{\hat{\pi}>0}}{\hat{\pi}}\right] & \leq \frac{1}{\pi}  + \frac{1}{4\pi} \frac{ n^{-\alpha}}{1-\frac{n^{-\alpha}}{2}}  + 2\frac{C_{\alpha,\pi}}{n} \\
&\leq \frac{1}{\pi} + \frac{n^{-\alpha}}{\pi} + 2\frac{C_{\alpha,\pi}}{\pi n^{\alpha}} \\
&= \frac{1 + n^{-\alpha} (1 + 2C_{\alpha,\pi})}{\pi}.
\end{align*}

\end{proof}

\begin{lemma}[Variance inequality]\label{lemma:variance-inequality}
Considering the Horvitz-Thomson estimator (Definition~\ref{def:HT}) and the difference-in-means estimator (Definition~\ref{def:difference-in-means}), with an internally valid randomized controlled trial of size $n$ (Assumption~\ref{a:trial-internal-validity}), then asymptotic variance of the  difference-in-means is always smaller or equal than the Horvitz-Thomson, such as
\begin{equation*}
      V_{ \text{\tiny DM}, \infty}= V_{ \text{\tiny HT}} - \left( \sqrt{\frac{1-\pi}{\pi}} \mathbb{E}_{ \text{\tiny R}}[Y^{(1)}] + \sqrt{\frac{\pi}{1-\pi}} \mathbb{E}_{\text{\tiny  R}}[Y^{(0)}]\right)^2 \le V_{ \text{\tiny HT}}.
\end{equation*}

In addition, and using the previous inequality, Lemma~\ref{lemma:DM-bias-and-variance} and Lemma~\ref{lemma:ineq-binomial-pi-hat}, one can bound the finite sample difference-in-means's variance:

\begin{align*}
     \operatorname{Var}\left[ \hat{\tau}_{\text{\tiny DM},n} \right] & = \frac{1}{n}\left( \frac{ \operatorname{Var}\left[Y^{(1)}\right] }{\pi}  + \frac{ \operatorname{Var}\left[Y^{(0)}\right] }{1-\pi} \right) + \mathcal{O}\left( n^{-3/2}\right) \\
     &\le  V_{ \text{\tiny HT}}+ \mathcal{O}\left( n^{-3/2}\right) .
\end{align*}

\end{lemma}


\begin{proof}

\textbf{Asymptotic inequality}\\

Recall that,

\begin{equation*}
 V_{ \text{\tiny HT}} =  \mathbb{E}\left[ \frac{\left( Y^{(1)} \right)^2}{\pi} \right]  + \mathbb{E}\left[ \frac{\left( Y^{(0)} \right)^2}{1-\pi} \right]  - \tau^2.
\end{equation*}

Noting that,

\begin{equation*}
    \tau^2 = \left( \mathbb{E}\left[  Y^{(1)} -  Y^{(0)} \right]\right)^2 =  \mathbb{E}\left[  Y^{(1)} \right]^2 + \mathbb{E}\left[  Y^{(0)} \right]^2 - 2\mathbb{E}\left[  Y^{(1)} \right] \mathbb{E}\left[  Y^{(0)} \right],
\end{equation*}

and that for any $a \in  \{0, 1\}$, 

\begin{equation*}
      \operatorname{Var}\left[Y^{(a)}\right] =  \mathbb{E}\left[ \left( Y^{(a)} \right)^2\right] - \mathbb{E}\left[  Y^{(a)} \right]^2,
\end{equation*}

allows to obtain,

\begin{align*}
    V_{ \text{\tiny HT}} &= \frac{ \operatorname{Var}\left[Y^{(1)}\right] }{\pi} + \frac{ \operatorname{Var}\left[Y^{(0)}\right] }{1-\pi} - (1-\frac{1}{\pi})\mathbb{E}\left[  Y^{(1)} \right]^2- (1-\frac{1}{1-\pi})\mathbb{E}\left[  Y^{(0)} \right]^2 + 2\mathbb{E}\left[  Y^{(1)} \right] \mathbb{E}\left[  Y^{(0)} \right]\\
     &=  V_{ \text{\tiny DM}, \infty} + \left( \sqrt{\frac{1-\pi}{\pi}} \mathbb{E}_{ \text{\tiny R}}[Y^{(1)}] + \sqrt{\frac{\pi}{1-\pi}} \mathbb{E}_{\text{\tiny  R}}[Y^{(0)}]\right)^2.
\end{align*}

\textbf{Finite sample inequality}\\
Recall the finite sample variance of the difference-in-means from Lemma~\ref{lemma:DM-bias-and-variance}, and using the inequality from Lemma~\ref{lemma:ineq-binomial-pi-hat},

\begin{align*}
     \operatorname{Var}\left[ \hat{\tau}_{\text{\tiny DM},n} \right] &= \mathbb{E}\left[  Y_i^{(1)}\right]^2 (1-\pi)^n +  \mathbb{E}\left[  Y_i^{(0)}\right]^2  \pi^n - \left(\mathbb{E}\left[  Y_i^{(1)}\right] (1-\pi)^n  +  \mathbb{E}\left[  Y_i^{(0)}\right]  \pi^n \right)^2  \\
     & \qquad + \frac{1}{n} \left( \mathbb{E}\left[  \frac{\mathds{1}_{\hat \pi > 0} }{\hat \pi}\right] \operatorname{Var}\left[Y_i^{(1)}\right] +  \mathbb{E}\left[  \frac{\mathds{1}_{(1-\hat \pi) > 0} }{1-\hat \pi}\right]  \operatorname{Var}\left[Y_i^{(0)}\right] \right) \\
     &\le \mathbb{E}\left[  Y_i^{(1)}\right]^2 (1-\pi)^n +  \mathbb{E}\left[  Y_i^{(0)}\right]^2  \pi^n - \left(\mathbb{E}\left[  Y_i^{(1)}\right] (1-\pi)^n  +  \mathbb{E}\left[  Y_i^{(0)}\right]  \pi^n \right)^2  \\
      & \qquad + \frac{1}{n} \left( \frac{1+C_{1/4, \pi} n^{-\frac{1}{4}}}{\pi} \operatorname{Var}\left[Y_i^{(1)}\right] +  \frac{1+C_{1/4, 1-\pi} n^{-\frac{1}{4}}}{1-\pi}\operatorname{Var}\left[Y_i^{(0)}\right] \right),
\end{align*}

where Lemma~\ref{lemma:ineq-binomial-pi-hat} is applied with $\alpha = 1/4$ and we recall that $C_{1/4, \pi}=1+2\left(\frac{32}{\pi^2}\right)^{4}$. Note that, at this stage, it is possible to write that,

\begin{align}\label{eq:short-writing-of-var-DM}
       \operatorname{Var}\left[ \hat{\tau}_{\text{\tiny DM},n} \right] &= \frac{1}{n}\left( \frac{ \operatorname{Var}\left[Y^{(1)}\right] }{\pi}  + \frac{ \operatorname{Var}\left[Y^{(0)}\right] }{1-\pi} \right) + \mathcal{O}\left( n^{-3/2}\right).
\end{align}

But the overall goal here is to compare $\operatorname{Var}\left[ \hat{\tau}_{\text{\tiny DM},n} \right]$ with  $\operatorname{Var}\left[ \hat{\tau}_{\text{\tiny HT},n} \right]$.

\begin{align*}
       \operatorname{Var}\left[ \hat{\tau}_{\text{\tiny DM},n} \right] &\le \operatorname{Var}\left[ \hat{\tau}_{\text{\tiny HT},n} \right] - \frac{1}{n}\left(\sqrt{\frac{1-\pi}{\pi}} \mathbb{E}_{ \text{\tiny R}}[Y^{(1)}] + \sqrt{\frac{\pi}{1-\pi}} \mathbb{E}_{\text{\tiny  R}}[Y^{(0)}] \right) \\ 
       &\qquad + \frac{1}{n} \left( \frac{C_{1/4, \pi} n^{-\frac{1}{4}}}{\pi} \operatorname{Var}\left[Y_i^{(1)}\right] +  \frac{C_{1/4, 1-\pi} n^{-\frac{1}{4}}}{1-\pi}\operatorname{Var}\left[Y_i^{(0)}\right] \right) && \text{$\ge 0$} \\
       &\qquad \qquad + \mathbb{E}\left[  Y_i^{(1)}\right]^2 (1-\pi)^n +  \mathbb{E}\left[  Y_i^{(0)}\right]^2  \pi^n - \left(\mathbb{E}\left[  Y_i^{(1)}\right] (1-\pi)^n  +  \mathbb{E}\left[  Y_i^{(0)}\right]  \pi^n \right)^2 
\end{align*}

\end{proof}

\subsection{Post-stratification estimator}
The post-stratified estimator (see Definition~\ref{def:post-stratification-estimator}) is an estimator of the average treatment effect from a RCT sample. The principle is to divide the RCT sample into strata, to compute the difference-in-means per strata, and then to average the estimand on each strata, weighting by the strata size. Indeed, the post-stratification estimator introduced in Definition~\ref{def:post-stratification-estimator} can be re-written as follows.

\begin{align*}
    \hat \tau_{\text{\tiny PS},n}=   \sum_{x \in \mathds{X}} \frac{n_{x,1} + n_{x,0}}{n}  \left( \frac{1}{n_{x,1}} \sum_{A_{i}=1, X_i = x} Y_{i}-\frac{1}{n_{x,0}} \sum_{A_{i}=0, X_i = x }Y_{i} \right), \quad \text{where } n_{x,a}= \sum_{i=1}^n \mathbbm{1}_{X_i = x}\mathbbm{1}_{A_i = a}.
\end{align*}

Therefore, the post-stratification estimator can be understood as a weighted estimate of each strata level difference-in-means estimates,

\begin{align*}
     \hat \tau_{\text{\tiny PS},n}=  \sum_{x \in \mathds{X}} \frac{n_x}{n} \hat \tau_{\text{\tiny DM},n_x}, \quad \text{where } n_{x}= \sum_{i=1}^n \mathbbm{1}_{X_i = x}.
\end{align*}

\begin{proof}
Recalling the definition of $\hat \pi_n(x)$ (Definition~\ref{def:procedure-for-densities-and-pi}) and denoting $n_{x,a}= \sum_{i=1}^n \mathbbm{1}_{X_i = x}\mathbbm{1}_{A_i = a}$
\begin{align*}
      \hat \tau_{\text{\tiny PS},n} &= \frac{1}{n} \sum_{i=1}^n \frac{A_iY_i}{\hat \pi_n(x)} -  \frac{(1-A_i)Y_i}{1-\hat \pi_n(x)}\\
      &= \frac{1}{n} \sum_{i=1}^n \frac{A_iY_i}{n_{x,1}/n_x} -  \frac{(1-A_i)Y_i}{n_{x,0}/n_x} && \text{Definition~\ref{def:procedure-for-densities-and-pi}}\\
      &=  \sum_{x \in \mathds{X}}    \frac{1}{n}  \sum_{i=1}^n \mathds{1}_{X_i=x}\frac{A_iY_i}{n_{x,1}/n_x} -  \frac{(1-A_i)Y_i}{n_{x,0}/n_x} && \text{Categorical covariates}\\
      &= \sum_{x \in \mathds{X}}    \frac{n_x}{n}  \sum_{i=1}^n \mathds{1}_{X_i=x}\frac{A_iY_i}{n_{x,1}} -  \frac{(1-A_i)Y_i}{n_{x,0}} && \text{Re-arranging $n_x$} \\
      &=  \sum_{x \in \mathds{X}}    \frac{n_x}{n}  \hat \tau_{\text{\tiny DM},n_x}.
\end{align*}

\end{proof}

The post-stratified estimator is extensively detailed in \cite{miratrix2013adjusting}, but largely focused on inference on a finite population (except in their Section 5). In particular the variance of the post-stratified estimator under a Bernoulli or a completely randomized design is given in \cite{miratrix2013adjusting} (see their Equation (16)). \cite{Imai2008Misunderstanding} also present derivation to compare the variance of a difference-in-means with a post-stratified estimator, quantifying the gain in precision (see Appendix A).

\section{(Non-exhaustive) Review of the different IPSW versions in the literature}\label{appendix:different-version-IPSW}
 
Within the generalization literature, the IPSW can be found under slightly different forms, such as with estimated $\pi$ or not, or with or without normalization. Here, and to help the reader navigates, we reference some of the different formulas found in the literature and in implementations. 
 
 \begin{center}
    \begin{table}[H]
        \begin{tabular}{|l|l|l|}
        \hline
        \textbf{Reference} & \textbf{IPSW formula} & \textbf{Comments} \\ \hline
         \cite{Huang2022Sensitivity}    &                  $\frac{1}{n} \sum_{i \in \text{Trial}}\hat w_{n,m}(X_i) \left(\frac{ Y_i A_i}{\hat \pi_n} - \frac{ Y_i (1-A_i)}{1-\hat \pi_n} \right)$        &  $\pi$ estimated once    $\hat \pi_n = \sum_{i=1}^n A_i / n$      \\ \hline
                  \cite{josey2021transporting}    &                  $\frac{1}{n} \sum_{i \in \text{Trial}}\hat w_{n,m}(X_i) \left(\frac{ Y_i A_i}{\hat \pi_n} - \frac{ Y_i (1-A_i)}{1-\hat \pi_n} \right)$        &  $\pi$ estimated by any consistent estimator    \\ \hline
        \cite{nie2021covariate}     &    $\frac{1}{n} \sum_{i \in \text{Trial}}\hat w_{n,m}(X_i) \left(\frac{ Y_i A_i}{\pi} - \frac{ Y_i (1-A_i)}{1-\pi} \right)$                   &  Oracle $\pi$        \\ \hline
        \cite{dahabreh2020extending}     &    $\frac{1}{n} \sum_{i \in \text{Trial}}\hat w_{n,m}(X_i) \left(\frac{ Y_i A_i}{\hat \pi_n(X)} - \frac{ Y_i (1-A_i)}{1-\hat \pi_n(X)} \right)$                   &  $\hat \pi_n(X)$ estimated with logistic regression         \\ \hline
         \cite{buchanan2018generalizing}   &                  $\frac{1}{n} \sum_{i \in \text{Trial}}\hat w_{n,m}(X_i) \left(\frac{ Y_i A_i}{\hat \pi_n} - \frac{ Y_i (1-A_i)}{1-\hat \pi_n} \right)$        &  $\pi$ estimated once    $\hat \pi_n = \sum_{i=1}^n A_i / n$             \\ \hline
        \end{tabular}
    \caption{\textbf{Non-exhaustive review of the different IPSW versions}, illustrating that different approaches exist.}
    \label{tab:non-exhaustive-review-IPSW}
    \end{table}    
 \end{center}
 
\end{document}